\documentclass[reqno, 11pt]{amsart}
\usepackage[utf8]{inputenc}
\usepackage{geometry}
\usepackage{amsfonts}
\usepackage{hyperref}
\usepackage[backend=biber, style=numeric-comp]{biblatex}
\hypersetup{
pdftitle={TBA},
pdfsubject={},
pdfauthor={Shijia Jin},
pdfkeywords={}
}
\usepackage{amsmath}
\usepackage{xcolor}
\usepackage{ulem}
\usepackage{amsthm}
\usepackage{pdflscape}
\usepackage{pgfplots}
\usepackage{mathrsfs}
\calclayout

\newtheorem{theorem}{Theorem}[section]

\newtheorem{proposition}[theorem]{Proposition}
\newtheorem{lemma}[theorem]{Lemma}

\newtheorem{assumption}[theorem]{Assumption}
\newtheorem{remark}[theorem]{Remark}

\theoremstyle{definition}

\newtheorem{example}[theorem]{Example}

\usepackage{relsize}

\usepackage{mlmodern}
\usepackage[T1]{fontenc}

\title{Macroscopic Market Making}

\author{Ivan Guo \hspace{1cm} Shijia Jin \hspace{1cm} Kihun Nam}

\begin{document}

\maketitle

\begin{abstract}
We propose a macroscopic market making model \`a la Avellaneda-Stoikov, using continuous processes for orders instead of discrete point processes. The model intends to bridge the gap between market making and optimal execution problems, while shedding light on the influence of order flows on the optimal strategies. We demonstrate our model through three problems. The study provides a comprehensive analysis from Markovian to non-Markovian noises and from linear to non-linear intensity functions, encompassing both bounded and unbounded coefficients. Mathematically, the contribution lies in the existence and uniqueness of the optimal control, guaranteed by the well-posedness of the strong solution to the Hamilton-Jacobi-Bellman equation and the (non-)Lipschitz forward-backward stochastic differential equation. Finally, the model's applications to price impact and optimal execution are discussed.\\


\end{abstract}

\tableofcontents
\addtocontents{toc}{\setcounter{tocdepth}{1}}

\section{Introduction}
\label{intro}
\noindent In this paper, we focus on the strategic behaviour of the market maker, which is defined as a liquidity provider in the financial market. More specifically, the market maker provides bid and ask prices on one or several assets, and makes profit by the bid–ask spread (the price difference between buy and sell orders). Pioneered by the work of \cite{ho1981optimal} and subsequently \cite{avellaneda2008high}, market making---as a stochastic control problem---has been the subject of extensive literature in market microstructure. We propose a macroscopic model \`a la Avellaneda-Stoikov \cite{avellaneda2008high} in view of a major motivation as follows.

While the market making problem focuses on the provision of liquidity, the other side of the coin studies how liquidity can be consumed in an optimal fashion. This is known as the optimal execution problem, set in motion by \cite{bertsimas1998optimal} and \cite{almgren2001optimal}. We refer to \cite{gueant2016financial} and \cite{cartea2015algorithmic} for excellent overviews of both topics. Although market making and optimal execution problems are highly correlated, they have been studied separately in two streams of literature. One of the main reasons is the difference in the underlying order characteristics: the total volume of (market) orders in the optimal execution is absolutely continuous with respect to time, allowing for the consideration of trading rates; However, in the market making problem, orders follow prescribed point processes. 

To bridge the gap between these two topics, we study a market making problem similar to the framework in \cite{avellaneda2008high}, but replace discrete point processes by continuous processes, aligning with the order rate concept prevalent in the optimal execution literature. In practical terms, we expand the time horizon of the market making from seconds, as traditionally approached, to minutes or hours (see the introduction section in \cite{guo2017optimal}). This expanded horizon enables us to safely model orders using continuous processes, resulting in what we refer to as a \textit{macroscopic} model. As mentioned in \cite{guo2017optimal}, the solution to an optimal execution problem can be used to determine the number of orders to trade over a small time scale. Similarly, a connection between our model and the microstructural market making model (e.g., \cite{abergel2022algorithmic}) is that our solution can specify the inventory process to be tracked over a small time scale.

The comparison between the macroscopic model and the traditional model is summarized as follows:
\begin{itemize}
    \item \textit{Methodology}: In prior studies, with the Markovian setup, market making models primarily utilize dynamic programming and partial differential equations (PDEs). In this paper, extending the analysis to non-Markovian settings, we employ the stochastic maximum principle and forward-backward stochastic differential equations (FBSDEs)\\
    \vspace{-0.1cm}
    
    \item \textit{Orders}: In earlier research, discrete orders are modelled using Cox processes, where the intensity is given by a deterministic function. An extension to Markov chain-based intensities is provided in \cite{campi2020optimal}.  A model with random trade sizes has been proposed in \cite{bergault2021size}. In our model, orders are modelled as rates. In a model-free way, they are represented by two adapted processes with no Markovian constraints.\\
    \vspace{-0.1cm}

    \item \textit{Execution Probability}: In the previous work, each order has a probability of being executed by the market maker. In our setting, this probability becomes a portion, meaning a specific fraction of the orders will be executed by the agent.\\
    \vspace{-0.1cm}

    \item \textit{Inventory Constraint}: Traditional models impose upper and lower constraints on inventory, leading to dimensional reduction and the use of compactness arguments. In contrast, we remove these restrictions.
    
\end{itemize}

The goal of this paper is to propose the macroscopic market making problems as stochastic control problems, and to study the associated mathematical issues---the existence and uniqueness of the optimal control. Our model shares several features with the Avellaneda-Stoikov model, hence it also inherits some of their limitations. Noted by \cite{law2019market}, the seminal model does not account for price-time priority. Moreover, for certain liquid stocks, the execution probability of limit orders placed more than one tick away from the best quote is considerably low on short time scales. Consequently, our model mainly focuses on quote-driven markets and order-driven markets where the ratio of bid-ask spread to tick size is large (see \cite{gueant2017optimal}). We direct the reader to \cite{fodra2015high}, \cite{guilbaud2013optimal}, and \cite{guilbaud2015optimal} for market making models specifically designed for limit order books.

We look at three problems, extending the analysis from Markovian to non-Markovian noise and from linear to non-linear intensity functions. This progression allows us to address scenarios with both bounded and unbounded coefficients. Moreover, the study on each problem is conducted in three steps: (1) we introduce the market making setup, providing a description of the stochastic control problem; (2) the stochastic control problem is then transformed to either a PDE or a FBSDE, depending on the specific context; (3) we finally establish the well-posedness of the equation, hence ensuring the existence and uniqueness of the optimal control.

We start with the case when noises are Markovian, and the intensity function---which is used to model the execution probability---is linear. Applying the dynamic programming principle, the resulting PDE exhibits similarities to the main PDEs studied in \cite{barger2019optimal}, \cite{fouque2022optimal} and \cite{souza2022regularized}, all of which investigate the optimal execution problems with stochastic price impact. While \cite{barger2019optimal}, \cite{fouque2022optimal} study the approximation solutions and \cite{souza2022regularized} deduces the well-posedness of viscosity solutions, our contribution lies in the existence and uniqueness of the classical solution. Besides the connection with optimal execution problem, we also examine the case of linear intensity function due to the applicability of the technique employed in proving the well-posedness, which will be extended to the next problem. We conclude the first problem by presenting a simple example with an explicit solution, demonstrating how order flows are incorporated into the market maker's strategy.

Since order flows are known to be non-Markovian (see \cite{cont2000herd} for reference), the second problem adopts non-Markovian noises and also more general non-linear intensity functions introduced in \cite{gueant2017optimal}. By utilizing a truncation function and a version of the stochastic maximum principle, we obtain a Lipschitz FBSDE, the well-posedness of which is proved via the decoupling method introduced by \cite{ma2015well}. Furthermore, the solution of the Lipschitz FBSDE is used to verify the well-posedness of a non-Lipschitz FBSDE that corresponds to the case when the truncation is not applied. 

Moving beyond bounded noises and coefficients considered in the previous two problems, the third problem explores unbounded coefficients, along with non-Markovian noises and the linear intensity function. By leveraging a convex-analytic approach, we obtain a non-Lipschitz FBSDE and establish its well-posedness by constructing the solution using a monotonic sequence.

This paper concludes with three applications of the proposed model. First, we explore the connection between the macroscopic model and the Avellaneda-Stoikov model, from both modelling and solution perspectives. The macroscopic solution emerges as the average of the Avellaneda-Stoikov solutions, with their difference diminishing as the tick size approaches $0$. Our discussion is further supported by numerical results. Next, given the well-known empirical finding that price impact is typically concave (e.g., \cite{toth2011anomalous}), we explicitly derive this concavity in two simple cases of our model, and demonstrate it through simulations in more general settings. Finally, we establish a link between our market making model and the optimal execution problem, where the permanent price impact function is replaced by the market maker’s quoting strategy. We apply several simple execution strategies and numerically evaluate their performances.

The organization of the article aligns with the three problems outlined in preceding paragraphs: Section \ref{section_2} delves into the first problem, Section \ref{section_3} explores the second problem, and Section \ref{section_4} studies the third problem. Applications are presented in Section \ref{implement}. Appendix \ref{section_5} encompasses a version of the stochastic maximum principle.

\textit{Notation:} Throughout the present work, we fix $T > 0$ to represent our finite trading horizon. We denote by $(\Omega, \mathcal{F}, \mathbb{F}=(\mathcal{F}_t)_{0\leq t\leq T}, \mathbb{P})$ a complete filtered probability space, with $\mathcal{F}_T = \mathcal{F}$. An $m$-dimensional Brownian motion $W=(W^1,\dots,W^m)$ is defined on such space, for a fixed positive integer $m$, and the filtration $\mathbb{F}$ is generated by $W$ and augmented. Let $\mathcal{G}$ represents an arbitrary $\sigma$-algebra contained in $\mathcal{F}$ and consider the following spaces:
\begin{gather*}
    L^p(\Omega, \mathcal{G}):=\big\{X: X \text{ is } \mathcal{G} \text{-measurable and } \mathbb{E}|X|^p<\infty \big\};\\
    \mathbb{H}^p:=\Big\{X: X \text{ is } \mathbb{F} \text{-progressively measurable and } \mathbb{E}\Big[\big(\int_0^T |X_t|^2\,dt\big)^{p/2}\Big]<\infty\, \Big\};\\
    \mathbb{S}^p:=\Big\{X\in \mathbb{H}^p: \mathbb{E}\big[\sup_{0\leq t\leq T} |X_t|^p\big]<\infty \Big\};\\
    \mathcal{M}:= \big\{ M : M_t \in L^2(\Omega, \mathcal{F}_t) \text{ for a.e. } t \in [0, T] \text{ and } \{M_t, \mathcal{F}_t\}_{0\leq t\leq T}\text{ is a martingale} \big\}.
\end{gather*}

\section{Markovian Order Flows with Linear Intensity Function}  \label{section_2}

\noindent This section is devoted to introducing the macroscopic market making as a stochastic control problem. Motivated by linear price impact functions, we will begin with the linear intensity function and Markovian order flows, solving the problem by the dynamic programming principle. The reason is twofold: it reveals several interesting connections with optimal execution problems under stochastic price impact; the same technique used in the problem will be also applied in the next section, where we will look at the non-Markovian context.

Let $S=(S_t)_{0\leq t\leq T}$ be the mid-price of the asset being continuously traded throughout the time horizon $T$. We follow the framework in \cite{becherer2005classical} and consider the following stochastic differential equation for $L:=(L_t)_{t\in[0,T]}$ in $\mathbb{R}^d$:
\begin{equation}
\begin{aligned}
    L_0&=l_0\in\mathbb{R}^d,\\
    dL_t&=\Gamma(t,L_t)\,dt+\sum_{j=1}^m\Sigma_j(t,L_t)\,dW_t^j, \hspace{1cm} t\in[0,T],
    \label{markov_sde}
\end{aligned}
\end{equation}
for continuous functions $\Gamma : [0, T] \times \mathbb{R}^d \to \mathbb{R}^d$ and $\Sigma_j : [0, T] \times \mathbb{R}^d \to \mathbb{R}^d$, $j =1, \dots, m$. Define the matrix-valued function
$\Sigma : [0, T]\times \mathbb{R}^d \to \mathbb{R}^{d\times m}$ by $\Sigma^{i, j}
:= (\Sigma_j)^i$. The process $L$ serves as the driving process that determines the dynamics of other key variables in our model. It enables us to model interdependencies of variables in a coherent framework. We introduce the following assumptions.

\begin{assumption}
(1) The mid-price process $S$ belongs to $\mathcal{M}$. Hence, for any $K\in\mathbb{H}^2$, it holds 
\begin{equation}
    \mathbb{E}\int_0^T K_t\,dS_t=0.
    \label{martingale_price}
\end{equation}
(2) For $G\in\{\Gamma, \Sigma_1, \dots, \Sigma_m\}$, there exists a constant $C_L>0$ such that
\begin{equation*}
    |G(t,x)-G(t,y)|\leq C_L |x-y|,
\end{equation*}
for all $t\in[0,T]$ and $x,y\in\mathbb{R}^d$. Further, $\det(\Sigma(t,l)\,\Sigma^{tr}(t,l))\neq 0$ for all $(t,l)\in [0,T]\times \mathbb{R}^d$.
\label{sde_assumption}
\end{assumption}

\begin{remark}
The first assumption is common in optimal execution and market making literatures, where the mid-price $S$ usually has no effect on the optimization problem. In this case, the emphasis is placed on how the market maker deals with the order flows. The last condition on $\Sigma$ is added for technical reasons.
\end{remark}

A market maker keeps trading such an asset to profit from the bid-ask spread by posting bid and ask quotes during the time interval $[0, T]$. More precisely, she posts the ask orders and bid orders at the price levels
\begin{equation}
    S_t^a:=S_t+\delta_t^a \text{\quad and \quad} S_t^b:=S_t-\delta_t^b
    \label{admiss_control}
\end{equation}
respectively, where $\boldsymbol{\delta}:=(\delta^a, \delta^b)\in\mathbb{H}^2\times\mathbb{H}^2$ represents the control. While market orders are described by point processes in the discrete setting, here we regard them as \textit{flows}---continuous processes adapted to the Brownian filtration $\mathbb{F}$. The market order flows at the ask and bid side are modelled by
\begin{equation}
    a_t:=a(L_t) \text{\quad and \quad} b_t:=b(L_t)
    \label{admiss_control}
\end{equation}
accordingly, for some bounded positive functions $a, b:\mathbb{R}^d\to\mathbb{R}_+$. Denote by $\bar{a}, \bar{b}>0$ the upper bounds for functions $a$ and $b$ accordingly.
On a macroscopic scale, the execution probability in the original Avellaneda-Stoikov framework is now interpreted as the portion of order flows captured by the agent. More specifically, suppose that execution probability is $p\in(0,1)$ in the discrete setting, implying the agent has a probability of $p$ to execute each market order. Then, given a sufficient amount of orders, one expects that $p$ portion of total orders is matched by the agent \textit{on average}. Since such probability is determined by the \textit{intensity function} $\Lambda$, we first look at a linear function
\begin{equation}
    \Lambda(\delta)=\zeta-\gamma\,\delta,
    \label{linear_intensity}
\end{equation}
where $\zeta, \gamma>0$ represents the competition level among market makers and we also assume the bid-ask symmetry. Intuitively, as the agent posts further away from the mid-price, the chance of executions decreases in a linear fashion. Since the function $\Lambda$ represents a physically non-negative portion, the action $\delta$ makes the most sense if $\delta \leq \zeta / \gamma$. Here, we do not constrain the action to this region for the sake of linearity and, therefore, treat the function values in $(\zeta / \gamma, \infty)$ as modelling errors. However, if we impose such a constraint, the problem can be resolved in the next section with a minor modification. Consequently, the inventory and cash processes of the market maker now follow
\begin{gather*}
        Q_t = q_0-\int_0^t a_u\cdot(\zeta-\gamma \delta^a_u)\,du+\int_0^t b_u\cdot(\zeta-\gamma \delta^b_u)\,du,\\
        \nonumber
    X_t = x_0+\int_0^t a_u\cdot(\zeta-\gamma \delta^a_u)\,(S_u+\delta_u^a)\,du-\int_0^t b_u\cdot(\zeta-\gamma \delta^b_u)\,(S_u-\delta_u^b)\,du.
    \nonumber
\end{gather*}
The term $a_u(\zeta-\gamma \delta^a_u)$ shows the (passive) selling rate of the agent, and $S_u+\delta_u^a$ represents the selling price at time $u$. The agent then aims at maximizing the following objective functional 
\begin{equation}
    \mathbb{E}\big[X_T+S_T\,Q_T-\int_0^T \phi_t\,(Q_t)^2\,dt-A\,(Q_T)^2 \big]
    \label{obj_functional}
\end{equation}
by controlling $\boldsymbol{\delta}\in\mathbb{H}^2\times\mathbb{H}^2$, where $A:=A(L_T)$ and $\phi_t:=\phi(L_t)$ for some non-negative bounded functions $A, \phi:\mathbb{R}^d\to\mathbb{R}_+$ with the upper bounds $\bar{A}, \bar{\phi}>0$, accordingly. The functional \eqref{obj_functional} is introduced in \cite{cartea2014buy}.
 
\begin{assumption}
Function A is a continuous non-negative function with the upper bound $\bar{A}$. Functions $a, b$ are positive functions bounded above by $\bar{a}, \bar{b}$ accordingly; function $\phi$ is a non-negative function bounded above by $\bar{\phi}$. Moreover, functions $a, b, \phi$ are all Lipschitz continuous in any bounded domain of $\mathbb{R}^d$.
\label{lip_ord_flow}
\end{assumption}

\begin{remark}
In the functional \eqref{obj_functional}, the random variable $A$ stands for the terminal penalty. It can be interpreted as the overnight risk as in \cite{adrian2020intraday} or the terminal clearance cost as in the optimal execution literatures. The integral term is the well-known reduced-form model of risk and process $\phi$ could be understood as the product of a risk parameter and the stochastic volatility of the mid-price process (see \cite{souza2022regularized} for an example).
\end{remark}

By It\^o’s formula, we compute that
\begin{equation}
    \begin{aligned}
    X_T+Q_TS_T&=x_0+\int_0^T(S_t+\delta_t^a)\,a_t(\zeta-\gamma\delta_t^a)\,dt-\int_0^T(S_t-\delta_t^b)\,b_t(\zeta-\gamma\delta_t^b)\,dt\\
    &\hspace{1.45cm} +q_0s_0+\int_0^TS_t\,b_t(\zeta-\gamma\delta_t^b)\,dt-\int_0^TS_t\,a_t(\zeta-\gamma\delta_t^a)\,dt+\int_0^TQ_t\,dS_t\\
    &=x_0+q_0s_0+\int_0^T\delta_t^a\,a_t(\zeta-\gamma\delta_t^a)\,dt+\int_0^T\delta_t^b\,b_t(\zeta-\gamma\delta_t^b)\,dt+\int_0^TQ_t\,dS_t.
    \nonumber
    \end{aligned}
\end{equation}
While the last term will vanish after taking the expectation due to \eqref{martingale_price}, we can then rewrite the functional \eqref{obj_functional} as
\begin{equation}
    \mathbb{E}\big[\int_0^T\delta_t^a\,a_t(\zeta-\gamma\delta_t^a)\,dt+\int_0^T\delta_t^b\,b_t(\zeta-\gamma\delta_t^b)\,dt-\int_0^T \phi_t\,(Q_t)^2\,dt-A\,(Q_T)^2 \big],
    \label{equ_obj_functional}
\end{equation}
where the mid-price and cash no longer play a role. Note that we do not impose constraints on the action space. Moreover, previous studies assume only a lower bound on $(\delta_t^a, \delta_t^b)$ without explicitly specifying it. This may result in $S_t^a < S_t^b$, which is not feasible in practice. The following lemma demonstrates that such an action is always dominated by some other strategy, of which the associated ask price is no smaller than the bid price.

\begin{lemma}
    If a strategy $\boldsymbol{\delta}$ results in $S_t^a < S_t^b$ for some $t$, then it is dominated by some other strategy such that the associated ask price is no smaller than the bid price.
\end{lemma}

\begin{proof}
    Suppose a strategy $\boldsymbol{\delta}$ leads to $S_t^a < S_t^b$ (or equivalently $\delta_t^a + \delta_t^b < 0$) for some $t$. We show that there exists a better action at time $t$. If $b_t = 0$, we simply replace $\delta_t^b$ with $\tilde{\delta}_t^b = -\delta_t^a$. After that, the infinitesimal value of first two terms in \eqref{equ_obj_functional} and the inventory drift stays the same, while the ask price is equal to the new bid price.

    If $b_t > 0$, let us define a new action $(\tilde{\delta}_t^a, \tilde{\delta}_t^b)$ by
    \begin{equation*}
        y = -\frac{\delta_t^a + \delta_t^b{}}{1 + \frac{a_t}{b_t}}, \qquad \tilde{\delta}_t^a = \delta_t^a + y, \text{\quad and \quad} \tilde{\delta}_t^b = \delta_t^b + y \, \frac{a_t}{b_t}.
    \end{equation*}
    The new action satisfies $\tilde{\delta}_t^a + \tilde{\delta}_t^b = 0$, hence the associated ask price equals to the bid price. Direct calculations show that both $(\delta_t^a, \delta_t^b)$ and $(\tilde{\delta}_t^a, \tilde{\delta}_t^b)$ lead to the same inventory drift. It suffices to check the infinitesimal value of first two terms in \eqref{equ_obj_functional}, which can be deduced by
    \begin{equation*}
    \begin{aligned}
         \tilde{\delta}_t^a \, a_t (\zeta - \gamma \tilde{\delta}_t^a) &+ \tilde{\delta}_t^b \, b_t(\zeta - \gamma \tilde{\delta}_t^a)- \delta_t^a \, a_t (\zeta - \gamma\delta_t^a) - \delta_t^b \, b_t(\zeta-\gamma\delta_t^a)\\
         &= y \,a_t \big( \zeta - 2\gamma \delta_t^a - \gamma y \big) + y \, a_t \big(\zeta - 2\gamma \delta_t^b - \gamma y \, \frac{a_t}{b_t} \big)\\
         &= y \,a_t \Big[ \zeta - 2\gamma (\delta_t^a + \delta_t^b) - \gamma y \big(1 + \frac{a_t}{b_t} \big) \Big]\\
         &= y \,a_t \big[ \zeta - \gamma (\delta_t^a + \delta_t^b) \big] \geq 0.
    \end{aligned}
    \end{equation*}
    
\end{proof}

Define
\begin{equation}
\begin{aligned}
    H(t, l, q):=\sup_{\delta^a,\, \delta^b\in \mathbb{H}^2}\mathbb{E}\big[\int_t^T\delta_u^a\,a_u^{t, l}(\zeta-\gamma\delta_u^a)\,&du+\int_t^T\delta_u^b\,b_u^{t, l}(\zeta-\gamma\delta_u^b)\,du\\
    &-\int_t^T \phi_u^{t, l}\,(Q_u^{t, l, q, \boldsymbol{\delta}})^2\,du-A\,(Q_T^{t, l, q, \boldsymbol{\delta}})^2 \big].
\end{aligned}
\end{equation}
We use the superscript $(t, l)$ to show the initial condition of the underlying Markov diffusions, and $(t, l, q, \boldsymbol{\delta})$ to display the additional dependence on the control. Let $\mathcal{L}$ be the generator of $L$ as
\begin{equation*}
    (\mathcal{L}f)(t,l,q):=\sum_{i=1}^d\Gamma^i(t,l)\,\frac{\partial f}{\partial l^i}(t,l,q)+\frac{1}{2}\sum_{i,j=1}^d (\Sigma\Sigma^{tr})^{i,j}(t,l)\,\frac{\partial^2 f}{\partial l^i\partial l^j}(t,l,q)
\end{equation*}
for some sufficiently smooth function $f:[0,T]\times\mathbb{R}^d\times\mathbb{R}\to\mathbb{R}$, then the dynamic programming principle suggests the following partial differential equation (PDE)
\begin{equation}
\begin{aligned}
    \frac{\partial H}{\partial t}+\mathcal{L}H(t,l,q)&+\sup_{\delta^a\in\mathbb{R}}\big\{(\delta^a-\frac{\partial H}{\partial q})\,a(l)(\zeta-\gamma\delta^a)\big\}\\
    &+\sup_{\delta^b\in\mathbb{R}}\big\{(\delta^b+\frac{\partial H}{\partial q})\,b(l)(\zeta-\gamma\delta^b)\big\}-\phi(l)\,(q)^2=0
    \label{pde}
\end{aligned}
\end{equation}
for $(t,l,q)\in[0,T]\times\mathbb{R}^d\times\mathbb{R}$, subjecting to the boundary condition $H(T,l,q)=-A(l)\,( q)^2$. The optimal control in feedback form reads
\begin{equation}
    \delta_t^{b,*}=\frac{\zeta}{2\gamma}-\frac{1}{2}\frac{\partial H}{\partial q}, \hspace{1cm} \delta_t^{a,*}=\frac{\zeta}{2\gamma}+\frac{1}{2}\frac{\partial H}{\partial q}.
    \nonumber
\end{equation}
Combined with the affine ansatz widely used in optimal execution literatures (thanks to the linear structure of our problem):
\begin{equation}
    H(t,l,q)=h_0(t,l)+h_1(t,l)\,q+h_2(t,l)\,(q)^2,
    \nonumber
\end{equation}
we plug them back to \eqref{pde} to obtain, for $(t,l)\in[0,T]\times\mathbb{R}^d$, that
\begin{equation}
    \frac{\partial h_2}{\partial t}+\mathcal{L}h_2(t,l)+\gamma\big(a(l)+b(l)\big)\,h_2(t,l)^2-\phi(l)=0,
    \label{non_lin_pde}
\end{equation}
\begin{equation}
    \frac{\partial h_1}{\partial t}+\mathcal{L}h_1(t,l)+\zeta\big(b(l)-a(l)\big)\,h_2(t,l)+\gamma\big(a(l)+b(l)\big)\,h_2(t,l)\,h_1(t,l)=0,
    \label{lin_pde_1}
\end{equation}
\begin{equation}
    \frac{\partial h_0}{\partial t}+\mathcal{L}h_0(t,l)+\zeta\big(b(l)-a(l)\big)\,h_1(t,l)+\frac{\gamma\, b(l)}{4}\big(\frac{\zeta}{\gamma}-h_1(t,l)\big)^2+\frac{\gamma\, a(l)}{4}\big(\frac{\zeta}{\gamma}+h_1(t,l)\big)^2=0,
    \label{lin_pde_2}
\end{equation}
subjecting to the boundaries $h_2(T,l)=-A(l)$,  $h_1(T,l)=0$ and $h_0(T,l)=0$. One can see that the main challenge lies in the nonlinear PDE \eqref{non_lin_pde}, while the others are all linear and the well-posedness is direct given $h_2$ with sufficient regularity. We summarize this result in the next theorem with the verification step. Hence, the remaining of this section is devoted to \eqref{non_lin_pde}.

\begin{remark}
We compare our PDE system with the ones from related literatures:\\
\indent (1) Regarding the optimal execution problems that take the order flow into account (see \cite{cartea2015algorithmic} section 7.3), in the PDE for $h_2$, the quadratic term of $h_2$ has a constant coefficient, that is,
    \begin{equation}
        \frac{\partial h_2}{\partial t}+\mathcal{L}h_2(t,l)-\phi+\frac{1}{k}\,h_2(t, l)^2=0,
        \nonumber
    \end{equation}
where $k, \phi>0$ represent the permanent price impact and risk parameter accordingly. As there are no source terms $a$ and $b$, they deduce that $h_2$ is independent of the order flow and solve a Riccati equation. However, there is a source term of order flows in our case.

(2) With respect to execution problems with stochastic  parameters (\cite{barger2019optimal}, \cite{fouque2022optimal} and \cite{souza2022regularized}), they all consider PDEs for $h_2$ that are similar to ours. While \cite{barger2019optimal} and \cite{fouque2022optimal} study the approximation solutions, the existence and uniqueness of the viscosity solution is proven by \cite{souza2022regularized}. Our contribution lies in the well-posedness of the classical solution. Let us look at the equation from \cite{barger2019optimal} (equation (2.17) with no permanent price impact):
    \begin{equation}
        \frac{\partial h_2}{\partial t}+\mathcal{L}h_2(t,l)-\varphi+\frac{1}{f(l)}\,h_2(t,l)^2=0,
        \nonumber
    \end{equation}
where $\varphi>0$ stands for the risk and $f(l)$ is the stochastic impact parameter. Comparing $1/f(l)$ with $\gamma(a(l)+b(l))$ of our equation, it is interesting to see that the lower friction in optimal executions is analogous to the higher market activity in market makings. Furthermore, our method can be applied to prove the well-posedness of (2.17) by a change of variable.
\end{remark}

\begin{theorem}
\label{verification}
Suppose \eqref{non_lin_pde} admits a bounded classical solution denoted by $h_2$, it holds: 

(1) Both \eqref{lin_pde_1} and \eqref{lin_pde_2} also possess bounded classical solutions, written as $h_1$ and $h_0$ accordingly;

(2) Define $\tilde{H}:[0,T]\times\mathbb{R}^d\times\mathbb{R}\to\mathbb{R}$ as
\begin{equation*}
    \tilde{H}(t,l,q)=h_0(t,l)+h_1(t,l)\,q+h_2(t,l)\,(q)^2,
\end{equation*}
then $H=\tilde{H}$ on $[0,T]\times\mathbb{R}^d\times\mathbb{R}$;

(3) Let $\boldsymbol{\delta}^*=(\delta^{a,*}, \delta^{b,*})$ be defined by
\begin{equation}
    \delta_t^{b,*}=\frac{\zeta}{2\gamma}-\frac{1}{2}\frac{\partial \tilde{H}}{\partial q} \text{ \; and \; } \delta_t^{a,*}=\frac{\zeta}{2\gamma}+\frac{1}{2}\frac{\partial \tilde{H}}{\partial q},
    \label{uniq_Mark_cont}
\end{equation}
then $\boldsymbol{\delta}^*$ is the unique optimal control.
\end{theorem}

\begin{proof}
(1) The fact that both \eqref{lin_pde_1} and \eqref{lin_pde_2} accepts classical solutions follows from \cite{becherer2005classical}. It suffices to apply the Feynman–Kac formula to derive the boundedness.

(2) It is straightforward to see that $\tilde{H}$ is sufficiently regular to apply the It\^o's formula and satisfies the quadratic growth condition in $q$:
\begin{equation*}
    |\tilde{H}(t,l,q)|\leq C\,(1+q^2),
\end{equation*}
for a constant $C>0$. We follow the arguments in \cite{pham2009continuous}. First, we will show that $\tilde{H}\geq H$. For all $(t,l,q)\in[0,T)\times\mathbb{R}^d\times\mathbb{R}$, $\boldsymbol{\alpha}\in\mathbb{H}^2\times \mathbb{H}^2$, and any stopping time $\tau$ valued in $[t,\infty)$, it holds that
\begin{equation*}
\begin{aligned}
    \tilde{H}(s\wedge\tau, L_{s\wedge\tau}^{t,l}, Q_{s\wedge\tau}^{t,l,q,\boldsymbol{\alpha}})=\tilde{H}&(t,l,q)+\int_t^{s\wedge\tau}\frac{\partial \tilde{H}}{\partial l}(u, L_{u}^{t,l}, Q_{u}^{t,l,q,\boldsymbol{\alpha}})^{tr}\,\Sigma(u, L_{u}^{t,l})\,dW_u\\
    &+\int_t^{s\wedge\tau}\big[\frac{\partial \tilde{H}}{\partial t}(u, L_{u}^{t,l}, Q_{u}^{t,l,q,\boldsymbol{\alpha}})+\mathcal{L}\tilde{H}(u, L_{u}^{t,l}, Q_{u}^{t,l,q,\boldsymbol{\alpha}})\\
    &\hspace{1.8cm}-\frac{\partial \tilde{H}}{\partial q}\,a(L_{u}^{t,l})\,(\zeta-\gamma\,\alpha_u^a)\\
    &\hspace{2.8cm}+\frac{\partial \tilde{H}}{\partial q}\,b(L_{u}^{t,l})\,(\zeta-\gamma\,\alpha_u^b)\big]\,du.
\end{aligned}
\end{equation*}
If we choose
\begin{equation*}
    \tau=\tau_n=\inf\Big\{r\geq t:\; \int_t^{r}\big|\frac{\partial \tilde{H}}{\partial l}(u, L_{u}^{t,l}, Q_{u}^{t,l,q,\boldsymbol{\alpha}})^{tr}\,\Sigma(u, L_{u}^{t,l})\big|^2\,du\geq n \Big\},
\end{equation*}
the stopped stochastic integral is then a martingale and hence

\begin{equation*}
\begin{aligned}
    &\mathbb{E}\big[\tilde{H}(s\wedge\tau_n, L_{s\wedge\tau_n}^{t,l}, Q_{s\wedge\tau_n}^{t,l,q,\boldsymbol{\alpha}})\big]\\
    &=\tilde{H}(t,l,q)+\mathbb{E}\int_t^{s\wedge\tau_n}\big[\frac{\partial \tilde{H}}{\partial t}(u, L_{u}^{t,l}, Q_{u}^{t,l,q,\boldsymbol{\alpha}})+\mathcal{L}\tilde{H}(u, L_{u}^{t,l}, Q_{u}^{t,l,q,\boldsymbol{\alpha}})\\
    &\hspace{2cm}-\frac{\partial \tilde{H}}{\partial q}\,a(L_{u}^{t,l})\,(\zeta-\gamma\,\alpha_u^a)+\frac{\partial \tilde{H}}{\partial q}\,b(L_{u}^{t,l})\,(\zeta-\gamma\,\alpha_u^b)\big]\,du\\
    &\leq \tilde{H}(t,l,q)+\mathbb{E}\int_t^{s\wedge\tau_n}\big[\frac{\partial \tilde{H}}{\partial t}(u, L_{u}^{t,l}, Q_{u}^{t,l,q,\boldsymbol{\alpha}})+\mathcal{L}\tilde{H}(u, L_{u}^{t,l}, Q_{u}^{t,l,q,\boldsymbol{\alpha}})\\
    &\hspace{2cm}+\sup_{\delta^a\in\mathbb{R}}\big\{(\delta^a-\frac{\partial \tilde{H}}{\partial q})\,a(L_{u}^{t,l})\,(\zeta-\gamma\,\delta^a)\big\}\\
    &\hspace{2cm}+\sup_{\delta^b\in\mathbb{R}}\big\{(\delta^b+\frac{\partial \tilde{H}}{\partial q})\,b(L_{u}^{t,l})\,(\zeta-\gamma\,\delta^b)\big\}\\
    &\hspace{2cm}-\alpha_u^a\,a(L_{u}^{t,l})\,(\zeta-\gamma\,\alpha_u^a)-\alpha_u^b\,b(L_{u}^{t,l})\,(\zeta-\gamma\,\alpha_u^b)\big]\,du\\
    &= \tilde{H}(t,l,q)+\mathbb{E}\int_t^{s\wedge \tau_n}\Big[\phi(L_{u}^{t,l})\,(Q_{u}^{t,l,q,\boldsymbol{\alpha}})^2\\
    &\hspace{2cm}-\alpha_u^a\,a(L_{u}^{t,l})\,(\zeta-\gamma\,\alpha_u^a)-\alpha_u^b\,b(L_{u}^{t,l})\,(\zeta-\gamma\,\alpha_u^b)\big]\,du.
\end{aligned}
\end{equation*}
Note that
\begin{equation*}
\begin{aligned}
    &\big|\int_t^{s\wedge \tau_n}\Big[\phi(L_{u}^{t,l})\,(Q_{u}^{t,l,q,\boldsymbol{\alpha}})^2-\alpha_u^a\,a(L_{u}^{t,l})\,(\zeta-\gamma\,\alpha_u^a)-\alpha_u^b\,b(L_{u}^{t,l})\,(\zeta-\gamma\,\alpha_u^b)\big]\,du\,\big|\\
    &\hspace{2cm}\leq \int_t^{T}\big|\phi(L_{u}^{t,l})\,(Q_{u}^{t,l,q,\boldsymbol{\alpha}})^2-\alpha_u^a\,a(L_{u}^{t,l})\,(\zeta-\gamma\,\alpha_u^a)\\
    &\hspace{3cm}-\alpha_u^b\,b(L_{u}^{t,l})\,(\zeta-\gamma\,\alpha_u^b)\big|\,du
\end{aligned}
\end{equation*}
and
\begin{equation*}
    \tilde{H}(s\wedge\tau, L_{s\wedge\tau_n}^{t,l}, Q_{s\wedge\tau_n}^{t,l,q,\boldsymbol{\alpha}})\leq C\,(1+\sup_{u\in[t,T]}|Q_{u}^{t,l,q,\boldsymbol{\alpha}}|^2).
\end{equation*}
The right hand sides of both above terms are integrable and we apply the dominated convergence theorem twice to see 
\begin{equation*}
\begin{aligned}
    \mathbb{E}\big[\tilde{H}(s, &L_{s}^{t,l}, Q_{s}^{t,l,q,\boldsymbol{\alpha}})\big]\leq \tilde{H}(t,l,q)\\
    &+\mathbb{E}\int_t^{s}\Big[\phi(L_{u}^{t,l})\,(Q_{u}^{t,l,q,\boldsymbol{\alpha}})^2-\alpha_u^a\,a(L_{u}^{t,l})\,(\zeta-\gamma\,\alpha_u^a)-\alpha_u^b\,b(L_{u}^{t,l})\,(\zeta-\gamma\,\alpha_u^b)\big]\,du,
\end{aligned}
\end{equation*}
and further
\begin{equation*}
\begin{aligned}
    \mathbb{E}\big[ -A(L_{T}^{t,l})\,(Q_{T}^{t,l,q,\boldsymbol{\alpha}})^2 \big] &\leq \tilde{H}(t,l,q)\\
    &+\mathbb{E}\int_t^{T}\Big[\phi(L_{u}^{t,l})\,(Q_{u}^{t,l,q,\boldsymbol{\alpha}})^2-\alpha_u^a\,a(L_{u}^{t,l})\,(\zeta-\gamma\,\alpha_u^a)\\
    &\hspace{3.5cm}-\alpha_u^b\,b(L_{u}^{t,l})\,(\zeta-\gamma\,\alpha_u^b)\big]\,du.
\end{aligned}
\end{equation*}
Therefore, the arbitrariness of $\boldsymbol{\alpha}$ implies
\begin{equation*}
    H(t,l,q)\leq \tilde{H}(t,l,q).
\end{equation*}
To see the reversed inequality, let $Q^{t,l,q,\boldsymbol{\delta}^*}$ be the associated inventory of the control $\boldsymbol{\delta}$ defined in \eqref{uniq_Mark_cont} and apply the It\^o's formula to $\tilde{H}(u, L_{u}^{t,l},  Q_{u}^{t,l,q,\boldsymbol{\delta}^*})$ between $t\in[0,T)$ and $s\in[t,T)$. Through similar localizing and limiting procedures, one can obtain
\begin{equation*}
\begin{aligned}
    -A(L_{T}^{t,l})\,(Q_{T}^{t,l,q,\boldsymbol{\delta}^*})^2&=\tilde{H}(t,l,q)+\mathbb{E}\int_t^{T}\big[\frac{\partial \tilde{H}}{\partial t}(u, L_{u}^{t,l}, Q_{u}^{t,l,q,\boldsymbol{\delta}^*})+\mathcal{L}\tilde{H}(u, L_{u}^{t,l}, Q_{u}^{t,l,q,\boldsymbol{\delta}^*})\\
    &\hspace{3cm}-\frac{\partial \tilde{H}}{\partial q}\,a(L_{u}^{t,l})\,(\zeta-\gamma\,\delta_u^{a,*})\\
    &\hspace{3.8cm}+\frac{\partial \tilde{H}}{\partial q}\,b(L_{u}^{t,l})\,(\zeta-\gamma\,\delta_u^{b,*})\big]\,du\\
    &= \tilde{H}(t,l,q)+\mathbb{E}\int_t^{T}\Big[\phi(L_{u}^{t,l})\,(Q_{u}^{t,l,q,\boldsymbol{\delta}^*})^2\\
    &\hspace{2cm}-\alpha_u^a\,a(L_{u}^{t,l})\,(\zeta-\gamma\,\delta_u^{a,*})-\alpha_u^b\,b(L_{u}^{t,l})\,(\zeta-\gamma\,\delta_u^{b,*})\big]\,du.
\end{aligned}
\end{equation*}
To finally conclude that
    \begin{equation*}
    H(t,l,q)= \tilde{H}(t,l,q),
\end{equation*}
it suffices to show that  $\boldsymbol{\delta}^*\in\mathbb{H}^2$. Since the inventory satisfies the linear random ODE
\begin{equation*}
\begin{aligned}
    dQ_{u}^{t,l,q,\boldsymbol{\delta}^*}&=-a_u\,(\zeta-\gamma\,\delta_u^{a,*})\,du+b_u\,(\zeta-\gamma\,\delta_u^{b,*})\,du\\
    &=\frac{1}{2}\,[\zeta\,(b_u-a_u)+\gamma\,h_1(u,l)\,(a_u-b_u)]\,du+\gamma\,(a_u+b_u)\,h_2(u,l)\,Q_{u}^{t,l,q,\boldsymbol{\delta}^*}\,du,
\end{aligned}
\end{equation*}
the process $Q^{t,l,q,\boldsymbol{\delta}^*}$ turns out to be bounded due to the boundedness of $a, b, h_1,$ and $h_2$. Consequently, the control $\boldsymbol{\delta}^*$ is also bounded and thus admissible.

(3) It suffices to prove the uniqueness of the optimal control. The statement is true if the objective functional \eqref{equ_obj_functional} is strictly concave. The strict concavity of \eqref{equ_obj_functional} can be deduced from the strict concavity of $\delta_u^a\,a_u(\zeta-\gamma\delta_u^a)+\delta_u^b\,b_u(\zeta-\gamma\delta_u^b)$ and the convexity of $Q_t^2$ with respect to the control. The strict concavity of $\delta_u^a\,a_u(\zeta-\gamma\delta_u^a)+\delta_u^b\,b_u(\zeta-\gamma\delta_u^b)$ is clear because they are second order polynomials with negative signs for the leading terms. Let $\boldsymbol{\delta}, \boldsymbol{\beta}\in\mathbb{H}^2\times\mathbb{H}^2$ be two generic controls, and set $\lambda\in[0,1]$. We check the convexity of the quadratic inventory through
\begin{equation*}
\begin{aligned}
    (Q_u^{\lambda\,\boldsymbol{\delta}+(1-\lambda)\,\boldsymbol{\beta}})^2&-\lambda\,(Q_u^{\boldsymbol{\delta}})^2-(1-\lambda)\,(Q_u^{\boldsymbol{\beta}})^2\\
    &=\big(\lambda\,Q_u^{\boldsymbol{\delta}}+(1-\lambda)\,Q_u^{\boldsymbol{\beta}}\big)^2-\lambda\,(Q_u^{\boldsymbol{\delta}})^2-(1-\lambda)\,(Q_u^{\boldsymbol{\beta}})^2\\
    &\leq0,
\end{aligned}
\end{equation*}
where we have applied the linearity of the inventory with respect to the control. 
\end{proof}

We intend to prove the existence and uniqueness of the classical solution to \eqref{non_lin_pde} on the basis of \cite{becherer2005classical}, the result of which cannot be applied directly because of the quadratic term $h_2(t,l)^2$.  To begin with, we introduce two auxiliary independent linear PDEs
\begin{gather*}
    \frac{\partial h}{\partial t}+\mathcal{L}h(t,l)+\gamma\big(a(l)+b(l)\big)\,g(t,l)\,h(t,l)-\phi(l)=0, \quad \text{such that}\;\; h(T,l)=-A; \label{pde_1}\\
    \frac{\partial h}{\partial t}+\mathcal{L}h(t,l)+\gamma\big(a(l)+b(l)\big)\,\big[\psi\big(g(t,l)\big)\big]^2-\phi(l)=0, \quad \text{such that}\;\; h(T,l)=-A,
\end{gather*}
where $\psi:\mathbb{R}\to\mathbb{R}$ is a truncation function with the parameter $\xi>0$ defined by
\begin{equation*}
    \psi(x)=
    \begin{cases}
    -\xi, \quad &\text{if } x<-\xi\\
    x, \quad &\text{if } -\xi\leq x\leq\xi\\
    \xi, \quad &\text{if } x>\xi,
    \end{cases}
\end{equation*}
and $g$ is a bounded continuous function to be specified later. If the function $g$ and the solution of the linear PDE are considered as input and output accordingly, we can define two mappings $F_1$ and $F_2$ for the two equations by the Feynman-Kac formula:
\begin{gather*}
    (F_1\,g)(t,l)=\mathbb{E}\Big[-A\,e^{\int_t^T \gamma (a_s+b_s)\,g(s, L_s^{t,l})\,ds}-\int_t^T \phi_s\,e^{\int_t^s \gamma (a_s+b_s)\,g(u, L_u^{t,l})\,du}\,ds\Big],\\
    (F_2\,g)(t,l)=\mathbb{E}\Big[-A+\int_t^T \gamma (a_s+b_s)\,\big(\psi(g(s,L_s^{t,l}))\big)^2\,ds-\int_t^T \phi_s\,ds\Big].
\end{gather*}

\noindent It turns out that $F_1$ and $F_2$ are contraction mappings in the space of bounded functions.

\begin{lemma}
Denote by $C_b^-$ the set of bounded, non-positive and continuous functions on $[0,T]\times \mathbb{R}^d$, then the following statements hold:\\
\indent (1) $F_1$ defines a contraction mapping on $C_b^-$ with respect to the norm
\begin{equation}
    \|v\|_{\zeta_1}:=\sup_{(t,x)\in[0,T]\times \mathbb{R}^d}e^{-\zeta_1\,(T-t)}\cdot|v(t,x)|,
    \nonumber
\end{equation}
for some $\zeta_1\in\mathbb{R}_+$ large enough.\\
\indent (2) $F_2$ defines a contraction mapping on $C_b$ with respect to the norm
\begin{equation}
    \|v\|_{\zeta_2}:=\sup_{(t,x)\in[0,T]\times \mathbb{R}^d}e^{-\zeta_2\,(T-t)}\cdot|v(t,x)|,
    \nonumber
\end{equation}
for some $\zeta_2\in\mathbb{R}_+$ large enough.
\end{lemma}

\begin{proof}
We first introduce an inequality that will be used for both statements. For $\zeta_1>0$ and some bounded functions $v$ and $w$, one can see
\begin{equation}
    \begin{aligned}
        \int_t^T |v(s, L_s)-w(s, L_s)|\,ds &= \int_t^T |v(s, L_s)-w(s, L_s)| e^{-\zeta_1\,(T-s)}\,e^{\zeta_1\,(T-s)}\,ds\\
        &\leq\|v-w\|_{\zeta_1}\int_t^T e^{\zeta_1\,(T-s)}\,ds\\
        &=\frac{e^{\zeta_1\,(T-t)}-1}{\zeta_1}\cdot\|v-w\|_{\zeta_1}.
    \label{aux_ineq}
    \end{aligned}
\end{equation}
For any $v\in C_b^-$, it is clear that $F_1v$ is also a bounded negative function. Since the term inside the conditional expectation is continuous and uniformly bounded by $\bar{A}+\bar{\phi}\,T$, the continuity of $F_1v$ follows from the dominated convergence theorem. Consider any $v$, $w\in C_b^-$ and it follows that, for all $(t,l)\in[0,T]\times \mathbb{R}^d$, 
\begin{equation}
    \begin{aligned}
    e^{-\zeta_1\,(T-t)}\,&|(F_1\,v)(t, l)-(F_1\,w)(t, l)|\\
    &\leq e^{-\zeta_1\,(T-t)}\, \mathbb{E}\Big[A\,|e^{\int_t^T \gamma (a_s+b_s)\,v(s, L_s^{t,l})\,ds}-e^{\int_t^T \gamma (a_s+b_s)\,w(s, L_s^{t,l})\,ds}|\\
    &\hspace{2.5cm}+\bar{\phi}\int_t^T|e^{\int_t^s \gamma (a_u+b_u)\,v(u, L_u^{t,l})\,du}-e^{\int_t^s \gamma (a_u+b_u)\,w(u, L_u^{t,l})\,du}|\,ds\Big]\\
    &\leq e^{-\zeta_1\,(T-t)}\, \mathbb{E}\Big[A\,\big|\int_t^T \gamma (a_s+b_s)\,v(s, L_s^{t,l})\,ds-\int_t^T \gamma (a_s+b_s)\,w(s, L_s^{t,l})\,ds\big|\\
    &\hspace{2.5cm}+\bar{\phi}\int_t^T\big|\int_t^s \gamma (a_u+b_u)\,v(u, L_u^{t,l})\,du\\
    &\hspace{4,5cm}-\int_t^s \gamma (a_u+b_u)\,v(u, L_u^{t,l})\,du\big|\,ds\Big]\\
    &\leq e^{-\zeta_1\,(T-t)}\cdot \gamma\,(\bar{a}+\bar{b})\,(\bar{A}+\bar{\phi}\,T)\cdot \mathbb{E}\Big[\int_t^T |v(s, L_s^{t,l})-w(s, L_s^{t,l})|\,ds\Big]\\
    &\leq \frac{\gamma\,(\bar{a}+\bar{b})\,(\bar{A}+\bar{\phi}\,T)}{\zeta_1}\cdot\|v-w\|_{\xi_1}
    \nonumber
    \end{aligned}
\end{equation}
where the positiveness of $a, b$ and negativeness of $v$, $w$ are used to derive the second inequality, and the proposed inequality \eqref{aux_ineq} is applied in the last line. For a contraction mapping, it suffices to pick $\zeta_1\in\mathbb{R}_+$ such that 
\begin{equation}
    \zeta_1>\gamma\,(\bar{a}+\bar{b})\,(\bar{A}+\bar{\phi}\,T).
    \nonumber
\end{equation}
\indent Similarly, we look at the second statement and compute that
\begin{equation}
    \begin{aligned}
    e^{-\zeta_2\,(T-t)}\,|(F_2\,v)(t, l)&-(F_2\,w)(t, l)|\\
    &=e^{-\zeta_2\,(T-t)}\, \mathbb{E}\big|\int_t^T\gamma (a_s+b_s)\,\big[\psi(v(s,L_s^{t,l}))^2-\psi(w(s, L_s^{t,l}))^2\big]\,ds\,\big|\\
    &\leq2\xi\,e^{-\zeta_2\,(T-t)}\, \mathbb{E}\int_t^T\big|\gamma (a_s+b_s)\,\big[\psi(v(s,L_s^{t,l}))-\psi(w(s,L_s^{t,l}))\big]\,\big|\,ds\\
    &\leq \frac{2\xi\,\gamma\,(\bar{a}+\bar{b})}{\zeta_2}\cdot\|v-w\|_{\xi_2}.
    \nonumber
    \end{aligned}
\end{equation}
Again, it suffices to pick $\zeta_2$ large enough such that $\zeta_2>2\xi\,\gamma\,(\bar{a}+\bar{b})$ for a contraction mapping. The boundedness and continuity of $F_2v$ follows from the uniform boundedness and continuity inside the conditional expectation. 
\end{proof}

\noindent Due to the Banach fixed-point theorem, both $F_1$ and $F_2$ have their own fixed points. One reason why we consider both two mappings is that they actually share the same fixed point if the constant $\xi$ is sufficiently large.

\begin{lemma}
Let $v\in C_b^-$ be the fixed point of the $F_1$, then it is also the one of $F_2$, provided that the constant $\xi$ of the truncation $\psi$ is large enough.
\end{lemma}
\begin{proof}
Since $v$ is the fixed point of $F_1$, for any $(t,l)\in[0,T]\times \mathbb{R}^d$ it holds
\begin{equation}
    v(t, l)=\mathbb{E}\Big[-A\,e^{\int_t^T \gamma\,(a_s+b_s)\,v(s, L_s^{t,l})\,ds}-\int_t^T\phi_s\, e^{\int_t^s \gamma\,(a_u+b_u)\,v(u,L_u^{t,l})\,du}\,ds\Big].
    \nonumber
\end{equation}
We take the derivative of $\exp(\int_t^T\gamma (a_s+b_s)\,v(s,L_s^{t,l})\,ds)$ to obtain
\begin{equation*}
\begin{aligned}
    de^{\int_r^T \gamma (a_s+b_s)\,v(s,L_s^{t,l})\,ds}&=-\gamma (a_r+b_r)\,v(r,L_r^{t,l})\,e^{\int_r^T\gamma (a_s+b_s)\,v(s,L_s^{t,l})\,ds}\,dr,\\
    A-A\,e^{\int_t^T\gamma (a_s+b_s)\,v(s,L_s^{t,l})\,ds}&=-A\int_t^T\gamma (a_s+b_s)\,v(s,L_s^{t,l})\,e^{\int_s^T\gamma (a_u+b_u)\,v(u,L_u^{t,l})\,du}\,ds,
\end{aligned}
\end{equation*}
and then take the expectation to see
\begin{equation}
\begin{aligned}
   \mathbb{E}\big[ A-A\,e^{\int_t^T\gamma (a_s+b_s)\,v(s,L_s^{t,l})\,ds}\big]&=\mathbb{E}\int_t^T\gamma (a_s+b_s)\,v(s,L_s^{t,l})\\
    &\hspace{0.5cm}\cdot\mathbb{E}_s\big[(-A)\,e^{\int_s^T\gamma (a_u+b_u)\,v(u,L_u^{t,l})\,du}\big]\,ds,\\
    v(t, l)+\mathbb{E}\Big[\int_t^T\phi_s\,e^{\int_t^s \gamma\,(a_u+b_u)\,v(u,L_u^{t,l})\,du}\,ds\Big]&=\mathbb{E}[-A] + \mathbb{E}\int_t^T\gamma (a_s+b_s)\,v(s,L_s^{t,l})\\
    &\hspace{0.5cm}\cdot\mathbb{E}_s\big[(-A)\,e^{\int_s^T\gamma (a_u+b_u)\,v(u,L_u^{t,l})\,du}\big]\,ds.
    \label{fix_1}
\end{aligned}
\end{equation}
If one notes that
\begin{equation}
\begin{aligned}
    &\mathbb{E}\int_t^T\gamma (a_s+b_s)\,v(s,L_s^{t,l})\,\mathbb{E}_s\big[(-A)\,e^{\int_s^T\gamma (a_u+b_u)\,v(u,L_u^{t,l})\,du}\big]\,ds\\
    &=\mathbb{E}\int_t^T\gamma (a_s+b_s)\,v(s,L_s^{t,l})\,\big(v(s,L_s^{t,l})+\mathbb{E}_s\int_s^T \phi_r\,e^{\int_s^r \gamma (a_u+b_u) v(u,L_u^{s,L_s^{t,l}})\,du}\,dr\big)\,ds,
    \label{fix_2}
\end{aligned}
\end{equation}
which follows from the fixed point property, and also
\begin{equation}
\begin{aligned}
    \mathbb{E}\Big[\int_t^T&\gamma\,(a_s+b_s)\,v(s, L_s^{t,l})\,\mathbb{E}_s\big(\int_s^T\phi_r\, e^{\int_s^r \gamma\,(a_u+b_u)\,v(u, L_u^{s,L_s^{t,l}})\,du}\,dr\big)\,ds\Big]\\
    &=\mathbb{E}\Big[\int_t^T\phi_r\,\big(\int_t^r \gamma\,(a_s+b_s)\,v(s, L_s^{t,l})\,e^{\int_s^r \gamma\,(a_u+b_u)\,v(u, L_u^{t,l})\,du}\,ds\big)\,dr\Big]\\
    &=\mathbb{E}\Big[\int_t^T\phi_r\,\big(e^{\int_t^r \gamma\,(a_u+b_u)\,v(u, L_u^{t,l})\,du}-1\big)dr\Big],
\end{aligned}
\label{fix_3}
\end{equation}
one can observe the following result via \eqref{fix_1}--\eqref{fix_3}:
\begin{equation}
    v(t,l)=\mathbb{E}\Big[-A+\int_t^T \gamma\,(a_s+b_s)\,v(s, L_s^{t,l})^2\,ds-\int_t^T \phi_s\,ds\Big].
    \nonumber
\end{equation}
Since $v$ is the fixed point of $F_1$, it's clear that $v$ is uniformly bounded by $\bar{A}+\bar{\phi}\,T$ and, if we pick $\xi\geq \bar{A}+\bar{\phi}\,T$, the equation above is equivalent to
\begin{equation}
    v(t, l)=\mathbb{E}\Big[-A+\int_t^T \gamma\,(a_s+b_s)\,\psi(v(s,L_s^{t,l}))^2\,ds-\int_t^T \phi_s\,ds\Big],
    \nonumber
\end{equation}
showing that $v$ is also the fixed point of $F_2$. 
\end{proof}

\noindent The other reason why we consider both two mappings is that the mapping $F_2$ has been studied by \cite{becherer2005classical} under some relative general assumptions. We can now extend its result to our equation \eqref{non_lin_pde}.

\begin{theorem}
\label{pde_thm}
The equation \eqref{non_lin_pde} accepts a unique classical solution in $C_b^{1,2}$. More specifically, the solution takes value in $[-\bar{A}-\bar{\phi}\,T,0)$.
\end{theorem}

\begin{proof}
We start to look at the PDE \eqref{non_lin_pde} with a regularization:
\begin{equation}
        \frac{\partial h_2}{\partial t}+\mathcal{L}h_2(t,l)+\gamma\big(a(l)+b(l)\big)\,\psi(h_2(t,l))^2-\phi(l)=0, \quad \text{such that}\;\; h_2(T, l)=-A(l).
    \label{loc_pde}
\end{equation}
In order to apply the proposition 2.3 from \cite{becherer2005classical}, for convenience we set $\iota(t,l,v)=\gamma\, (a(l)+b(l))\,\psi(v)^2-\phi(l)$ and check the following set of conditions:

\begin{itemize}
    \item [1.] functions A and $\iota$ are continuous and bounded, with $\iota$ being Lipschitz in $v$ uniformly in $t$ and $l$;\\
    \vspace{-0.2cm}

    \item [2.] there exists a sequence $(D_n)_{n\in\mathbb{N}}$ of bounded domains with closure $\bar{D}_n \subseteq \mathbb{R}^d$ such that $\bigcup^\infty_
    {n=1}D_n = \mathbb{R}^d$ and each $D_n$ has a $C^2$-boundary;\\
    \vspace{-0.2cm}
    
    \item [3.] the functions $\Gamma$ and $\Sigma\,\Sigma^{tr}$ are uniformly Lipschitz-continuous on $[0,T] \times\bar{D}_n$;\\
        \vspace{-0.2cm}
    
    \item [4.] $\det
    (\Sigma(t,l)\,\Sigma^{tr}(t,l))\neq0$ for all $(t,l)\in[0,T] \times \mathbb{R}^d$;\\
        \vspace{-0.2cm}

    \item [5.] $(t, l, v)\mapsto \iota(t, l, v)$ is uniformly H\"{o}lder-continuous on $[0,T]\times \bar{D}_n\times \mathbb{R}$.
\end{itemize}

\noindent The first two conditions can be verified directly based on the Assumption \ref{lip_ord_flow}, while the third and fourth ones follow from the Assumption \ref{sde_assumption}. For any $(t, l, v), (t', l', v')\in [0,T]\times \bar{D}_n\times \mathbb{R}$, note that
\begin{equation}
\begin{aligned}
    |\iota(t', l', v')-\iota(t, l, v)|&\leq\gamma\,(a(l')+b(l'))\,|\psi(v')^2-\psi(v)^2|\\
    &\hspace{1.8cm}+\gamma\,\psi(v)^2\,|a(l')+b(l')-a(l)-b(l)|+|\phi(l')-\phi(l)|\\
    &\leq2\,\gamma\,\xi\, (\bar{a}+\bar{b})\,|v'-v|+2\,\gamma\,\xi^2\,C_{n}\,|l'-l|+C_n\,|l'-l|,
    \label{lip_iota}
\end{aligned}
\end{equation}
where $C_n$ denotes the Lipschitz coefficient of $a$, $b$ and $\phi$ in the domain $\bar{D}_n$ using the Assumption \ref{lip_ord_flow}. Equation \eqref{lip_iota} shows that $\iota$ is Lipschitz-continuous on $[0,T]\times \bar{D_n}\times \mathbb{R}$, and hence uniformly H\"{o}lder-continuous due to the boundedness. Consequently, the Proposition 2.3 from \cite{becherer2005classical} tells us that the `regular' PDE \eqref{loc_pde} has a unique classical solution in $C_b^{1,2}$, which is given by the fixed point $\hat{v}$ of the mapping $F_2$. Since $F_1$ and $F_2$ share the same fixed point, $\hat{v}$ is also the fixed point of $F_1$ and we already know that $\hat{v}$ is bounded by $\bar{A}+\bar{\phi}\,T$. Therefore, if $\xi$ is chosen to be larger than $\bar{A}+\bar{\phi}\,T$, it follows that 
\begin{equation}
\begin{aligned}
        \frac{\partial \hat{v}}{\partial t}+\mathcal{L}\hat{v}(t,l)&+\gamma\big(a(l)+b(l)\big)\,\psi(\hat{v}(t,l))^2-\phi(l)\\
        &=\frac{\partial \hat{v}}{\partial t}+\mathcal{L}\hat{v}(t,l)+\gamma\big(a(l)+b(l)\big)\,\hat{v}(t,l)^2-\phi(l)=0,
    \nonumber
\end{aligned}
\end{equation}
establishing the existence of solutions.

For the uniqueness, suppose that $u\in C_b^{1,2}$ is another solution to the PDE \eqref{non_lin_pde}. Consider another truncation function $\tilde{\psi}$ with the parameter 
$\tilde{\xi}$ satisfying
\begin{equation}
    \tilde{\xi}\;\geq\;(\bar{A}+\bar{\phi}\,T)\;\vee\sup_{(t, l)\in[0,T]\times \mathbb{R}^d}|u(t, l)|.
    \nonumber
\end{equation}
Note that both $\hat{v}$ and $u$ solve the following `regular' PDE:
\begin{equation}
        \frac{\partial h}{\partial t}+\mathcal{L}h(t,l)+\gamma\big(a(l)+b(l)\big)\,\tilde{\psi}(h(t,l))^2-\phi(l)=0, \quad \text{such that}\;\; h(T, l)=-A(l).
    \nonumber
\end{equation}
which has a unique solution according to previous discussions. Hence, $\hat{v}\equiv u$ and the uniqueness is obtained. 
\end{proof}

\vspace{0.2cm}

\begin{remark}
\label{unbound_order_flow}
\noindent Our method can be further generalized to unbounded functions $a$ and $b$. If $a, b$ satisfy proper growth conditions, consider equations
\begin{equation*}
    \frac{\partial h_n}{\partial t}+\mathcal{L}h_n(t,l)+\gamma\big(a(l)\wedge n+b(l)\wedge n\big)\,h_n(t,l)^2-\phi(l)=0, \quad \text{such that}\;\; h_n(T, l)=-A(l),
\end{equation*}
for $n\in \mathbb{N}$. A sequence of solutions $(h_n)_{n\in\mathbb{N}}$, being uniformly bounded, is guaranteed by the Theorem \ref{pde_thm} and is non-decreasing by the comparison result. If we denote by $g$ as the limit of the sequence, it turns out that $g$ satisfies the conditional representation
\begin{equation*}
    g(t,l)=\mathbb{E}\Big[-A+\int_t^T \gamma (a_s+b_s)\,\big(\psi(g(s,L_s^{t,l}))\big)^2\,ds-\int_t^T \phi_s\,ds\Big]
\end{equation*}
for a truncation function $\psi$ with the parameter $\bar{A}+\bar{\phi}\,T$. Then, one can use the proof of \cite{becherer2005classical} to show that $g\in C^{1,2}$ solves the PDE. We do not cover it here, because the boundedness of order flows guarantees the square integrability of the inventory (and thus  \eqref{martingale_price}). But, the idea will be illustrated in the later section through some other equations.
\end{remark}

\begin{example}
\label{section1_example}
We conclude this section with a particular example of \eqref{non_lin_pde} that has an explicit solution. Some economic intuition is also helpful in a later proof. First, the boundedness of $\partial h_2/\partial l$ needs to be derived. For this purpose, several additional conditions are imposed:
\begin{itemize}
    \item each element in the matrix-valued function $\Sigma(t,l)$ is a bounded function;\\
    \vspace{-0.2cm}

    \item $\phi\equiv0$, i.e., the agent is risk neutral;\\
    \vspace{-0.2cm}
    
    \item functions $a, b$, and $A$ are all uniformly Lipschitz, the Lipschitz coefficient of which is denoted by $C_{\text{Lip}}$.
\end{itemize}
Define the function $g:[0,T]\times\mathbb{R}^d\to \mathbb{R}$ as
\begin{equation*}
    g(t,l)= \frac{h_2(t,l+\epsilon\,e^i)-h_2(t,l)}{\epsilon},
\end{equation*}
where $\epsilon>0$ is a constant and $e^i$ is the unit vector of the $i$-coordinate in $\mathbb{R}^d$ for $i\in\{1,\dots, d\}$. Applying a similar transform to the equation \eqref{non_lin_pde}, one can observe that $g$ satisfies the following linear PDE:
\begin{equation}
\begin{aligned}
    \frac{\partial g}{\partial t}+\mathcal{L}g(t,l)&+\gamma\big(a(l+\epsilon\,e^i)+b(l+\epsilon\,e^i)\big)\,\big(h_2(t,l+\epsilon\,e^i)+h_2(t,l)\big)\,g(t,l)\\
    &+\big(\tilde{a}(l)+\tilde{b}(l)\big)\,h_2(t,l)=0,
\end{aligned}
\label{linear_pde_diff}
\end{equation}
such that $g(T,l)=-\tilde{A}(l)$, where we have defined 
\begin{equation*}
    \tilde{\theta}(l)=\frac{\theta(t,l+\epsilon\,e^i)-\theta(t,l)}{\epsilon}
\end{equation*}
for $\theta\in\{a, b, A\}$. To verify the term $(a(l+\epsilon\,e^i)+b(l+\epsilon\,e^i))\,(h_2(t,l+\epsilon\,e^i)+h_2(t,l))$ is uniformly H\"{o}lder-continuous, it suffices to notice that:
\begin{itemize}
    \item functions $a$ and $b$ are uniformly Lipschitz;
    \vspace{0.2cm}

    \item function $h_2(t,l)$ is uniformly Lipschitz on $[0,T]\times D$ for any compact set $D\subset \mathbb{R}^d$, since partial derivatives of $h_2$ are continuous;
    \vspace{0.2cm}

    \item functions $a, b$, and $h_2$ are uniformly bounded.
\end{itemize}
While the other conditions can be checked in a similar way, one can again learn from \cite{becherer2005classical} that the PDE \eqref{linear_pde_diff} accepts a unique solution $g\in C_b^{1,2}$. The solution $g$ possesses the representation
\begin{equation*}
\begin{aligned}
    &g(t,l)=\mathbb{E}\Big[-\Tilde{A}(L_T^{t,l})\,e^{\int_t^T \gamma(a(L_s^{t,l}+\epsilon\,e^i)+b(L_s^{t,l}+\epsilon\,e^i))\,(h_2(s, L_s^{t,l}+\epsilon\,e^i)+h_2(s, L_s^{t,l}))\,ds}\\
    &\hspace{1cm} +\int_t^T (\tilde{a}(L_s^{t,l})+\tilde{b}(L_s^{t,l}))\,h_2(s,L_s^{t,l})\\
    &\hspace{4cm}\cdot e^{\int_t^s \gamma(a(L_u^{t,l}+\epsilon\,e^i)+b(L_u^{t,l}+\epsilon\,e^i))\,(h_2(u,L_u^{t,l}+\epsilon\,e^i)+h_2(u,L_u^{t,l}))\,du}\Big],
\end{aligned}
\end{equation*}
which further implies
\begin{equation}
    |g(t,l)|\leq C_{\text{Lip}}\,e^{2\gamma T(\bar{a}+\bar{b})(\bar{A}+\bar{\phi}\,T)}+2\,T\,C_{\text{Lip}}(\bar{A}+\bar{\phi}\,T)\,e^{2\gamma T(\bar{a}+\bar{b})(\bar{A}+\bar{\phi}\,T)}
    \label{upper_bound_linear_pde}
\end{equation}
for all $(t,l)\in[0,T]\times\mathbb{R}^d$. We emphasize that the right hand side of \eqref{upper_bound_linear_pde} is independent of $\epsilon$. It follows $\partial h_2/\partial l^i$ is uniformly bounded and the same is true for all elements in the vector $\partial h_2/\partial l$.

Let $Y_t:=h_2(t,L_t)$, $Z_t:=(\partial h_2/\partial l)(t,L_t)\,\Sigma(t,L_t)$, and then apply It\^o's formula to obtain
\begin{equation*}
    Y_t=-A+\int_t^T \gamma\,(a_s+b_s)\,Y_s^2\,ds-\int_t^TZ_s\,dW_s.
\end{equation*}
Because $Y_t<0$, to observe the following we use It\^o's formula again:
\begin{equation}
\begin{aligned}
    d(Y_t^{-1})&=-A^{-1}-\int_t^T \gamma\,(a_s+b_s)\,ds-\int_t^T\frac{Z_s}{Y_s^2}\big(dW_s-\frac{Z_s}{Y_s}\,ds\big)\\
    &=-A^{-1}-\int_t^T \gamma\,(a_s+b_s)\,ds-\int_t^T\frac{Z_s}{Y_s^2}\,d\tilde{W}_s,
\end{aligned}
\label{chang_measur}
\end{equation}
where $\tilde{W}$ is a standard Brownian motion with respect to some probability measure $\mathbb{Q}$, guaranteed by the Girsanov theorem. Indeed, the Girsanov transform can be applied since $Z_t$ is bounded and $Y_t<-A\,\exp{(-\gamma\,(\bar{a}+\bar{b})\,A\,T)}$ by the mapping $F_1$. An explicit form of $h_2$ is obtained via taking the conditional expectation of \eqref{chang_measur}:
\begin{equation*}
    h_2(t,L_t)=Y_t=\bigg\{\,\mathbb{E}_t^\mathbb{Q} \Big[-A^{-1}-\int_t^T \gamma\,(a_s+b_s)\,ds\Big]\,\bigg\}^{-1}.
\end{equation*}
The first-order approximation, mentioned in \cite{barger2019optimal} and \cite{fouque2022optimal}, turns out to be the rectangle approximation of the integral $\int_t^T (a_s+b_s)\,ds$. In our context, the integral $\int_t^T (a_s+b_s)\,ds$ stands for the total volume of future order flows. This implies that $h_2$ decreases as the total volume declines. This can be naturally explained by the fact that when expected future market activity is low, it becomes more challenging for the market maker to clear their inventory. Therefore, the market maker is willing to quote a more favourable price to adjust her position.

While $h_2$ is the coefficient of the second-order term in $H$, the coefficient of the first-order term $h_1$---by the Feynman–Kac Formula---can be represented by
\begin{equation*}
    h_1(t, l )=\zeta\,\mathbb{E}_{t, l}\Big[\int_t^T (b_s-a_s)\cdot h_2(s, L_s)\,e^{\gamma\int_t^s (a_u+b_u)\,h_2(u, L_u)\,du}\,ds\Big],
\end{equation*}
which estimates a weighted sum of future order imbalances. Since
\begin{equation*}
    \delta_t^{a,*} = \frac{\zeta}{2\gamma} + \frac{1}{2}\frac{\partial H}{\partial q} = \frac{\zeta}{2\gamma} + \frac{1}{2} h_1(t, l) + h_2(t, l)\, q,
\end{equation*}
we observe that the quote depends on a weighted sum of future order imbalances, adjusted by the current inventory level. 

\end{example}

\section{Non-Markovian Order Flows: General Intensity Functions} \label{section_3}
\noindent In this section, we extend the market making problem in two directions. First, the order flows and penalty parameters are not necessarily Markovian.

\begin{assumption}
Let $a$, $b\in\mathbb{H}^2$ be positive processes and $\phi\in\mathbb{H}^2$ be a non-negative process,  such that they are bounded by constants $\bar{a}$, $\bar{b},$ $\bar{\phi}>0$ respectively. Let $A\in L^2(\Omega,\mathcal{F}_T)$ be a non-negative random variable, bounded by $\bar{A}>0$.
\end{assumption}

\noindent On the other hand, more general intensity functions are considered. We borrow the following notations, definitions, and associated results from \cite{gueant2017optimal}.

\begin{assumption}[\cite{gueant2017optimal}]
\label{general_inten}
A function $\Lambda:\mathbb{R}\to\mathbb{R_+}$ belongs to the class of intensity functions $\boldsymbol{\Lambda}$ if:
\begin{itemize}
    \item[1.] $\Lambda$ is twice continuously differentiable;\\
    \vspace{-0.2cm}
    
    \item[2.] $\Lambda$ is strictly decreasing and hence $\Lambda'(x)<0$ for any $x\in\mathbb{R}$;\\
    \vspace{-0.2cm}
    
    \item[3.] $\lim_{x\to\infty}\Lambda(x)=0\,$ and $\;\sup_{x\in\mathbb{R}}\frac{\Lambda(x)\,\Lambda''(x)}{(\Lambda'(x))^2}<2$.
\end{itemize}
\label{inten_assu}
\end{assumption}

\begin{lemma}[\cite{gueant2017optimal}]
\label{inten_fun}
For any $\Lambda\in\boldsymbol{\Lambda}$, define the function $\mathcal{W}:\mathbb{R}\to\mathbb{R}$ as $\mathcal{W}(p)=\sup_{\delta\in\mathbb{R}}\Lambda(\delta)\,(\delta-p)$. Then, the following holds: 
\begin{itemize}
    \item[1.] $\mathcal{W}$ is a decreasing function of class $C^2$;\\
    \vspace{-0.2cm}
    
    \item[2.] The supremum in the definition of $\mathcal{W}$ is attained at a unique $\delta^*(p)$ characterized by 
    \begin{equation}
        \delta^*(p)=\Lambda^{-1}\big(-\mathcal{W}'(p)\big),
        \nonumber
    \end{equation}
    where $\Lambda^{-1}$ denotes the inverse function of $\Lambda$;\\
    \vspace{-0.2cm}
    
    \item[3.] The function $p\mapsto\delta^*(p)$ belongs to $C^1$ and is increasing. Its derivative reads
    \begin{equation}
        (\delta^{*})'(p)=\Big[2-\frac{\Lambda(\delta^*(p))\,\Lambda''(\delta^*(p))}{\Lambda'(\delta^*(p))^2}\Big]^{-1}>0.
        \nonumber
    \end{equation}
\end{itemize}
\end{lemma}
\vspace{0.1cm}

\begin{remark}
A common choice for the function $\Lambda$ is the exponential decay form, $\Lambda(x) = \exp(-\gamma \, x)$, where $\gamma > 0$. For its derivation, we refer readers to \cite{avellaneda2008high}. To estimate the parameter, we can collect data that records the number of orders executed at various distances from the reference price. The parameter $\gamma$ can then be determined using suitable fitting procedures.
\end{remark}

\noindent Denote by $(\Lambda^a,\Lambda^b)\in\boldsymbol{\Lambda}\times\boldsymbol{\Lambda}$ the intensity functions of the ask and bid sides. Additionally, a strategy $\boldsymbol{\delta}:=(\delta^a,\delta^b)$ is admissible if both two entries are bounded by $\xi$ for some large $\xi>0$, i.e., the admissible space is defined as
\begin{equation*}
    \mathbb{A}:=\{\delta\in \mathbb{H}^2: |\delta_t|\leq \xi \text{ for all } t\in[0,T]\,\}.
\end{equation*}    
The inventory and cash processes of the market maker are then given by
\begin{gather*}
        Q_t = q_0-\int_0^t a_u\Lambda^a(\delta_u^a)\,du+\int_0^t b_u\Lambda^b(\delta_u^b)\,du,\\
        \nonumber
    X_t = x_0+\int_0^t a_u\Lambda^a(\delta_u^a)\,(S_u+\delta_u^a)\,du-\int_0^t b_u\Lambda^b(\delta_u^b)\,(S_u-\delta_u^b)\,du.
    \nonumber
\end{gather*}
The agent intends to maximise the same objective functional:
\begin{equation*}
\begin{aligned}
    &\mathbb{E}\big[X_T+S_T\,Q_T-\int_0^T \phi_t\,(Q_t)^2\,dt-A\,(Q_T)^2 \big]\\
   &=\mathbb{E}\Big[\int_0^T\delta_t^a\,a_t\,\Lambda^a(\delta_t^a)\,dt+\int_0^T\delta_t^b\,b_t\,\Lambda^b(\delta_t^b)\,dt-\int_0^T \phi_t\,(Q_t)^2\,dt-A\,(Q_T)^2 \Big]
\end{aligned}
\end{equation*}
by controlling $\boldsymbol{\delta}\in\mathbb{A}\times\mathbb{A}$. To apply the Pontryagin maximum principle, the Hamiltonian of the agent with respect to the objective functional reads
\begin{equation}
    \mathscr{H}(t,Q_t,Y_t,\boldsymbol{\delta}_t)=\big[b_t\,\Lambda^b(\delta_t^b)-a_t\,\Lambda^a(\delta_t^a)\big]\,Y_t+b_t\,\delta_t^b\Lambda^b(\delta_t^b)+a_t\,\delta_t^a\Lambda^a(\delta_t^a)-\phi\,Q_t^2.
    \nonumber
\end{equation}
While the function $\mathscr{H}$ is concave in the state variable $Q$, the concavity with respect to the control $\boldsymbol{\delta}$ is not guaranteed, which is required by the usual stochastic maximum principle (for example \cite{carmona2016lectures}). However, the separation between the state variable and control enables us to still apply this principle, the proof of which is placed in the appendix. Denote by $(\delta^{a*},\delta^{b*})$ the maximisers associated with the intensity functions $(\Lambda^a, \Lambda^b)$ according to Lemma \ref{inten_fun}, and further define 
\begin{equation*}
    \tilde{\delta}^{i*}(p):=\delta^{i*}(p)\wedge\xi\vee(-\xi)
\end{equation*}
for $i\in\{a,b\}$. Then the optimal feedback controls on the ask and bid sides are given by 
\begin{equation}
    \tilde{\delta}^{a*}(Y_t) \quad\text{and}\quad \tilde{\delta}^{b*}(-Y_t).
    \nonumber
\end{equation}
Indeed, on the ask side, we look for the maximizer of 
\begin{equation*}
    \sup_{|\delta_t^a|\leq \xi}g(\delta_t^a):=\sup_{|\delta_t^a|\leq \xi}\Lambda^a(\delta_t^a)\cdot(\delta_t^a-Y_t).
\end{equation*}
Without the constraint $|\delta_t^a|\leq \xi$, the maximiser simply reads $\delta^{a*}(Y_t)$ according to Lemma \ref{inten_fun}. However, being aware of
\begin{equation*}
    g'(\delta_t^a)=\Lambda^{a'}(\delta_t^a)\,\big[\delta_t^a+\frac{\Lambda^a(\delta_t^a)}{\Lambda^{a'}(\delta_t^a)}-Y_t\big]
\end{equation*}
and the fact that $\delta_t^a+\Lambda^a(\delta_t^a)/\Lambda^{a'}(\delta_t^a)$ is an increasing function because
\begin{equation*}
    \big[\delta_t^a+\frac{\Lambda^a(\delta_t^a)}{\Lambda^{a'}(\delta_t^a)}\big]'=2-\frac{\Lambda^a(\delta_t^a)\,\Lambda^{a''}(\delta_t^a)}{\Lambda^{a'}(\delta_t^a)^2}>0,
\end{equation*}
one can see that the function $g$ first increases then decreases. Hence, the maximiser is $\tilde{\delta}^{a*}(Y_t)$ under the constraint stated in $\mathbb{A}$. The bid side can be derived in the same way. Combined with the stochastic maximum principle, the above discussion yields the next result.

\begin{theorem}
\label{stochast_max_principle}
Let $\boldsymbol{\delta}\in\mathbb{A}\times\mathbb{A}$ be an admissible control, 
$Q = Q^{\boldsymbol{\delta}}$ be the corresponding controlled inventory, and $(Y, Z)$ be the adjoint processes. Then $\boldsymbol{\delta}$ is an optimal control if and only if it holds $\mathbb{P}$-a.s. that
\begin{equation}
    \mathscr{H}(t,Q_t,Y_t,\boldsymbol{\delta}_t)=\sup_{\boldsymbol{\alpha}\in [-\xi, \xi]^2}\mathscr{H}(t,Q_t,Y_t,\boldsymbol{\alpha}),  \qquad \text{a.e. in } t\in[0,T],
    \nonumber
\end{equation}
or equivalently
\begin{equation}
    \delta^a_t=\tilde{\delta}^{a*}(Y_t), \quad\text{and}\quad \delta^b_t=\tilde{\delta}^{b*}(-Y_t),  \qquad \text{a.e. in } t\in[0,T].
    \nonumber
\end{equation}
Further, the optimal inventory $Q$ together with the adjoint processes $(Y,Z)$ solves the FBSDE 
\begin{equation}
\left\{
\begin{aligned}
\;& dQ_t  = -a_t\,\Lambda^a\big(\tilde{\delta}^{a*}(Y_t)\big)dt+b_t\,\Lambda^b\big(\tilde{\delta}^{b*}(-Y_t)\big)dt, \\
& dY_t=2\phi_t\,Q_t\,dt+Z_t\,dW_t,\\
& Q_0=q_0,\quad Y_T=-2A\,Q_T.
\label{Lip_den_FBSDE}
\end{aligned}
\right.
\end{equation}
\end{theorem}
\begin{proof}
    See the \hyperref[section_5]{Appendix} for the stochastic maximum principle. 
\end{proof}

\noindent Equation \eqref{Lip_den_FBSDE} is a degenerate FBSDE with Lipschitz coefficients. To see this, note that $(\Lambda^a,\Lambda^b)$ have bounded derivatives in the admissible action space, and the same is true for $(\delta^{a*},\delta^{b*})$ since
\begin{equation}
        0<(\delta^{i*})'(p)=\Big[2-\frac{\Lambda^i(\delta^{i*}(p))\,\Lambda^{i''}(\delta^{i*}(p))}{\Lambda^{i'}(\delta^{i*}(p))^{2}}\Big]^{-1}\leq\Big[2-\sup_{u\in\mathbb{R}}\frac{\Lambda^i(u)\,\Lambda^{i''}(u)}{\Lambda^{i'}(u)^{2}}\Big]^{-1}<\infty,
    \nonumber
\end{equation}
for $i\in\{a,b\}$. The Lipschitz property of the forward equation of \eqref{Lip_den_FBSDE} then follows from the boundedness of these two derivatives and flows $a$, $b$. The remaining of section is devoted to the (global) well-posedness of the FBSDE \eqref{Lip_den_FBSDE}.

\begin{remark}
(1) The constant $\xi$ in the definition of $\mathbb{A}$ can be also interpreted as a regularizer for the equation \eqref{Lip_den_FBSDE} to make it Lipschitz. Later, we will see how it can be removed.\\
\indent (2) If one additionally imposes that $a$, $b$ and $\phi$ are bounded away from $0$, the well-posedness of \eqref{Lip_den_FBSDE} can be verified by the method of continuation (see \cite{peng1999fully}).
\end{remark}

We first introduce an auxiliary result on a linear degenerate FBSDE, the proof of which is based on the same technique used for \eqref{non_lin_pde}.

\begin{lemma}
\label{growth_lip_fbsde}
Let $(\mu_t)_{0\leq t\leq T}\in \mathbb{H}^2$ be a non-negative bounded process with the upper bound $\bar{\mu}>0$, and $\nu_T\in L^2(\Omega,\mathcal{F}_T)$ be a non-negative bounded random variable bounded by $\bar{\nu}>0$. The FBSDE
\begin{equation}
\left\{
\begin{aligned}
\;& dQ_t  = \mu_t\,Y_t\,dt, \\
& dY_t=2\phi_t\,Q_t\,dt+Z_t\,dW_t,\\
& Q_0=q_0,\quad Y_T=-\nu_T\,Q_T
\label{Lin_noncont_FBSDE}
\end{aligned}
\right.
\end{equation}
admits a unique solution in $\mathbb{S}^2\times\mathbb{S}^2\times\mathbb{H}^2$. In particular, the solution accepts the representation
\begin{equation}
    Y_t=P_t\cdot Q_t \text{\quad and \quad} Q_t=q_0\,\exp\big(\int_0^t\mu_s\,P_s\,ds\big),
\nonumber
\end{equation}
with $(P_t)_{0\leq t\leq T}\in\mathbb{H}^2$ being non-positive and bounded by $\bar{\nu}+2\bar{\phi}\,T$.
\end{lemma}

\begin{proof}
Motivated by the linearity of \eqref{Lin_noncont_FBSDE}, we adopt the affine ansatz $Y_t=P_t\cdot Q_t$ for some $P:=(P_t)_{0\leq t\leq T}$ to be determined. Through applying It\^o's formula and matching the coefficients, one can see that $P$ satisfies the BSDE
\begin{equation}
    dP_t=\big(-\mu_t\,(P_t)^2+2\phi_t\big)\,dt+\tilde{Z}_t\,dW_t, \quad \text{such that}\;\; P_T=-\nu_T.
    \label{nonlip_bsde}
\end{equation}
To study this non-Lipschitz BSDE, we similarly look at two auxiliary linear BSDEs:
\begin{gather*}
    dP_t=\big(-\mu_t\,(G_t\wedge(-\tilde{\xi})\vee \tilde{\xi})^2+2\phi_t\big)\,dt+Z_t^1\,dW_t, \quad \text{such that}\;\; P_T=-\nu_T;\\
    dP_t=\big(-\mu_t\,G_t P_t+2\phi_t\big)\,dt+Z_t^2\,dW_t, \quad \text{such that}\;\; P_T=-\nu_T,
\end{gather*}
and their associated mappings 
\begin{gather*}
    (F_3\,G)_t=\mathbb{E}_t\Big[-\nu_T+\int_t^T \mu_s\,(G_s\wedge(-\tilde{\xi})\vee \tilde{\xi})^2\,ds-\int_t^T 2\phi_s\,ds\Big],\\
    (F_4\,G)_t=\mathbb{E}_t\Big[-\nu_T\,e^{\int_t^T \mu_u\,G_u\,du}-\int_t^T2\phi_s\,e^{\int_t^s\mu_u\,G_u\,du}\,ds\Big],
\end{gather*}
for some constant $\tilde{\xi}>0$ and process $G\in\mathbb{S}^2$ to be specified. As before, $F_3$ is a contraction mappings on the space $\mathbb{S}^2$ with respect to the norm
\begin{equation*}
    \|G\|_{\zeta_3}:=\mathbb{E}\big[\sup_{0\leq t\leq T}e^{-2\zeta_3\,(T-t)}\cdot |G_t|^2\,\big]^{1/2},
\end{equation*}
and $F_4$ is a contraction mappings on the space of non-positive processes in $\mathbb{S}^2$ with respect to the norm
\begin{equation*}
    \|G\|_{\zeta_4}:=\mathbb{E}\big[\sup_{0\leq t\leq T}e^{-2\zeta_4\,(T-t)}\cdot |G_t|^2\,\big]^{1/2},
\end{equation*}
for some $\zeta_3, \zeta_4>0$ large enough. We learn from the Banach fixed point theorem that both $F_3$ and $F_4$ have their own unique fixed points. Given $\tilde{\xi}>\bar{\nu}+2\bar{\phi}\,T$, they further share the same fixed point. Finally, the common fixed point solves the equation \eqref{nonlip_bsde} and guarantees the existence of the solution. To see that \eqref{nonlip_bsde} accepts a unique bounded solution, let $P^1$, $P^2\in\mathbb{S}^2$ be two solutions with bounds $\bar{P}^1$, $\bar{P}^2>0$ respectively. If we let $\Tilde{\xi}\geq\bar{P}^1\vee\bar{P}^2$, both $P^1$ and $P^2$ are then fixed points of the mapping $F_3$, implying $P^1=P^2$ due the uniqueness of the fixed point.

The uniqueness of \eqref{Lin_noncont_FBSDE} follows from the method of continuation as in \cite{peng1999fully}. Given two solutions $(Q, Y, Z )$ and $(\tilde{Q}, \tilde{Y}, \tilde{Z})$ in $\mathbb{S}^2\times\mathbb{S}^2\times\mathbb{H}^2$, then $(\mathcal{Q}, \mathcal{Y}, \mathcal{Z}):=(\tilde{Q}-Q, \tilde{Y}-Y, \tilde{Z}-Z)$ solves the FBSDE
\begin{equation}
\left\{
\begin{aligned}
\;& d\mathcal{Q}_t  = \mu_t\,\mathcal{Y}_t dt, \\
& d\mathcal{Y}_t=2\phi_t\,\mathcal{Q}_t\,dt+\mathcal{Z}_t\,dW_t,\\
& \mathcal{Q}_0=0,\quad \mathcal{Y}_T=-\nu_T\,\mathcal{Q}_T.
\nonumber
\end{aligned}
\right.
\end{equation}
While It\^o's formula implies
\begin{equation*}  d(\mathcal{Q}_t\,\mathcal{Y}_t)=\big[\mu_t(\mathcal{Y}_t)^2+2\phi_t(\mathcal{Q}_t)^2\big]\,dt+\mathcal{Q}_t\,\mathcal{Z}_t\,dW_t,
\end{equation*}
through observing
\begin{equation*}
        0\geq\mathbb{E}[-\nu_T\,(\mathcal{Q}_T)^2]=\mathbb{E}\int_0^T\big[\mu_t\,(\mathcal{Y}_t)^2+2\phi_t\,(\mathcal{Q}_t)^2\big]\,dt\geq \mathbb{E}\int_0^T\mu_t\,(\mathcal{Y}_t)^2\,dt\geq 0
\end{equation*}
we can see $\mu_t\,\mathcal{Y}_t=0$ a.e. in $t$ and thus $\mathcal{Q}=0$. The fact $\mathcal{Y}=\mathcal{Z}=0$ then follows immediately. 
\end{proof}

\begin{remark}
The arguments of the continuation method in proving the uniqueness of the solution will be applied to some other equations later. We emphasize that such arguments do not need the boundedness of $\mu$, $\phi$ and $\nu_T$.
\end{remark}

\noindent Since the FBSDE \eqref{Lip_den_FBSDE} is Lipschitz, denote by $\iota>0$ the Lipschitz constant for the coefficient functions of forward and backward equations of $\eqref{Lip_den_FBSDE}$ (excluding the terminal condition), i.e., it holds $\mathbb{P}$-a.s. for any $y, \Tilde{y}, q, \Tilde{q}$ that
\begin{equation*}
\begin{aligned}
    \Big|-a_t\,\Lambda^a\big(\tilde{\delta}^{a*}(y)\big)+b_t\,\Lambda^b\big(\tilde{\delta}^{b*}(-y)\big)+a_t\,\Lambda^a\big(\tilde{\delta}^{a*}(\tilde{y})\big)-b_t\,\Lambda^b\big(\tilde{\delta}^{b*}(-\tilde{y})\big)\Big|&\leq \iota\,|\tilde{y}-y|,\\
    \big|2\phi_t\,q-2\phi_t\,\tilde{q}\big|&\leq  \iota\,|\tilde{q}-q|,\\
    \big|2A\,q-2A\,\tilde{q}\big|&\leq  \iota\,|\tilde{q}-q|.
\end{aligned}
\end{equation*}
To prove the well-posedness, we apply the decoupling approach introduced by \cite{ma2015well}. Actually, Lemma \ref{growth_lip_fbsde} studies the growth of the Lipschitz coefficient---a critical step for such technique.

\begin{theorem}
\label{control_mono}
The FBSDE \eqref{Lip_den_FBSDE} accepts a unique solution in $\mathbb{S}^2\times\mathbb{S}^2\times\mathbb{H}^2$. Further, the optimal control is monotonic with respect to the initial (inventory) condition, i.e., if $q_1>q_2$, then it holds that
\begin{equation}
  \hat{\delta}_{1,t}^{a}\leq\hat{\delta}_{2,t}^{a}         \quad\text{and}\quad \hat{\delta}_{1,t}^{b}\geq\hat{\delta}_{2,t}^{b}
    \nonumber
\end{equation}
for  $t\in[0,T]$ $\mathbb{P}$-a.s., where $(\hat{\delta}_{i,t}^{a}, \hat{\delta}_{i,t}^{b})_{0\leq t\leq T}$ represents the optimal control associated with the initial condition $q_i$ for $i\in\{1, 2\}$.
\end{theorem}

\begin{proof}
\uline{\textit{Short time analysis}}: Since the FBSDE is of Lipschitz type, we start with the short time analysis. Fix any $q_0\in\mathbb{R}$, if $Q=(Q_t)_{0\leq t\leq T}\in\mathbb{S}^2$ is given such that $Q_0=q_0$, denote by $(Y, Z)\in\mathbb{S}^2\times\mathbb{H}^2$ the unique solution of the BSDE:
\begin{equation}
    dY_t=2\phi_t\,Q_t\,dt+Z_t\,dW_t, \text{\; such that \;} Y_T=-2A\,Q_T.
    \nonumber
\end{equation}
Subsequently, given $(Y, Z)$, we let $Q'_t$ be the uniqueness solution of the SDE
\begin{equation}
    dQ'_t=-a_t\,\Lambda^a\big(\tilde{\delta}^{a*}(Y_t)\big)dt+b_t\,\Lambda^b\big(\tilde{\delta}^{b*}(-Y_t)\big)dt,\text{\; such that \;} Q'_0=q_0.
    \nonumber
\end{equation}
In such way, we have defined a mapping
\begin{equation}
    \mathbb{S}^2\ni Q\hookrightarrow \Phi(Q)=Q'\in\mathbb{S}^2
    \nonumber
\end{equation}
and proceed to show that $\Phi$ is a contraction mapping when $T$ is small enough. For any $Q$, $\tilde{Q}\in\mathbb{S}^2$, write $Q'=\Phi(Q)$, $\tilde{Q}'=\Phi(\tilde{Q})$ and note that
\begin{equation}
\begin{aligned}
    |\tilde{Q}'_t-Q'_t|&\leq \iota \, T \sup_{0\leq s \leq T}|\tilde{Y}_s-Y_s|\\ \mathbb{E}\sup_{0\leq s \leq T}|\tilde{Q}'_s-Q'_s|^2&\leq (\iota \, T)^2\,\mathbb{E}\sup_{0\leq s \leq T}|\tilde{Y}_s-Y_s|^2.
\label{control_stab_1}
\end{aligned}
\end{equation}
To obtain the stability for the $\tilde{Y}$ and $Y$, we further compute
\begin{equation}
    \begin{aligned}
    |\tilde{Y}_t-Y_t|^2&\leq (2\bar{A}+2\iota T)^2\,\mathbb{E}_t\big[\sup_{0\leq s\leq T}|\tilde{Q}_s-Q_s|\big]^2,\\
    \mathbb{E}\sup_{0\leq s\leq T}|\tilde{Y}_s-Y_s|^2&\leq (2\bar{A}+2\iota T)^2\,\mathbb{E}\Big[\sup_{0\leq t\leq T}\mathbb{E}_t\big[\sup_{0\leq s\leq T}|\tilde{Q}_s-Q_s|\big]^2\Big],\\
    &\leq 4\,(2\bar{A}+2\iota T)^2\,\mathbb{E}\sup_{0\leq s\leq T}|\tilde{Q}_s-Q_s|^2,
    \end{aligned}
    \label{control_stab_2}
\end{equation}
where the Doob's $L^p$ inequality is applied in the last line.  Combining \eqref{control_stab_1} and \eqref{control_stab_2}, we find that
\begin{equation}
    \mathbb{E}\sup_{0\leq s \leq T}|\tilde{Q}'_s-Q'_s|^2\leq 4\,(\iota T)^2\,(2\bar{A}+2\iota T)^2\,\mathbb{E}\sup_{0\leq s\leq T}|\tilde{Q}_s-Q_s|^2
    \nonumber   
\end{equation}
and it suffices to pick $T$ satisfying
\begin{equation}
    4\,(\iota T)^2\,(2\bar{A}+2\iota T)^2<1
    \nonumber
\end{equation}
for a contraction mapping. The FBSDE \eqref{Lip_den_FBSDE} accepts a unique solution in $\mathbb{S}^2\times\mathbb{S}^2\times\mathbb{H}^2$ with respect to the such $T$.

\uline{\textit{Decoupling approach}}: we then extend the local result to any finite time horizon $T$. Define time step $\Delta>0$ as
\begin{equation}
    \Delta=\frac{1}{\sqrt{2}}\cdot \big[4\iota^2\,(2\bar{A}+2\bar{\phi}T+2\iota T)^2\big]^{-1/2}.
    \label{control_time_step}
\end{equation}
Note that 
\begin{equation*}
    4\,(\iota \Delta)^2\,(2\bar{A}+2\iota \Delta)^2\leq4\,(\iota \Delta)^2\,(2\bar{A}+2\bar{\phi}T+2\iota T)^2=\frac{1}{2}
\end{equation*}
and thus the following FBSDE is well-posed by the previous discussion:
\begin{equation}
\left\{
\begin{aligned}
\;& dQ_t  = -a_t\,\Lambda^a\big(\tilde{\delta}^{a*}(Y_t)\big)dt+b_t\,\Lambda^b\big(\tilde{\delta}^{b*}(-Y_t)\big)dt, \\
& dY_t=2\phi_t\,Q_t\,dt+Z_t\,dW_t,\\
& Q_{T-\Delta}=q_{T-\Delta},\quad Y_T=-2A\,Q_T.
\end{aligned}
\right.
\label{control_fbsde_1}
\end{equation}
We can then define the decoupling field $u: [T-\Delta,T]\times\Omega\times\mathbb{R}\to\mathbb{R}$ by
\begin{equation}
    u(t,q):=Y_t^{t,q},
    \label{control_decoupling}
\end{equation}
where the superscript $(t,q)$ denotes the initial time and condition of \eqref{control_fbsde_1}.  Let $q_1$ and $q_2$ be two initial conditions for the FBSDE \eqref{control_fbsde_1}, and---due to the well-posedness---denote by $(Q^1,Y^1,Z^1)$ and $(Q^2,Y^2,Z^2)$ the solutions corresponding to these initial conditions. For $t$ such that $Y_t^1\neq Y_t^2$, let
\begin{equation}
\begin{aligned}
    d(Q_t^1-Q_t^2)/dt&=-a_t\big[\Lambda^a\big(\tilde{\delta}^{a*}(Y_t^1)\big)-\Lambda^a\big(\tilde{\delta}^{a*}(Y_t^2)\big)\big]+b_t\,\big[\Lambda^b\big(\tilde{\delta}^{b*}(-Y_t^1)\big)-\Lambda^b\big(\tilde{\delta}^{b*}(-Y_t^2)\big)\big]\\
    &=:\mathcal{U}_t\cdot(Y_t^1-Y_t^2).
\end{aligned}
    \nonumber
\end{equation}
Naturally, for those $t$ such that $Y_t^1= Y_t^2$, we let $\mathcal{U}_t=0$. Note that $\mathcal{U}\in \mathbb{H}^2$ is bounded by $\iota$ due to the Lipschitz continuity. It is further non-negative since the expression
\begin{equation}
    -a\,\Lambda^a\big(\tilde{\delta}^{a*}(y)\big)+b\,\Lambda^b\big(\tilde{\delta}^{b*}(-y)\big)
    \nonumber   
\end{equation}
is non-decreasing with respect to $y$, which can be deduced from the increasing property of $(\delta^{a*},\delta^{b*})$ and decreasing property of $(\Lambda^a,\Lambda^b)$. One can then observe that $(Q^1-Q^2, Y^1-Y^2, Z^1-Z^2)$ satisfies the following FBSDE
\begin{equation}
    \left\{
\begin{aligned}
\;& d\mathcal{Q}_t  = \mathcal{U}_t\,\mathcal{Y}_tdt, \\
& d\mathcal{Y}_t=2\phi_t\,\mathcal{Q}_t\,dt+\mathcal{Z}_t\,dW_t,\\
& \mathcal{Q}_{T-\Delta}=q_1-q_2,\quad \mathcal{Y}_T=-2A\,\mathcal{Q}_T,
\end{aligned}
\right.
    \label{var_FBSDE}
\end{equation}
for $t\in[T-\Delta,T]$, which is called the variational FBSDE associated with the original \eqref{Lip_den_FBSDE}. While \eqref{var_FBSDE} is a special case of the result in Lemma \ref{growth_lip_fbsde}, equation \eqref{var_FBSDE} admits a unique solution 
\begin{equation}
    (\mathcal{Q}_t,\mathcal{Y}_t,\mathcal{Z}_t)=\big(Q_t^1-Q_t^2, Y_t^1-Y_t^2, Z_t^1-Z_t^2\big)
    \nonumber
\end{equation}
and, more importantly, process $\mathcal{Y}$ is bounded by $|q_1-q_2|\cdot(2\bar{A}+2\bar{\phi}\,\Delta)$. According to the definition of the decoupling field, we deduce that $u$ is uniformly Lipschitz continuous $\mathbb{P}$-a.s. with respect to variable $q$  with Lipschitz coefficient $2\bar{A}+2\bar{\phi}\,\Delta$. Consider the FBSDE in the next interval
\begin{equation}
\left\{
\begin{aligned}
\;& dQ_t  = -a_t\,\Lambda^a\big(\tilde{\delta}^{a*}(Y_t)\big)dt+b_t\,\Lambda^b\big(\tilde{\delta}^{b*}(-Y_t)\big)dt, \\
& dY_t=2\phi_t\,Q_t\,dt+Z_t\,dW_t,\\
& Q_{T-2\Delta}=q_{T-2\Delta},\quad Y_{T-\Delta}=u(T-\Delta,Q_{T-\Delta}).
\end{aligned}
\right.
\end{equation}
With the same short time analysis, it suffices to make sure
\begin{equation*}
    4\,(\iota \Delta)^2\,(2\bar{A}+2\bar{\phi}\Delta+2\iota \Delta)^2<1,
\end{equation*}
and our choice \eqref{control_time_step} is clearly qualified. Referring to Lemma \ref{growth_lip_fbsde}, we study the variational FBSDE and observe the linear growth of the Lipschitz coefficient of the decoupling field. Because the potential Lipschitz coefficient of the decoupling field can not exceed $2\bar{A}+2\bar{\phi}T$, the time step $\Delta$ will always qualify, and hence the field $u$ can be extended to the whole time horizon $[0,T]$ by a finite number of iterations. The global existence of the solution then can be inferred from the (global) decoupling field $u$. Because the uniqueness of solution can be proved in the same fashion as Lemma \ref{growth_lip_fbsde}, the well-posedness of \eqref{Lip_den_FBSDE} is established.

\uline{\textit{Monotonicity}}: to verify the monotonicity property, we let $q_1>q_2$ be two initial conditions, denoting by $(Q^1,Y^1,Z^1)$ and $(Q^2,Y^2,Z^2)$ the associated solutions. Since $(Q^1-Q^2, Y^1-Y^2, Z^1-Z^2)$ solves the variational FBSDE \eqref{var_FBSDE} with initial condition $q_1-q_2>0$, again by Lemma \ref{growth_lip_fbsde} we see $Y^1-Y^2$ is negative. Recalling that the optimal controls read
\begin{equation}
    \hat{\delta}_t^a=\tilde{\delta}^{a*}(Y_t) \quad\text{and}\quad \hat{\delta}_t^b=\tilde{\delta}^{b*}(-Y_t),
    \nonumber
\end{equation}
the smaller value of $Y_t$ results in a smaller $\hat{\delta}_t^a$ but a larger $\hat{\delta}_t^b$. 
\end{proof}

We finish this section by removing the constant $\xi$. It turns out that, provided $\xi$ is picked large enough, the corresponding constraint will have no impact on the control. That is to say, we are able to find a (unique) solution to the equation
\begin{equation}
\left\{
\begin{aligned}
\;& dQ_t  = -a_t\,\Lambda^a\big(\delta^{a*}(Y_t)\big)dt+b_t\,\Lambda^b\big(\delta^{b*}(-Y_t)\big)dt, \\
& dY_t=2\phi_t\,Q_t\,dt+Z_t\,dW_t,\\
& Q_0=q_0,\quad Y_T=-2A\,Q_T.
\end{aligned}
\right.
\label{non_lip_fbsde}
\end{equation}
Note that this FBSDE is non-Lipschitz because the derivatives of $(\Lambda^a, \Lambda^b)$ are not necessarily bounded.

\begin{proposition}
The FBSDE \eqref{non_lip_fbsde} accepts a unique solution in $\mathbb{S}^2\times\mathbb{S}^2\times \mathbb{H}^2$. Further, the optimal control is monotonic with respect to the initial (inventory) condition, i.e., if $q_1>q_2$, then it holds that
\begin{equation}
  \hat{\delta}_{1,t}^{a}<\hat{\delta}_{2,t}^{a}         \quad\text{and}\quad \hat{\delta}_{1,t}^{b}>\hat{\delta}_{2,t}^{b}
    \nonumber
\end{equation}
for  $t\in[0,T]$ $\mathbb{P}$-a.s., where $(\hat{\delta}_{i,t}^{a}, \hat{\delta}_{i,t}^{b})_{0\leq t\leq T}$ represents the optimal control associated with the initial condition $q_i$ for $i\in\{1, 2\}$.
\label{proof_control_non_lip_fbsde}
\end{proposition}

\begin{proof}
We intend to construct a solution to \eqref{non_lip_fbsde} via the `truncated' version:
\begin{equation}
\left\{
\begin{aligned}
\;& dQ_t  = -a_t\,\Lambda^a\big(\tilde{\delta}^{a*}(Y_t)\big)dt+b_t\,\Lambda^b\big(\tilde{\delta}^{b*}(-Y_t)\big)dt, \\
& dY_t=2\phi_t\,Q_t\,dt+Z_t\,dW_t,\\
& Q_0=q_0,\quad Y_T=-2A\,Q_T.
\label{nonLip_den_FBSDE}
\end{aligned}
\right.
\end{equation}
The idea comes from the solution of the linear case. Recall that the value function $H$ is composed of three terms with different order (see Theorem \ref{verification}): the one of second order is induced from the terminal inventory penalty and is uniformly bounded by the penalty parameters; the one of first order is the estimation of the integral of future order imbalance. We are then motivated to look at two extreme unbalanced situations: the case when the buy market order dominates and the case when the dominance is achieved by the sell market order. Thus, we look at two associate FBSDEs:
\begin{equation}
\left\{
\begin{aligned}
\;& dQ_t  = b_t\,\Lambda^b\big(\tilde{\delta}^{b*}(-Y_t)\big)dt, \\
& dY_t=2\phi_t\,Q_t\,dt+Z_t\,dW_t,\\
& Q_0=q_0,\quad Y_T=-2A\,Q_T;
\label{nonLip_dom_FBSDE_1}
\end{aligned}
\right.
\end{equation}
\begin{equation}
\left\{
\begin{aligned}
\;& dQ_t  = -a_t\,\Lambda^a\big(\tilde{\delta}^{a*}(Y_t)\big)dt, \\
& dY_t=2\phi_t\,Q_t\,dt+Z_t\,dW_t,\\
& Q_0=q_0,\quad Y_T=-2A\,Q_T.
\label{nonLip_dom_FBSDE_2}
\end{aligned}
\right.
\end{equation}
We hope the solutions of these two equations would reveal certain properties of \eqref{nonLip_den_FBSDE}, so that the parameter $\xi$ can be removed safely. Denote by $(\tilde{Q},\tilde{Y},\tilde{Z})$ and $(\hat{Q},\hat{Y},\hat{Z})$ the (unique) solutions of \eqref{nonLip_dom_FBSDE_1} and \eqref{nonLip_dom_FBSDE_2} accordingly. By definition, process $\tilde{Y}$ has the representation
\begin{equation*}
    -\tilde{Y}_t=\mathbb{E}_t\Big[2A\,\tilde{Q}_T+ 2\int_t^T\phi_s\,\tilde{Q}_s\,ds\Big]\geq\mathbb{E}_t\Big[2A\, \tilde{Q}_0+ 2\int_t^T\phi_s\,\tilde{Q}_0\,ds\Big]\geq(-2)\,|q_0|\,(\bar{A}+\bar{\phi}\,T)
\end{equation*}
for all $t\in[0,T]$, where we have used the fact that $\Tilde{Q}$ is non-decreasing. Consequently, a uniform upper bound for $\tilde{Q}$ can be derived by
\begin{equation}
\begin{aligned}
    q_0\leq\tilde{Q}_t&=q_0+\int_0^t b_s\,\Lambda^b\big(\tilde{\delta}^{b*}(-\tilde{Y}_s)\big)\,ds\\
    &\leq q_0+\bar{b}\int_0^T \Lambda^b\Big(\delta^{b*}\big((-2)\,|q_0|\,(\bar{A}+\bar{\phi}\,T)\big)\Big)\,ds\\
    &=q_0+\bar{b}\,T\, \Lambda^b\Big(\delta^{b*}\big((-2)\,|q_0|\,(\bar{A}+\bar{\phi}\,T)\big)\Big).
\end{aligned}
\label{q^tilde_bound}
\nonumber
\end{equation}
In light of this, the lower bound for $\tilde{Y}_t$ can also be found by
\begin{equation}
    -\tilde{Y}_t=\mathbb{E}_t\Big[2A\,\tilde{Q}_T+ 2\int_t^T\phi_s\,\tilde{Q}_s\,ds\Big]\leq 2\,(\bar{A}+\bar{\phi}\,T)\cdot\bigg[|q_0|+\bar{b}\,T\, \Lambda^b\Big(\delta^{b*}\big((-2)\,|q_0|\,(A+\bar{\phi}\,T)\big)\Big)\bigg].
\label{y^tilde_bound}
\end{equation}
We remark that these bounds for $(\tilde{Q},\tilde{Y})$ are independent of the parameter $\xi$, implying that we can already find a solution for this `one-sided' non-Lipschitz FBSDE when $\xi$ is large enough. In a similar fashion, the corresponding bounds for $(\hat{Q},\hat{Y})$ read
\begin{gather}
    \hat{Y}_t=\mathbb{E}_t\Big[-2A\,\hat{Q}_T- 2\int_t^T\phi_s\,\hat{Q}_s\,ds\Big]\geq(-2)\,|q_0|\,(\bar{A}+\bar{\phi}\,T),\nonumber\\
    q_0\geq\hat{Q}_t=q_0-\int_0^t a_s\,\Lambda^a\big(\tilde{\delta}^{a*}(\hat{Y}_s)\big)\,ds\geq q_0-\bar{a}\,T\,\Lambda^a\Big(\delta^{a*}\big((-2)\,|q_0|\,(\bar{A}+\bar{\phi}\,T)\big)\Big),
\label{q^hat_bound}\\
    \hat{Y}_t\leq 2\,(\bar{A}+\bar{\phi}\, T)\cdot\bigg[|q_0|+\bar{a}\,T\,\Lambda^a\Big(\delta^{a*}\big((-2)\,|q_0|\,(\bar{A}+\bar{\phi}\,T)\big)\Big)\bigg],
\end{gather}
holding for all $t\in[0,T]$. Again, such bounds are independent of the parameter $\xi$.

Set $(Q,Y,M)$ as the solution to the original FBSDE \eqref{nonLip_den_FBSDE}, and we then compute the difference of two equations \eqref{nonLip_den_FBSDE} and \eqref{nonLip_dom_FBSDE_1}. Note that
\begin{equation}
\begin{aligned}
    d(\tilde{Q}_t-Q_t)&=b_t\,\Big[\Lambda^b\big(\tilde{\delta}^{b*}(-\tilde{Y}_t)\big)-\Lambda^b\big(\tilde{\delta}^{b*}(-Y_t)\big)\Big]\,dt+a_t\,\Lambda^a\big(\tilde{\delta}^{a*}(Y_t)\big)\,dt\\
    &=b_t\,\Big[\Lambda^b\big(\tilde{\delta}^{b*}(-\tilde{Y}_t)\big)-\Lambda^b\big(\tilde{\delta}^{b*}(-Y_t)\big)\Big]\,dt\\
    &\hspace{2cm}+a_t\,\Big[\Lambda^a\big(\tilde{\delta}^{a*}(Y_t)\big)-\Lambda^a\big(\tilde{\delta}^{a*}(\tilde{Y}_t)\big)\Big]\,dt+a_t\,\Lambda^a\big(\tilde{\delta}^{a*}(\tilde{Y}_t)\big)\,dt\\
    &=\pi_t\,(\tilde{Y}_t-Y_t)\,dt+a_t\,\Lambda^a\big(\tilde{\delta}^{a*}(\tilde{Y}_t)\big)\,dt,
\end{aligned}
\nonumber
\end{equation}
where we let
\begin{equation}
    \pi_t=\frac{b_t\,\Big[\Lambda^b\big(\tilde{\delta}^{b*}(-\tilde{Y}_t)\big)-\Lambda^b\big(\tilde{\delta}^{b*}(-Y_t)\big)\Big]+a_t\,\Big[\Lambda^a\big(\tilde{\delta}^{a*}(Y_t)\big)-\Lambda^a\big(\tilde{\delta}^{a*}(\tilde{Y}_t)\big)\Big]}{\tilde{Y}_t-Y_t}
    \nonumber
\end{equation}
when $\tilde{Y}_t-Y_t\neq0$, and $\pi_t=0$ in the case of $\tilde{Y}_t-Y_t=0$. Due to the monotonicity properties of $(\Lambda^a,\Lambda^b)$ and $(\delta^{a*},\delta^{b*})$, we know that $\pi\in\mathbb{H}^2$ is a bounded non-negative process. Afterwards, the triplet $(\tilde{Q}-Q,\tilde{Y}-Y,\tilde{Z}-Z)$ solves the following linear FBSDE:
\begin{equation}
\left\{
\begin{aligned}
\;& d\mathcal{Q}_t  = \pi_t\,\mathcal{Y}_t dt+a_t\,\Lambda^a\big(\tilde{\delta}^{a*}(\tilde{Y}_t)\big)\,dt, \\
& d\mathcal{Y}_t=2\phi_t\,\mathcal{Q}_t\,dt+\mathcal{Z}_t\,dW_t,\\
& \mathcal{Q}_0=0,\quad \mathcal{Y}_T=-2A\,\mathcal{Q}_T,
\nonumber
\end{aligned}
\right.
\end{equation}
which accepts a unique solution as a simple variation of Lemma \ref{growth_lip_fbsde}. To solve this equation, its linear structure suggests the affine ansatz $\mathcal{Y}_t=\mathscr{A}_t\,\mathcal{Q}_t+\mathscr{B}_t$, where processes $\mathscr{A}$ and $\mathscr{B}$ satisfy two coupled BSDEs:
\begin{gather*}
    d\mathscr{A}_t=\big(-\pi_t\,(\mathscr{A}_t)^2+2\phi_t\big)\,dt+\mathscr{Z}^{\mathscr{A}}_t\,dW_t, \quad\text{with}\quad \mathscr{A}_T=-2A;\\
    d\mathscr{B}_t=-\Big[\pi_t\,\mathscr{A}_t\,\mathscr{B}_t+\mathscr{A}_t\,a_t\,\Lambda^a\big(\tilde{\delta}^{a*}(\tilde{Y}_t)\big)\Big]\,dt+\mathscr{Z}^{\mathscr{B}}_t\,dW_t, \quad\text{with}\quad \mathscr{B}_T=0.
\nonumber
\end{gather*}
We already know the first BSDE accepts a solution $\mathscr{A}$ that is negative and bounded by $2(\bar{A}+\Bar{\phi}\,T)$. In turn, the solution of the second linear BSDE is given explicitly by
\begin{equation} \mathscr{B}_t=\mathbb{E}_t\Big[\int_t^T\mathscr{A}_s\,a_s\,\Lambda^a\big(\tilde{\delta}^{a*}(\tilde{Y}_s)\big)\,e^{\int_t^s\pi_u\,\mathscr{A}_u\,du}\,ds\Big]\leq0.
\nonumber
\end{equation}
Note that the expression inside the conditional expectation is uniformly bounded on $[0,T]$ by some constant independent of $\xi$. Since $\mathcal{Q}_t$ satisfies the simple ODE
\begin{equation*}
    d\mathcal{Q}_t=\pi_t\,(\mathscr{A}_t\,\mathcal{Q}_t+\mathscr{B}_t)\,dt+a_t\,\Lambda^a\big(\tilde{\delta}^{a*}(\tilde{Y}_t)\big)\,dt,
\end{equation*}
the solution can be also given precisely by
\begin{equation*}
    \mathcal{Q}_t=\int_0^t \Big[\pi_s\,\mathscr{B}_s+a_s\,\Lambda^a\big(\tilde{\delta}^{a*}(\tilde{Y}_s)\big)\Big]\,e^{\int_s^t\pi_u\,\mathscr{A}_u\,du}\,ds.
\end{equation*}
By definition, we further deduce
\begin{equation*}
\begin{aligned}
    \tilde{Q}_t-Q_t&=\int_0^t \Big[\pi_s\,\mathscr{B}_s+a_s\,\Lambda^a\big(\tilde{\delta}^{a*}(\tilde{Y}_s)\big)\Big]\,e^{\int_s^t\pi_u\,\mathscr{A}_u\,du}\,ds\\
    &\leq\int_0^t a_s\,\Lambda^a\big(\tilde{\delta}^{a*}(\tilde{Y}_s)\big)\,ds,\\
    Q_t&\geq\tilde{Q}_t-\int_0^t a_s\,\Lambda^a\big(\tilde{\delta}^{a*}(\tilde{Y}_s)\big)\,ds.
\end{aligned}
\end{equation*}
The uniform boundedness of $(\tilde{Q},\tilde{Y})$ implies that $Q$ is bounded below by a constant that is independent of $\xi$. On the other hand, to search for an upper bound for $Q$, we study the relation between $(Q,Y,Z)$ and  $(\hat{Q},\hat{Y},\hat{Z})$. Taking the difference of \eqref{nonLip_den_FBSDE} and \eqref{nonLip_dom_FBSDE_2}, we apply a similar transform to see
\begin{equation}
\begin{aligned}
    d(\hat{Q}_t-Q_t)&=-b_t\,\Lambda^b\big(\tilde{\delta}^{b*}(-Y_t)\big)\,dt-a_t\,\Big[\Lambda^a\big(\tilde{\delta}^{a*}(\hat{Y}_t)\big)-\Lambda^a\big(\tilde{\delta}^{a*}(Y_t)\big)\Big]\,dt\\
    &=b_t\,\Big[\Lambda^b\big(\tilde{\delta}^{b*}(-\hat{Y}_t)\big)-\Lambda^b\big(\tilde{\delta}^{b*}(-Y_t)\big)\Big]\,dt\\
    &\hspace{2cm}+a_t\,\Big[\Lambda^a\big(\tilde{\delta}^{a*}(Y_t)\big)-\Lambda^a\big(\tilde{\delta}^{a*}(\hat{Y}_t)\big)\Big]\,dt-b_t\,\Lambda^b\big(\tilde{\delta}^{b*}(\hat{Y}_t)\big)\,dt\\
    &=\kappa_t\,(\hat{Y}_t-Y_t)\,dt-b_t\,\Lambda^b\big(\tilde{\delta}^{b*}(\hat{Y}_t)\big)\,dt,
\end{aligned}
\nonumber
\end{equation}
where we let
\begin{equation*}
    \kappa_t=\frac{b_t\,\Big[\Lambda^b\big(\tilde{\delta}^{b*}(-\hat{Y}_t)\big)-\Lambda^b\big(\tilde{\delta}^{b*}(-Y_t)\big)\Big]+a_t\,\Big[\Lambda^a\big(\tilde{\delta}^{a*}(Y_t)\big)-\Lambda^a\big(\tilde{\delta}^{a*}(\hat{Y}_t)\big)\Big]}{\hat{Y}_t-Y_t}
\end{equation*}
when $\hat{Y}_t-Y_t\neq0$, and $\kappa_t=0$ in the case of $\hat{Y}_t-Y_t=0$. One can see that $\kappa\in\mathbb{H}^2$ is non-negative and bounded as well. Similarly, the triplet $(\hat{Q}_t-Q_t,\hat{Y}_t-Y_t,\hat{Z}_t-Z_t)$ solves the following (well-posed) linear FBSDE:
\begin{equation}
\left\{
\begin{aligned}
\;& d\mathcal{Q}_t  = \kappa_t\,\mathcal{Y}_t dt-b_t\,\Lambda^b\big(\tilde{\delta}^{b*}(\hat{Y}_t)\big)\,dt, \\
& d\mathcal{Y}_t=2\phi_t\,\mathcal{Q}_t\,dt+\mathcal{Z}_t\,dW_t,\\
& \mathcal{Q}_0=0,\quad \mathcal{Y}_T=-2A\,\mathcal{Q}_T.
\nonumber
\end{aligned}
\right.
\end{equation}
The linear structure again advocates the affine ansatz $\mathcal{Y}_t=\mathscr{P}_t\,\mathcal{Q}_t+\mathscr{H}_t$, with processes $\mathscr{P}$ and $\mathscr{H}$ solving two coupled BSDEs
\begin{gather*}
    d\mathscr{P}_t=(-\kappa_t\,\mathscr{P}_t^2+2\phi_t)\,dt+\mathscr{Z}^{\mathscr{P}}_t\,dW_t, \quad\text{with}\quad \mathscr{P}_T=-2A;\\
    d\mathscr{H}_t=-\Big[\kappa_t\,\mathscr{P}_t\,\mathscr{H}_t-\mathscr{P}_t\,b_t\,\Lambda^b\big(\tilde{\delta}^{b*}(\hat{Y}_t)\big)\Big]\,dt+\mathscr{Z}^{\mathscr{H}}_t\,dW_t, \quad\text{with}\quad \mathscr{H}_T=0.
\nonumber
\end{gather*}
We know the first BSDE accepts a solution $\mathscr{P}\in\mathbb{H}^2$ that is negative and bounded. In turn, the solution of the second linear BSDE is given explicitly by
\begin{equation}
    \mathscr{H}_t=-\,\mathbb{E}_t\Big[\int_t^T\mathscr{P}_s\,b_s\,\Lambda^b\big(\tilde{\delta}^{b*}(\hat{Y}_s)\big)\,e^{\int_t^s\kappa_u\,\mathscr{P}_u\,du}\,ds\Big]\geq 0.
\nonumber
\end{equation}
Note that $\mathscr{H}$ is also uniformly bounded by some constant being independent of $\xi$. As $\mathcal{Q}$ can be represented by
\begin{equation*}
    \mathcal{Q}_t=\int_0^t \Big[\kappa_s\,\mathscr{H}_s-b_s\,\Lambda^b\big(\tilde{\delta}^{b*}(\hat{Y}_s)\big)\Big]\,e^{\int_s^t\kappa_u\,\mathscr{P}_u\,du}\,ds,
\end{equation*}
by definition we observe
\begin{equation*}
\begin{aligned}
    \hat{Q}_t-Q_t&=\int_0^t \Big[\kappa_s\,\mathscr{H}_s-b_s\,\Lambda^b\big(\tilde{\delta}^{b*}(\hat{Y}_s)\big)\Big]\,e^{\int_s^t\kappa_u\,\mathscr{P}_u\,du}\,ds\\
    &\geq-\int_0^t b_s\,\Lambda^b\big(\tilde{\delta}^{b*}(\hat{Y}_s)\big)\,ds,\\
    Q_t&\leq\hat{Q}_t+\int_0^t b_s\,\Lambda^b\big(\tilde{\delta}^{b*}(\hat{Y}_s)\big)\,ds.
\end{aligned}
\end{equation*}
The uniform boundedness of $(\hat{Q},\hat{Y})$ suggests that $Q$ is bounded above by a constant that is independent of $\xi$. While now we know $Q$ is uniformly lower bounded, it leads to the uniform boundedness of $Y$ since
\begin{equation*}
    Y_t=\,\mathbb{E}_t\Big[-2A\,Q_T- 2\int_t^T\phi_s\, Q_s\,ds\Big].
\end{equation*}
In consequence, by picking $\xi$ large enough, the solution $(Q,Y,Z)$ to the Lipschitz FBSDE \eqref{nonLip_den_FBSDE} also solves the non-Lipschitz \eqref{non_lip_fbsde}. The uniqueness again follows from the continuation method in Lemma \ref{growth_lip_fbsde}.

With respect to the monotonicity, it suffices to apply Theorem \ref{control_mono} and the fact that the truncation has been removed. 
\end{proof}

\section{Non-Markovian Order Flows: Unbounded Coefficients} \label{section_4}
\noindent In the previous two cases, the coefficients---order flows $(a, b)$ and penalisation parameters $(\phi, A)$---are all bounded. Since we have mentioned the possibility of unbounded order flows in Remark \ref{unbound_order_flow}, this section is devoted to such situations. We present the result in the case of linear intensity function, i.e., $\Lambda(\delta)=\zeta-\gamma\,\delta$ for some constants $\zeta, \gamma>0$. 

We start by introducing proper spaces for the order flows $(a, b)$ and penalisation parameters $(\phi, A)$.

\begin{assumption}
Let $A\in L^{8\mathscr{D}}(\mathcal{F}_T, \mathbb{R}_+)$, $a$, $b\in\mathbb{S}^{8\mathscr{D}}$, and $\phi\in\mathbb{H}^{8\mathscr{D}}$, where $\mathscr{D}\in\mathbb{N}$ will be specified later.
\end{assumption}

\noindent Since the intensity function is $\Lambda(\delta)=\zeta-\gamma\,\delta$, we recall the inventory and cash processes of the market maker as follows:
\begin{gather*}
        Q_t = q_0-\int_0^t a_u(\zeta-\gamma \delta^a_u)\,du+\int_0^t b_u(\zeta-\gamma \delta^b_u)\,du,\\
        \nonumber
    X_t = x_0+\int_0^t a_u(\zeta-\gamma \delta^a_u)\,(S_u+\delta_u^a)\,du-\int_0^t b_u(\zeta-\gamma \delta^b_u)\,(S_u-\delta_u^b)\,du,
    \nonumber
\end{gather*}
where $\delta^a,\delta^b\in\mathbb{H}^8$ represent the control. Note that here we seek strategies in the space $\mathbb{H}^8$ to make sure the associated inventory $Q$ is still in $\mathbb{H}^4$ (and indeed in $\mathbb{S}^4$). Hence, the following objective functional is well-defined and can be simplified as before:
\begin{equation*}
\begin{aligned}
    J(\boldsymbol{\delta}):&=\mathbb{E}\big[X_T+S_T\,Q_T-\int_0^T \phi_t\,(Q_t)^2\,dt-A\,(Q_T)^2 \big]\\
   &=\mathbb{E}\Big[\int_0^T\delta_t^a\,a_t(\zeta-\gamma\delta_t^a)\,dt+\int_0^T\delta_t^b\,b_t(\zeta-\gamma\delta_t^b)\,dt-\int_0^T \phi_t\,(Q_t)^2\,dt-A\,(Q_T)^2 \Big],
\end{aligned}
\end{equation*}
by controlling $\boldsymbol{\delta}\in\mathbb{H}^8\times\mathbb{H}^8$. Indeed, it suffices to observe
\begin{equation*}
\begin{aligned}
    \big|\int_0^T\delta_t^a\,a_t(\zeta-\gamma\delta_t^a)\,dt\big|&\leq C\,\sup_{0\leq t\leq T}|a_t|^2+C\,\big|\int_0^T(\delta_t^a)^2\,dt\big|+C\,\big|\int_0^T(\delta_t^a)^2\,dt\big|^2,\\
    \big|\int_0^T \phi_t\,(Q_t)^2\,dt\big|&\leq C\,\sup_{0\leq t\leq T}|Q_t|^4+C\,\big|\int_0^T(\phi_t)^2\,dt\big|
\end{aligned}
\end{equation*}
to justify the integrability. Since the convex nature of the objective functional $J$ is revealed in the proof of Theorem \ref{verification}, the convex-analytic method is used to characterize the (unique) optimal control as the solution of a FBSDE.

\begin{theorem}
\label{unbound_linear_fbsde}
    The functional $\boldsymbol{\delta}\hookrightarrow J(\boldsymbol{\delta})$ is G\^ateaux differentiable and the derivative in the direction of $\boldsymbol{\beta}:=(\beta^a, \beta^b)\in\mathbb{H}^8\times\mathbb{H}^8$ reads
\begin{equation}
\begin{aligned}
    \frac{d}{d\epsilon}&J(\boldsymbol{\delta}+\epsilon\boldsymbol{\beta})\big|_{\epsilon=0}:=\lim_{\epsilon\searrow0}\frac{1}{\epsilon}\,\big[J(\boldsymbol{\delta}+\epsilon\boldsymbol{\beta})-J(\boldsymbol{\delta})\big]\\
    &=\mathbb{E}\Big[\int_0^T\big(a_t\,\beta_t^a\,(\zeta-2\gamma\,\delta_t^a)+b_t\,\beta_t^b\,(\zeta-2\gamma\,\delta_t^b)\big)\,dt-2\,\int_0^T\phi_t\,V_t\,Q_t\,dt-2\,A\,V_T\,Q_T\Big],
    \label{liner_gat_deriv}
\end{aligned}
\end{equation}
where $V_t=\gamma\int_0^t(a_s\,\beta_s^a-b_s\,\beta_s^b)\,ds$. A control $(\delta^{*,a},\delta^{*,b})\in\mathbb{H}^8\times\mathbb{H}^8$ maximises the functional $J$ if it can be represented by
\begin{equation}
    \delta_t^{b,*}=\frac{\zeta}{2\gamma}-\frac{1}{2}Y_t, \text{\; and \;} \delta_t^{a,*}=\frac{\zeta}{2\gamma}+\frac{1}{2}Y_t,
    \label{optimal_feedback_unbound}
\end{equation}
where $Y$ solves the following FBSDE:
\begin{equation*}
\left\{
\begin{aligned}
\;& dQ_t  = \zeta\,(b_t-a_t)\,dt\,/\,2+\gamma\,(a_t+b_t)\,Y_t\,dt\,/\,2, \\
& dY_t=2\,\phi_t\,Q_t\,dt+Z_t\,dW_t,\\
& Q_0=q_0,\quad Y_T=-2A\,Q_T.
\end{aligned}
\right.
\end{equation*}
The representation \eqref{optimal_feedback_unbound} is also necessary when $a,b>0$.
\end{theorem}

\begin{proof}
Fix an admissible control $\boldsymbol{\delta}=(\delta^a, \delta^b)\in \mathbb{H}^8\times\mathbb{H}^8$ and denote by $Q = Q^{\boldsymbol{\delta}}\in\mathbb{S}^2$ the corresponding controlled inventory. Next, consider $\boldsymbol{\beta}=(\beta^a, \beta^b)\in \mathbb{H}^8\times\mathbb{H}^8$---that is also uniformly bounded---as the direction in which we are computing the
G\^ateaux derivative of $J$. For each $\epsilon > 0$,
we consider the admissible control $\boldsymbol{\delta}^\epsilon\in\mathbb{H}^8\times\mathbb{H}^8$ defined by $\boldsymbol{\delta}^\epsilon_t = \boldsymbol{\delta}_t + \epsilon\,\boldsymbol{\beta}_t$, and the corresponding controlled state is denoted by $Q^\epsilon:= Q^{\boldsymbol{\delta}^\epsilon}\in\mathbb{S}^2$. Note that $V\in\mathbb{S}^{8\mathscr{D}}$. Direct calculations yield
\begin{equation}
\begin{aligned}
\label{compute_deriva}
    \frac{1}{\epsilon}\,&\big[J(\boldsymbol{\delta}+\epsilon\boldsymbol{\beta})-J(\boldsymbol{\delta})\big]\\
    &=\mathbb{E}\Big[\int_0^Ta_t\,\beta_t^a\,(\zeta-2\,\gamma\,\delta_t^a-\epsilon\,\gamma\,\beta_t^a)\,dt\Big]+\mathbb{E}\Big[\int_0^Tb_t\,\beta_t^b\,(\zeta-2\,\gamma\,\delta_t^b-\epsilon\,\gamma\,\beta_t^b)\,dt\Big]\\
    &\hspace{1cm}-\mathbb{E}\Big[\int_0^T \phi_t\,(Q_t^\epsilon+Q_t)\,V_t\,dt\Big]-\mathbb{E}\Big[A\,(Q_T^\epsilon+Q_T)\,V_T\Big].
\end{aligned}
\end{equation}
We look at the first term on the right hand side of \eqref{compute_deriva}. To perform the limiting procedure, note that
\begin{equation*}
\begin{aligned}
    a_t\,\beta_t^a\,(\zeta-2\,\gamma\,\delta_t^a-\epsilon\,\gamma\,\beta_t^a)\,dt & \leq a_t\,|\beta_t^a|\,\big(\zeta+2\,\gamma\,|\delta_t^a|+\gamma\,|\beta_t^a|\big),\\
\end{aligned}
\end{equation*}
the right hand side of which is $d\mathbb{P}\times dt$-integrable. Therefore, it follows from the dominated convergence theorem that
\begin{equation*}
    \lim_{\epsilon\to 0}\,\mathbb{E}\Big[\int_0^Ta_t\,\beta_t^a\,(\zeta-2\,\gamma\,\delta_t^a-\epsilon\,\gamma\,\beta_t^a)\,dt\Big]=\mathbb{E}\Big[\int_0^Ta_t\,\beta_t^a\,(\zeta-2\,\gamma\,\delta_t^a)\,dt\Big].
\end{equation*}
The second term of \eqref{compute_deriva} can be verified in the same way. With respect to third term, we are aware of
\begin{equation*}
\begin{aligned}
    \int_0^T \phi_t\,(Q_t^\epsilon+Q_t)\,V_t\,dt&\leq C\int_0^T \phi_t^2\,\big(|Q_t^\epsilon|+|Q_t|\big)\,dt+C\int_0^T V_t^2\,\big(|Q_t^\epsilon|+|Q_t|\big)\,dt,\\
    &\leq C\sup_{t\in[0,T]}\big(|Q_t^\epsilon|^2+|Q_t|^2\big)+C\,\Big(\int_0^T \phi_t^2\,dt\Big)^2+C\,\Big(\int_0^T V_t^2\,dt\Big)^2\\
    &\leq C\sup_{t\in[0,T]}|Q_t|^2+C\,\Big(\int_0^T (a_t+b_t)\,dt\Big)^2+C\,\Big(\int_0^T \phi_t^2\,dt\Big)^2\\
    &\hspace{2.8cm}+C\,\Big(\int_0^T V_t^2\,dt\Big)^2\\
    \phi_t\,(Q_t^\epsilon+Q_t)\,V_t\,&\leq C\,|Q_t|^2+C\,\Big(\int_0^T (a_t+b_t)\,dt\Big)^2+C\,\phi_t^2+C\,V_t^2\\
\end{aligned}
\end{equation*}
and the integrability on the right hand sides, and hence can conclude
\begin{equation*}
    \lim_{\epsilon\to 0}\,\mathbb{E}\Big[\int_0^T \phi_t\,(Q_t^\epsilon+Q_t)\,V_t\,dt\Big]=2\,\mathbb{E}\Big[\int_0^T \phi_t\,Q_t\,V_t\,dt\Big].
\end{equation*}
Since the final term can be proved similarly, the equation \eqref{compute_deriva} is verified and an integration by parts further yields
\begin{equation}
\begin{aligned}
    \frac{d}{d\epsilon}J(\boldsymbol{\delta}+\epsilon&\boldsymbol{\beta})\big|_{\epsilon=0}\\
    &=\mathbb{E}\bigg\{\int_0^T\Big[a_t\,\beta_t^a\,\big(\zeta-2\gamma\,\delta_t^a-2\gamma\int_t^T\phi_s\,Q_s\,ds-2\gamma\, A \, Q_T\big)\Big]\,dt\bigg\}\\
    &\hspace{2cm}+\mathbb{E}\bigg\{\int_0^T\Big[b_t\,\beta_t^b\,\big(\zeta-2\gamma\,\delta_t^b-2\gamma\int_t^T\phi_s\,Q_s\,ds+2\gamma\, A\, Q_T \big)\Big]\,dt\bigg\}.
    \label{derive_fbsde}
\end{aligned}
\end{equation}
Since the functional $J$ is concave, a control $\boldsymbol{\delta}^*=(\delta^{*,a}, \delta^{*,b})\in\mathbb{H}^8\times\mathbb{H}^8$ is optimal if and only if \eqref{derive_fbsde} is equal to $0$ for any bounded $\boldsymbol{\beta}\in\mathbb{H}^8\times\mathbb{H}^8$. With the law of total expectation, this infers that
\begin{equation*}
\begin{aligned}
    \gamma\,Y_t=\zeta-2\gamma\,\delta_t^{*,a}&=\mathbb{E}_t\Big[2\gamma\int_t^T\phi_s\,Q_s\,ds+2\gamma\, A \, Q_T\Big],\\
    \gamma\,Y_t=2\gamma\,\delta_t^{*,b}-\zeta&=\mathbb{E}_t\Big[2\gamma\int_t^T\phi_s\,Q_s\,ds+2\gamma\, A\, Q_T\Big]
\end{aligned}
\end{equation*}
$d\mathbb{P} \times dt$-almost everywhere. The FBSDE can be obtained by `differentiating' $Y$. 
\end{proof}

\noindent The remaining of this section is devoted the well-posedness of the FBSDE. Before looking at the FBSDE associated with the stochastic maximum principle, we first investigate the cornerstone of its well-posedness theory---a quadratic BSDE: 
\begin{equation}
    dY_t=\big(-\gamma\,(a_t+b_t)\,(Y_t)^2+2\,\phi_t\big)\,dt+Z_t\,dW_t, \text{\quad with \quad} Y_T=-2\,A,
    \label{quad_bsde_unbound_coeff}
\end{equation}
the coefficients $a, b, \phi, A$ of which are not bounded. Note that the $(Y,Z)$ here is different from the one in the FBSDE, and here we only focus on the existence of solutions. We localize this equation via the truncation:
\begin{equation}
    dY_t=\big(-\gamma\,\big((a_t+b_t)\wedge m\big)\,(Y_t)^2+2\,(\phi_t\wedge n)\big)\,dt+Z_t\,dW_t, \text{\quad with \quad} Y_T=-2\,(A\wedge n),
    \label{trun_bsde_unbound_coeff}
\end{equation}
where $m, n\in\mathbb{N}$. First, we try to remove the constant $n$.

\begin{lemma}
The BSDE
\begin{equation}
    dY_t=\big(-\gamma\,\big((a_t+b_t)\wedge m\big)\,(Y_t)^2+2\,\phi_t\big)\,dt+Z_t\,dW_t, \text{\quad with \quad} Y_T=-2\,A,
        \label{half_trun_bsde_unbound_coeff}
\end{equation}
accepts a solution $(Y, Z)\in\mathbb{S}^{8\mathscr{D}}\times\mathbb{H}^2$.
\end{lemma}    

\begin{proof}
Recall that \eqref{trun_bsde_unbound_coeff} is well-posed for any pair $(m, n)$ by Lemma \ref{growth_lip_fbsde}. Fixing any $m\in\mathbb{N}$, equation \eqref{trun_bsde_unbound_coeff} accepts a unique solution $(Y^{(n)}, Z^{(n)})\in\mathbb{S}^2\times\mathbb{H}^2$ for all $n\in\mathbb{N}$. Because $|Y_t^{(n)}|\leq 2n+2nT$, the standard comparison theorem for Lipschitz BSDEs implies
\begin{equation*}
    0 \geq Y_t^{(n)} \geq Y_t^{(n+1)}, \qquad 0 \leq t \leq T, \qquad \mathbb{P}\text{-a.s.}
\end{equation*}
Consequently, for almost every $\omega\in\Omega$, we can define the limit as follows:
\begin{equation*}
    U_t(\omega):=\lim_{n\to\infty}Y_t^{(n)}(\omega),\qquad 0 \leq t \leq T.
\end{equation*}
Note that $-U_t(\omega)$ may be infinity for some $t$ and $\omega$. Thanks to the monotonicity of $(Y^{(n)})_{n\in\mathbb{N}}$, we observe that
\begin{equation*}
    \lim_{n\to\infty}|Y_t^{(n)}|^{8\mathscr{D}}\leq \lim_{n\to\infty}\,\sup_{0\leq u\leq T}|Y_u^{(n)}|^{8\mathscr{D}}, \qquad 0 \leq t \leq T,
\end{equation*}
the right hand side of which is well-defined (possibly infinite). It further reveals the integrability of $U$ as follows
\begin{equation}
\begin{aligned}
    \mathbb{E}\Big[\sup_{0\leq t\leq T}|U_t|^{8\mathscr{D}}\Big]&\leq \lim_{n\to\infty}\,\mathbb{E}\Big[\sup_{0\leq t\leq T}|Y_t^{(n)}|^{8\mathscr{D}}\Big]\\
    &= \lim_{n\to\infty}\,\mathbb{E}\Big[\sup_{0\leq u\leq T}\Big|\,\mathbb{E}_u\big[-2\,(A\wedge n)-2\int_t^T (\phi_s\wedge n)\,ds\\
    &\hspace{3.5cm}+\int_t^T\gamma\,\big((a_u+b_u)\wedge m\big)\,(Y_u)^2\,du\big]\,\Big|^{8\mathscr{D}}\Big]\\
    &\leq \lim_{n\to\infty}\,\mathbb{E}\Big[\sup_{0\leq u\leq T}\mathbb{E}_u\big[2\,(A\wedge n)+2\int_t^T (\phi_s\wedge n)\,ds\big]^{8\mathscr{D}}\Big]\\
    &\leq \lim_{n\to\infty}\,\mathbb{E}\Big[\sup_{0\leq u\leq T}\mathbb{E}_u\big[2\,A+2\int_0^T \phi_s\,ds\big]^{8\mathscr{D}}\Big]\\
    &\leq C\,\mathbb{E}\Big[\big(2\,A+2\int_0^T \phi_s\,ds\big)^{8\mathscr{D}}\Big]\\
    &<\infty,
    \label{integ_Y}
\end{aligned}
\end{equation}
where we have applied the monotone convergence theorem, the fact $Y_t^{(n)}\leq 0$, and the Doob's inequality. For any $u\in[0,T]$, since two random variables
\begin{equation*}
    \mathbb{E}_u\Big[2\,A+2\int_u^T \phi_s\,ds\Big] \text{\quad and \quad} \mathbb{E}_u\Big[\int_u^T \big((a_s+b_s)\wedge m\big)\,(U_s)^2\,ds\Big]
\end{equation*}
are both integrable, then they are almost surely finite and one can use the monotone convergence theorem again to see
\begin{equation*}
\begin{aligned}
    U_t&=\lim_{n\to\infty}Y_t^{(n)}\\
    &=\lim_{n\to\infty}\mathbb{E}_t\Big[-2\,(A\wedge n)-2\int_t^T (\phi_s\wedge n)\,ds+\gamma\int_t^T \big((a_s+b_s)\wedge m\big)\,\big(Y_s^{(n)}\big)^2\,ds\Big]\\
    &=\mathbb{E}_t\Big[-2\,A-2\int_t^T \phi_s\,ds+\gamma\int_t^T \big((a_s+b_s)\wedge m\big)\,(U_s)^2\,ds\Big].
\end{aligned}
\end{equation*}
While the random variable
\begin{equation*}
    -2\,A-2\int_0^T \phi_s\,ds+\gamma\int_0^T \big((a_s+b_s)\wedge m\big)\,(Y_s)^2\,ds
\end{equation*}
is square integrable, we finally conclude that $U$ satisfies the BSDE
\begin{equation*}
    dU_t=\big(-\gamma\,\big((a_t+b_t)\wedge m\big)\,(U_t)^2+2\,\phi_t\big)\,dt+Z_t^U\,dW_t, \text{\quad with \quad} U_T=-2\,A,
\end{equation*}
where $Z^U\in\mathbb{H}^2$ is obtained via the martingale representation theorem. 
\end{proof}

\noindent Note that the solution of \eqref{half_trun_bsde_unbound_coeff} may not be bounded, and hence the standard comparison theorem for Lipschitz BSDEs can not be applied. We summarize the comparison result in the following lemma.

\begin{lemma}
For each $m\in\mathbb{N}$, denote by $(Y^{(m)}, Z^{(m)})$ the solution of equation \eqref{half_trun_bsde_unbound_coeff}. Then, $(Y^{(m)})_{m\in\mathbb{N}}$ is non-decreasing with respect to $m$.
\end{lemma}

\begin{proof}
Define $\Delta Y:=Y^{(m+1)}-Y^{(m)}$ and $\Delta Z:=Z^{(m+1)}-Z^{(m)}$. The proof starts with taking the difference of two BSDEs:
\begin{equation*}
\begin{aligned}
    d\Delta Y_t&=\Big(-\gamma\,\big((a_t+b_t)\wedge (m+1)\big)\,(Y_t^{(m+1)})^2+\gamma\,\big((a_t+b_t)\wedge (m+1)\big)\,(Y_t^{(m)})^2\\
    &\hspace{1cm}-\gamma\,\big((a_t+b_t)\wedge (m+1)\big)\,(Y_t^{(m)})^2+\gamma\,\big((a_t+b_t)\wedge m\big)\,(Y_t^{(m)})^2\Big)\,dt+\Delta Z_t\,dW_t\\
    &=\Big(-\gamma\,\big((a_t+b_t)\wedge (m+1)\big)\,(Y_t^{(m+1)}+Y_t^{(m)})\,\Delta Y_t\\
    &\hspace{1cm}-\gamma\, (Y_t^{(m)})^2\,\big[(a_t+b_t)\wedge (m+1)-(a_t+b_t)\wedge m\big]\Big)\,dt+\Delta Z_t\,dW_t.
\end{aligned}
\end{equation*}
Set $\iota_t=\gamma\,\big((a_t+b_t)\wedge (m+1)\big)\,(Y_t^{(m+1)}+Y_t^{(m)})$, and observe that $\iota_t\leq 0$ since any $Y^{(m)}$ is non-positive. The integration factor gives
\begin{equation*}
    \Delta Y_t e^{\int_0^t \iota_u\,du}=\gamma\,\mathbb{E}_t\Big[\int_t^Te^{\int_0^s\iota_u\,du}\,(Y_s^{(m)})^2\,\big[(a_s+b_s)\wedge (m+1)-(a_s+b_s)\wedge m\big]\,ds\Big]\geq 0,
\end{equation*}
where we have used the fact that $\Delta Y_T=0$. 
\end{proof}

\noindent Given the comparison result, we can finally remove the constant $m$ as well.

\begin{theorem}
The BSDE \eqref{quad_bsde_unbound_coeff} accepts a solution $(Y, Z)\in\mathbb{S}^{8\mathscr{D}}\times\mathbb{H}^2$.
\label{well_posed_quad_bsde}
\end{theorem}

\begin{proof}
\noindent Denote by $(Y^{(m)}, Z^{(m)})\in\mathbb{S}^{8\mathscr{D}}\times\mathbb{H}^2$ the solution of equation \eqref{half_trun_bsde_unbound_coeff} for each $m\in\mathbb{N}$. Given the monotonicity of $(Y^{(m)})_{m\in\mathbb{N}}$, for almost every $\omega\in\Omega$, let us define the limit similarly as
\begin{equation*}
    Y_t(\omega):=\lim_{m\to\infty}Y_t^{(m)}(\omega)\leq 0,\qquad 0 \leq t \leq T.
\end{equation*}
Its integrability can be directly seen from
\begin{equation*}
    \mathbb{E}\Big[\sup_{0\leq t\leq T}|Y_t|^{8\mathscr{D}}\Big]\leq \mathbb{E}\Big[\sup_{0\leq t\leq T}|Y_t^{(1)}|^{8\mathscr{D}}\Big]<\infty,
\end{equation*}
since $Y^{(m)}\leq Y^{(m+1)}\leq0$ for any $m\in\mathbb{N}$. Thanks to the integrability of the random variable
\begin{equation*}
    \mathbb{E}_u\Big[\int_u^T (a_s+b_s)\,\big(Y_s^{(1)}\big)^2\,ds\Big],
\end{equation*}
the dominated convergence theorem implies
\begin{equation*}
\begin{aligned}
    Y_t=\lim_{m\to\infty}Y_t^{(m)}&=\lim_{m\to\infty}\mathbb{E}_t\Big[-2\,A-2\int_t^T \phi_s\,ds+\gamma\int_t^T \big((a_s+b_s)\wedge m\big)\,\big(Y_s^{(m)}\big)^2\,ds\Big]\\
    &=\mathbb{E}_t\Big[-2\,A-2\int_t^T \phi_s\,ds+\gamma\int_t^T (a_s+b_s)\,(Y_s)^2\,ds\Big].
\end{aligned}
\end{equation*}
Finally, the square integrability of the random variable
\begin{equation*}
    -2\,A-2\int_0^T \phi_s\,ds+\gamma\int_0^T (a_s+b_s)\,(Y_s)^2\,ds
\end{equation*}
helps us conclude that $Y$ solves the BSDE
\begin{equation*}
    dY_t=\big(-\gamma\,(a_t+b_t)\,(Y_t)^2+2\,\phi_t\big)\,dt+Z_t\,dW_t, \text{\quad with \quad} Y_T=-2\,A.
\end{equation*}
Note that $Z$ is again obtained via the martingale representation theorem. 
\end{proof}

Subsequently, we look at the FBSDE associated with the stochastic maximum principle
\begin{equation}
\left\{
\begin{aligned}
\;& dQ_t  = \zeta\,(b_t-a_t)\,dt\,/\,2+\gamma\,(a_t+b_t)\,Y_t\,dt\,/\,2, \\
& dY_t=2\,\phi_t\,Q_t\,dt+Z_t\,dW_t,\\
& Q_0=q_0,\quad Y_T=-2A\,Q_T.
\label{linear_fbsde_unbound}
\end{aligned}
\right.
\end{equation}
The next theorem presents the well-posedness result of the main FBSDE, and then specifies the value of $\mathscr{D}$ in line with the admissibility of the control.

\begin{theorem}
The FBSDE \eqref{linear_fbsde_unbound} accepts a unique solution $(Q, Y, Z)\in \mathbb{S}^{2\mathscr{D}}\times\mathbb{S}^{\mathscr{D}}\times\mathbb{H}^{2}$. It suffices to let $\mathscr{D}=8$ to ensure the resulting control is admissible.
\end{theorem}

\begin{proof}
The linear structure of \eqref{linear_fbsde_unbound} suggests the affine ansatz
\begin{equation}
    Y_t = \mathscr{P}_t\,Q_t+\mathscr{H}_t,
    \label{unbound_affine_ansatz}
\end{equation}
where processes $\mathscr{P}$ and $\mathscr{H}$ solve the BSDEs
\begin{equation*}
    d\mathscr{P}_t=\big(-\gamma\,(a_t+b_t)\,(\mathscr{P}_t)^2\,/\,2+2\,\phi_t\big)\,dt+Z_t^1\,dW_t, \text{\quad with \quad} \mathscr{P}_T=-2\,A,
\end{equation*}
\begin{equation*}
    d\mathscr{H}_t=\big(-\gamma\,(a_t+b_t)\,\mathscr{P}_t\,\mathscr{H}_t-\zeta\,(b_t-a_t)\,\mathscr{P}_t\big)\,dt\,/\,2+Z_t^2\,dW_t, \text{\quad with \quad} \mathscr{H}_T=0.
\end{equation*}
The well-posedness of $\mathscr{P}\in \mathbb{S}^{8\mathscr{D}}$ is guaranteed by Theorem \ref{well_posed_quad_bsde}, and the linear BSDE of $\mathscr{H}$ can be solved explicitly by
\begin{equation*}
\begin{aligned}
    \mathscr{H}_t=\mathbb{E}_t\Big[\int_t^T\zeta\,(b_s-a_s)\,\mathscr{P}_s\,e^{\int_s^T\gamma\,(a_u+b_u)\,\mathscr{P}_u\,du}\,ds\Big].
\end{aligned}
\end{equation*}
Direct estimations and Doob’s inequality then yield
\begin{equation*}
\begin{aligned}
    |\mathscr{H}_t| &\leq C\sup_{0\leq u\leq T}\mathbb{E}_u\Big[\int_0^T(a_s-b_s)^2\,ds\Big]+C\sup_{0\leq u\leq T}\mathbb{E}_u\Big[\int_0^T(\mathscr{P}_s)^2\,ds\Big],\\
    \mathbb{E}\big[\sup_{0\leq t\leq T}|\mathscr{H}_t|^{4\mathscr{D}}]&\leq C\,\mathbb{E}\Big[\big(\int_0^T(a_s-b_s)^2\,ds\big)^{4\mathscr{D}}\Big]+C\,\mathbb{E}\Big[\big(\int_0^T(\mathscr{P}_s)^2\,ds\big)^{4\mathscr{D}}\Big]<\infty,
\end{aligned}
\end{equation*}
from which one can see $\mathscr{H}\in\mathbb{S}^{4\mathscr{D}}$. On the other hand, we can calculate the corresponding inventory
\begin{equation*} Q_t=q_0\,e^{\gamma\,\int_0^t(a_s+b_s)\,\mathscr{P}_s\,ds\,/\,2}+\int_0^t e^{\gamma\,\int_s^t(a_u+b_u)\,\mathscr{P}_u\,du\,/\,2}\big[\zeta\,(b_s-a_s)\,/\,2+\gamma\,(a_s+b_s)\,\mathscr{H}_s\,/\,2\big]\,ds,
\end{equation*}
and deduce that $Q\in\mathbb{S}^{2\mathscr{D}}$ from the following estimation:
\begin{equation*}
\begin{aligned}
    \mathbb{E}\big[\sup_{0\leq t\leq T}|Q_t|^{2\mathscr{D}}\big]&\leq C\, q_0^{2\mathscr{D}}+C\,\mathbb{E}\Big[\big(\int_0^T(b_t-a_t)^2\,dt\big)^{\mathscr{D}}\Big]+C\,\mathbb{E}\Big[\big(\int_0^T(a_t+b_t)\,|\mathscr{H}_t|\,dt\big)^{2\mathscr{D}}\Big]\\
    &\leq C\, q_0^{2\mathscr{D}}+C\,\mathbb{E}\Big[\big(\int_0^T(b_t-a_t)^2\,dt\big)^{\mathscr{D}}\Big]+C\,\mathbb{E}\Big[\big(\int_0^T(a_t+b_t)^2\,dt\big)^{2\mathscr{D}}\Big]\\
    &\hspace{3cm}+C\,\mathbb{E}\Big[\big(\int_0^T|\mathscr{H}_t|^2\,dt\big)^{{2\mathscr{D}}}\Big]\\
    &<\infty.
\end{aligned}
\end{equation*}
Since $Y_t = \mathscr{P}_t\,Q_t+\mathscr{H}_t$ and
\begin{equation*}
    \mathbb{E}\big[\sup_{0\leq t\leq T} |\mathscr{P}_t\,Q_t|^\mathscr{D}\big]\leq C\,\mathbb{E}\big[\sup_{0\leq t\leq T} |\mathscr{P}_t|^{2\mathscr{D}}\big]+C\,\mathbb{E}\big[\sup_{0\leq t\leq T} |Q_t|^{2\mathscr{D}}\big]<\infty,
\end{equation*}
it yields $Y\in\mathbb{S}^{\mathscr{D}}$. The uniqueness of solution again follows from the continuation argument in Lemma \ref{growth_lip_fbsde}.

Recalling that the optimal control is given by \eqref{optimal_feedback_unbound} and the admissibility requires that the control is in $\mathbb{H}^8$, it suffices to take $\mathscr{D}=8$. 
\end{proof}

\section{Applications}
\label{implement}
\noindent With the theoretical foundations of the macroscopic market making model established, this section presents several examples to illustrate its applications. We begin by discussing the connection between the new model and the seminal Avellaneda-Stoikov model. Next, we explore how our model derives the concavity of price impact from a market-making perspective. Finally, we demonstrate the link between the optimal execution problem and the macroscopic market making model.

\subsection{Comparison to the Avellaneda-Stoikov Model}
From a modelling perspective, a key difference between the macroscopic market making model and the Avellaneda-Stoikov model is that the probability concept in the latter is replaced by a proportion in the former. Indeed, the Avellaneda-Stoikov model focuses on the outcome of each individual order, while the macroscopic model examines the collective behaviour of numerous orders. By the law of large numbers, the proportion naturally approximates probability. Hence, we expect the solution of the macroscopic model to represent the `average' of the solutions to the Avellaneda-Stoikov model. We will provide a numerical example later. To facilitate further comparison, we restrict our analysis to the Avellaneda-Stoikov settings by following the general setup in \cite{gueant2017optimal} and \cite{cartea2015algorithmic}. Here, we omit the inventory threshold, and assume bid-ask symmetry for notational convenience.

The mid-price of the asset is modelled by $dS_t = \sigma dW_t$, where $(W_t)_t$ is a standard Brownian motion adapted to the filtration $(\mathcal{F}_t)_{t \in \mathbb{R}_+}$. The bid and ask quotes from the market maker are respectively denoted by $(S^b_t)_t$ and $(S^a_t)_t$. Denote by $(N^b_t)_t$ and $(N^a_t)_t$ the two point processes (independent of $(W_t)_t$) representing the number of transactions at the bid and at the ask, respectively. We assume that all orders have the same size $\Delta \in \mathbb{R}_+$ and thus the inventory process is given by $dQ_t = \Delta dN^b_t - \Delta dN^a_t$. The intensities of $(N^b_t)_t$ and $(N^a_t)_t$ are given by $\lambda^b \, \Lambda(\delta_t^b)$ and $\lambda^a \, \Lambda(\delta_t^a)$ respectively, where $\delta_t^b = S_t - S_t^b$ and $\delta_t^a = S_t^a - S_t$. Here, constants $\lambda^a$ and $\lambda^b$ represent the intensities of the orders, while function $\Lambda$ describes the probability that the order is captured by the considered market maker. The cash process $(X_t)_t$ then has the dynamics $dX_t = S_t^a \Delta dN_t^a - S_t^b \Delta dN_t^b$. The goal of the market maker is to
maximize the objective functional
\begin{equation*}
    \mathbb{E} \Big[ X_T + Q_T S_T - A \,(Q_T)^2 - \frac{1}{2} \sigma^2 \int^T_0 Q^2_t \, dt \Big],
\end{equation*}
where $A$ is a positive constant.

By the dynamic programming principles, let $u(t, x, q, s)$ be the value function. If we further apply the  ansatz $u(t, x, q, s) = x + q s + \theta(t, q)$, as shown in (3.8) of \cite{gueant2017optimal}, the HJB equation is simplified to 
\begin{equation}
\begin{aligned}
    0 = -\frac{\partial \theta}{\partial t} + \frac{1}{2} \sigma^2 q^2 &- \Delta \, \lambda^b \, \sup_{\delta^b} \bigg\{ \Lambda(\delta^b
    )\cdot \Big[\delta^b + \frac{\theta(t, q + \Delta) - \theta(t, q)}{\Delta} \Big] \bigg\}\\
    & - \Delta \, \lambda^a \, \sup_{\delta^a} \bigg\{ \Lambda(\delta^a
    )\cdot \Big[\delta^a + \frac{\theta(t, q - \Delta) - \theta(t, q)}{\Delta} \Big] \bigg\}
    \label{old_hjb}
\end{aligned}
\end{equation}
with the boundary condition $\theta(T, q) = -A q^2$. The associated optimal control reads
\begin{equation}
    \hat{\delta}^a_t = \delta^*\Big( \frac{\theta(t, q) - \theta(t, q - \Delta)}{\Delta}\Big) \text{\; and \;} \hat{\delta}^b_t = \delta^*\Big( \frac{\theta(t, q) - \theta(t, q + \Delta)}{\Delta}\Big).
    \label{old_optimal_feedback}
\end{equation}

Now, we `translate' the above discrete setting to our continuous version. Since $\lambda^a$ and $\lambda^b$ represent the intensities of orders with size $\Delta$, we let $a_t = \Delta \, \lambda^a$ and $b_t = \Delta \, \lambda^b$ for all $t$. Interpreting $\Lambda(\delta)$ as portion, the inventory and cash processes become
\begin{gather*}
    \tilde{Q}_t = q_0 - \int_0^t \Delta \lambda^a \cdot \Lambda(\delta_u^a) \, du + \int_0^t \Delta \lambda^b \cdot \Lambda(\delta^b_u) \, du, \\
    \tilde{X}_t = \int_0^t \Delta \lambda^a \cdot \Lambda(\delta_u^a) \, (S_u + \delta_u^a) \, du - \int_0^t \Delta \lambda^b \cdot \Lambda(\delta_u^b) \, (S_u - \delta_u^b) \, du.
\end{gather*}
The market maker seeks to maximize the same objective functional. Note that this is a specific case of our study in previous sections. Let $\tilde{u}(t, x, q, s)$ be the value function and its HJB equation reads
\begin{equation*}
\begin{aligned}
    0 = -\frac{\partial \tilde{u}}{\partial t} + \frac{1}{2} \sigma^2 q^2 - \frac{1}{2}\sigma^2 \frac{\partial^2 \tilde{u}}{\partial s^2} &- \Delta \lambda^b \sup_{\delta^b} \bigg\{ \Lambda(\delta^b
    ) \cdot \Big[(s - \delta^b) \, \frac{\partial \tilde{u}}{\partial x} + \frac{\partial \tilde{u}}{\partial q} \Big] \bigg\}\\
    &- \Delta \lambda^a \sup_{\delta^a} \bigg\{ \Lambda(\delta^a
    ) \cdot \Big[(s + \delta^a) \, \frac{\partial \tilde{u}}{\partial x} - \frac{\partial \tilde{u}}{\partial q} \Big] \bigg\}.
\end{aligned}
\end{equation*}
By the same ansatz $\tilde{u}(t, x, q, s) = x + q s + \tilde{\theta}(t, q)$, the equation above becomes
\begin{equation}
\begin{aligned}
    0 = -\frac{\partial \tilde{\theta}}{\partial t} + \frac{1}{2} \sigma^2 q^2 &- \Delta \lambda^b \sup_{\delta^b} \bigg\{ \Lambda(\delta^b
    ) \cdot \Big[\delta^b + \frac{\partial \tilde{\theta}}{\partial q} \Big] \bigg\}\\
    &- \Delta \lambda^a \sup_{\delta^a} \bigg\{ \Lambda(\delta^a
    ) \cdot \Big[\delta^a - \frac{\partial \tilde{\theta}}{\partial q} \Big] \bigg\}
\end{aligned}
\label{new_hjb}
\end{equation}
with the boundary condition $\tilde{\theta}(T, q) = -A q^2$. The corresponding optimal control follows
\begin{equation}
    \hat{\delta}^a_t = \delta^*\Big( \frac{\partial \tilde{\theta}}{\partial q}\Big) \text{\; and \;} \hat{\delta}^b_t = \delta^*\Big( -\frac{\partial \tilde{\theta}}{\partial q}\Big).
    \label{new_optimal_feedback}
\end{equation}
Comparing \eqref{old_optimal_feedback} with \eqref{new_optimal_feedback}, as well as \eqref{old_hjb} with \eqref{new_hjb}, the only differences are that $[\theta(t, q) - \theta(t, q - \Delta)]/\Delta$ is replaced by $\partial \tilde{\theta} / \partial q$ and $[\theta(t, q) - \theta(t, q + \Delta)]/\Delta$ is replaced by $-\partial \tilde{\theta} / \partial q$. This indicates that the macroscopic market making model serves as a `continuous version' of the Avellaneda-Stoikov model from a solution perspective. Intuitively, we expect the difference between $\theta$ and $\tilde{\theta}$ to be small when $\Delta$ is small.

For numerical illustrations, we set $\Lambda(x) = \exp(-\gamma \, x)$. In the first example, using the same set of parameters, we solve the market making problem in the Avellaneda-Stoikov model. We simulate $500$ optimal inventory paths and present them using a green heat map alongside their average path. Due to the bid-ask symmetry, the problem admits an explicit solution; see Section 10.2 in \cite{cartea2015algorithmic}. Subsequently, we address the problem using the macroscopic model, where the optimal inventory is determined by solving \eqref{new_hjb} via the finite difference method. The result, shown in Figure \ref{average path}, aligns with our conjecture that the macroscopic model represents an `average' of the traditional model from the modelling perspective.

Our second example examines the difference between $\theta$ in the Avellaneda-Stoikov model and $\tilde{\theta}$ in the macroscopic model. Keeping all other parameters fixed, we compute the solutions for $\theta$ and $\tilde{\theta}$ across three different values of $\Delta$. Figure \ref{convergence} illustrates the function values at $t = 0$, supporting our conjecture that the difference between $\theta$ and $\tilde{\theta}$ diminishes as $\Delta$ approaches smaller values. Additionally, both $\theta$ and $\tilde{\theta}$ decrease as $\Delta$ becomes smaller. This is because, when $\lambda^a$ and $\lambda^b$ are fixed, a smaller order size $\Delta$ results in reduced market trading volume. Consequently, there are fewer profits to exploit, and the market maker becomes more cautious about exposing inventory; see also Example \ref{section1_example}. Finally, the macroscopic solution, exhibiting greater values, appears to be slightly more `optimistic' than the Avellaneda-Stoikov solution. We partially attribute this to differences in inventory penalty. Consider an agent liquidating a one-unit long position, with the corresponding market order having an intensity of $1$. In this case, the running penalty remains constant at $1$ until the order has arrived. Consequently, the expected total running penalty equals the expectation of an exponential random variable, which is $1$. In the macroscopic model, the order is represented continuously with a rate of $1$. As a result, the inventory decreases linearly, with the total running penalty reduced to just $1/3$.

\begin{figure}
    \centering
    \includegraphics[width = 0.6 \linewidth]{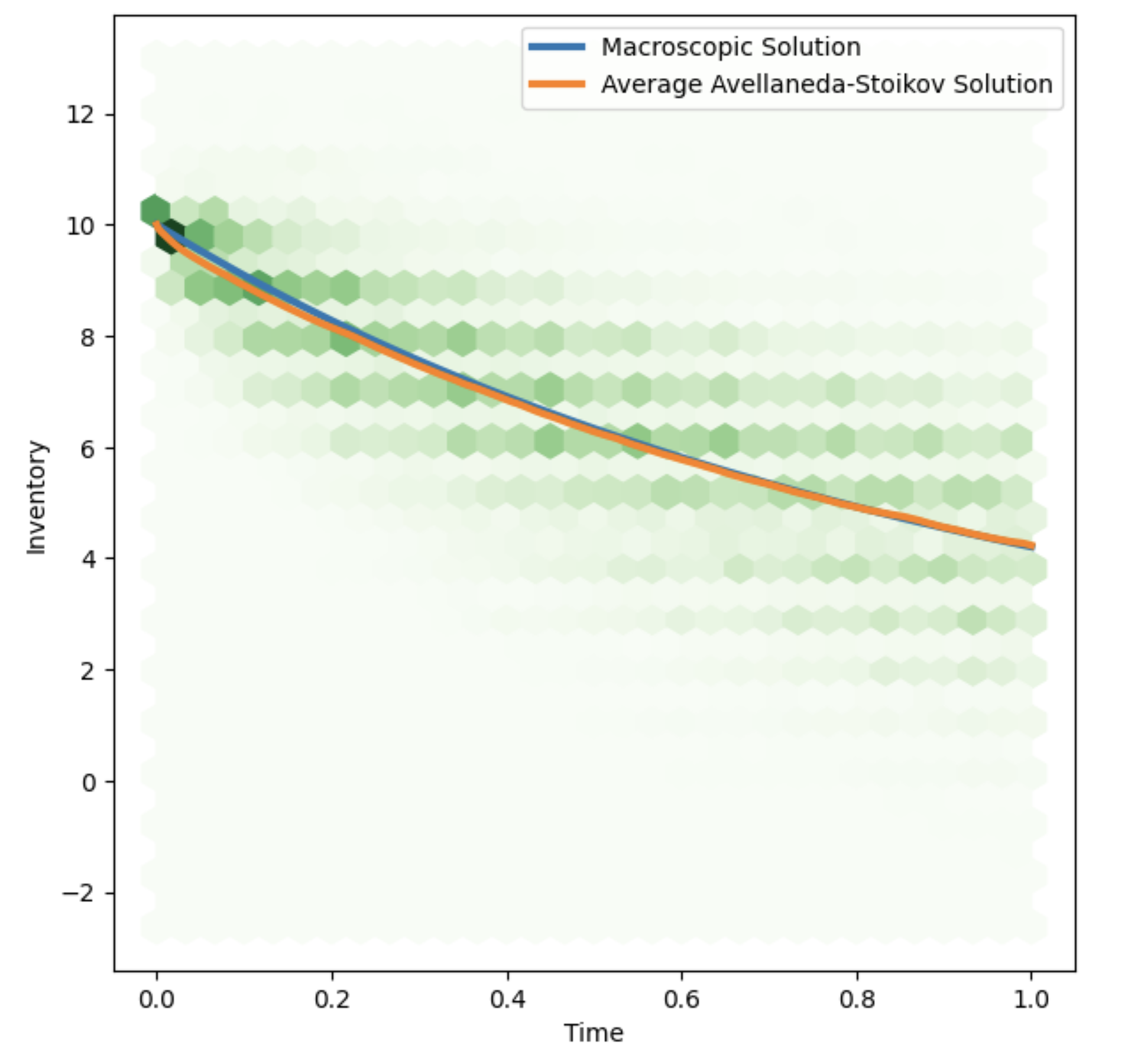}
    \caption{Average path and heat map of 500 Avellaneda-Stoikov solutions, alongside the macroscopic solution. Parameters: $T = \gamma = \Delta = 1$, $\lambda^a = \lambda^b = 10$, and $\sigma^2/2 = A = 0.05$, $q_0 = 10$.}
    \label{average path}
\end{figure}

\begin{figure}
    \centering
    \includegraphics[width = 0.8 \linewidth]{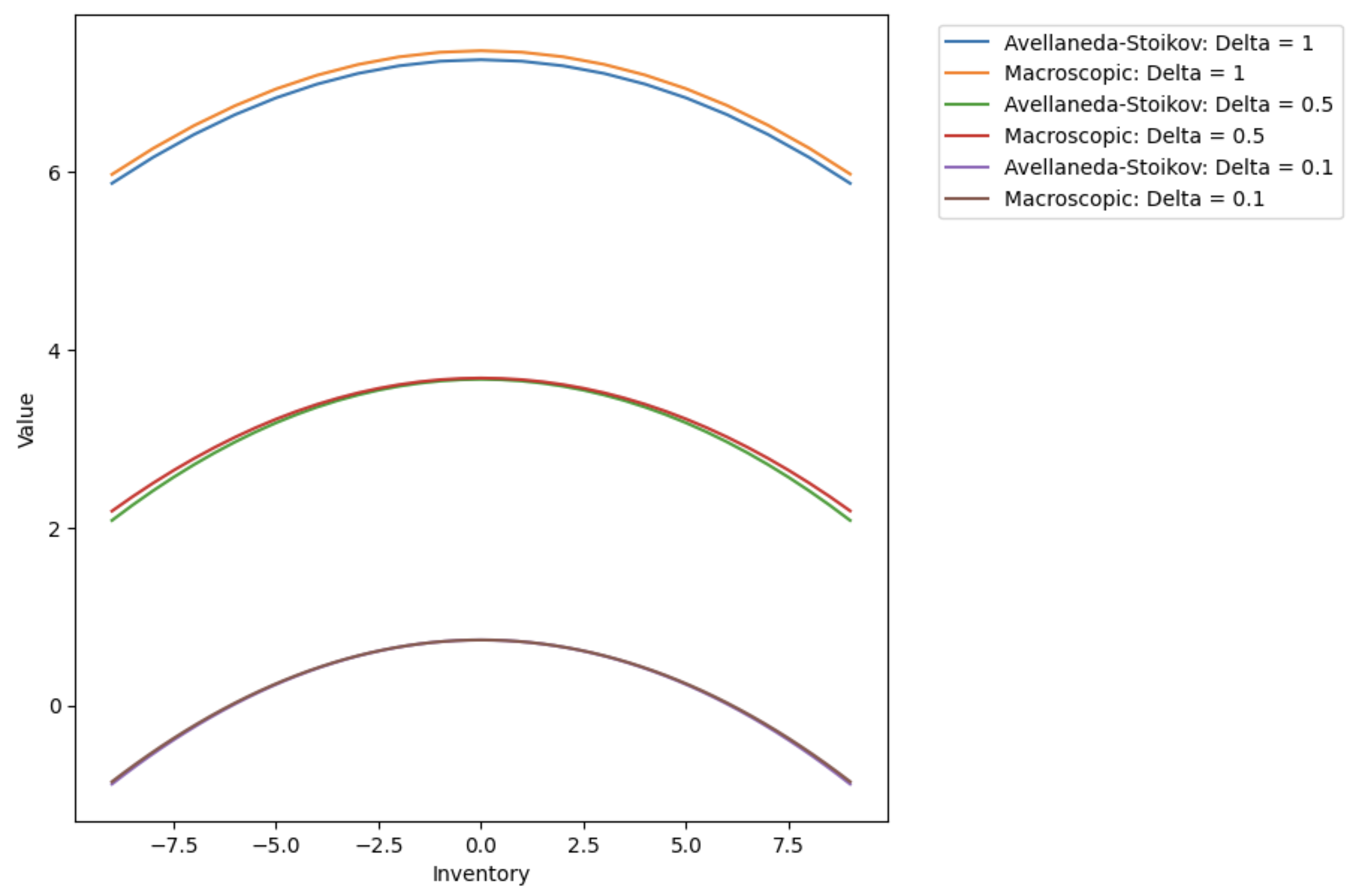}
    \caption{Values of $\theta(0, q)$ and $\tilde{\theta}(0, q)$ for $\Delta \in \{1, 0.5, 0.1\}$. Parameters: $T = \gamma = 1$, $\lambda^a = \lambda^b = 10$, and $\sigma^2/2 = A = 0.01$.}
    \label{convergence}
\end{figure}

\subsection{Concavity of price impact}
Price impact refers to the influence of trading activity on asset prices. When market participants execute trades, especially large orders, they can affect prices in both temporary and permanent ways. Temporary impact arises from short-term liquidity imbalances, whereas permanent impact reflects the market’s assimilation of potential information contained within the trades. The foundational work on price impact dates back to \cite{kyle1985continuous}, who introduced a framework where informed traders strategically execute orders while considering their price impact, and \cite{glosten1985bid}, who developed a model of asymmetric information in dealer markets. Empirical studies, such as \cite{toth2011anomalous}, have demonstrated that price impact generally follows a concave function of order size, often approximated by a square-root law. This finding has significantly influenced the development of modern trading algorithms and execution strategies. Recently, \cite{nadtochiy2022simple} provided a microstructural explanation for the asymptotic concavity of price impact.

Here, we restrict ourselves to the permanent price impact: the expected price change as a function of the volume of a meta-order (a large order that is split into smaller pieces and executed incrementally). Specifically, we analyse this impact from the perspective of a single market maker, where the price change is measured as the difference between the agent's ask price and the mid-price at the end; that is, $\delta_T^a$. In other words, we look at how the agent’s quote price responds to varying order volumes. In practice, with respect to liquidity provision, the best price is determined collectively by numerous market makers. Accordingly, the subsequent work of \cite{guo2024macroscopic} investigates the strategic interactions among multiple market makers.

We derive the concavity property within two simple scenarios of the macroscopic model. Let $A$ be a positive constant, and $b_t = \phi_t = 0$ for all $t$. In the first example, we assume the linear intensity $\Lambda(\delta) = \zeta - \gamma \delta$ and then Theorem \ref{unbound_linear_fbsde} yields the FBSDE
\begin{equation*}
\left\{
\begin{aligned}
\;& dQ_t  = - \zeta \, a_t \, dt \, /\,2 + \gamma\, a_t \,Y_t \, dt \, / \,2, \\
& dY_t = 0,\\
& Q_0 = q_0, \quad Y_T = -2A \, Q_T,
\end{aligned}
\right.
\end{equation*}
with the optimal control $\delta_t^{a,*}=\frac{\zeta}{2\gamma}+\frac{1}{2}Y_t$. Observing that $Y_t = -2A\, Q_T$ for all $t$, direct calculations give 
\begin{equation*}
    Q_T = \frac{1}{2} \, \frac{2q_0 - \zeta \int_0^T a_t \, dt}{1 + \gamma A \int_0^T a_t \, dt}.
\end{equation*}
The process $a$ is regarded as the execution procedure of a meta-order. If we consider $\delta_T^{a,*}$ (denoted as impact $f(x)$) as a function of $\int_0^T a_t \, dt$ (denoted as volume $x$), this function is expressed as follows:
\begin{equation*}
    f(x) = \frac{\zeta}{2\gamma} + \frac{1}{2} \, \frac{\zeta x - 2q_0}{A^{-1} + \gamma x}.
\end{equation*}
It is straightforward to check the concavity of $f$. If we further assume $q_0 \leq 0$, the explicit expression leads to several economic interpretations of the result. Parameters $\zeta$ and $\gamma$ measure the liquidity of the order book. As $\zeta$ increases or $\gamma$ decreases, the order book becomes more illiquid (or `thinner'), allowing the agent to enjoy a higher premium for each trade. Additionally, as the penalty parameter $A$ increases or $q_0$ decreases, the market maker---fearing inventory risk---demands a higher premium for sellings.

The second example is identical to the previous one, except that we adopt the exponential intensity $\Lambda(\delta) = \exp(-\gamma \, x)$. The $\delta^*$ function in Lemma \ref{inten_fun} has the expression $\delta^*(p) = \frac{1}{\gamma} + p$. Equation \eqref{non_lip_fbsde} yields the FBSDE
\begin{equation*}
\left\{
\begin{aligned}
\;& dQ_t  = - a_t \, \exp(-1-\gamma Y_t) \, dt, \\
& dY_t = 0,\\
& Q_0 = q_0, \quad Y_T = -2A \, Q_T,
\end{aligned}
\right.
\end{equation*}
with the optimal control $\delta_t^{a,*} = \frac{1}{\gamma} + Y_t$. We can obtain 
\begin{equation*}
    Q_T + \exp ( 2\gamma A \, Q_T - 1 ) \int_0^T a_t \, dt - q_0 = 0.
\end{equation*}
Since $\delta_t^{a,*} = \frac{1}{\gamma} -2A \, Q_T$, it suffices to prove the convexity of $Q_T$ with respect to $\int_0^T a_t \, dt$. We write $\int_0^T a_t \, dt$ as $x$, and $Q_T$ as $y$. The expression on the left hand side becomes
\begin{equation*}
    F(x, y) := y + x \, \exp \big( 2\gamma A \, y - 1 \big) - q_0.
\end{equation*}
We intend the study the relation $y = f(x)$ such that $F(x, f(x)) = 0$ for any $x \geq 0$. Note that function $f$ is differentiable due to the implicit function theorem. Through taking the derivative on both sides, we then have
\begin{equation*}
\begin{aligned}
    f'(x) &= (-1) \, \frac{\partial F} {\partial x} \big/ \frac{\partial F} {\partial y}\\
        &= (-1) \, \frac{\exp(2 \gamma A \, f(x) - 1)}{ 1 + 2 \gamma A \, x \, \exp(2 \gamma A \, f(x) - 1)},\\
   (-\frac{\partial F}{\partial y}) \, f''(x) &= \frac{\partial^2 F}{\partial x^2} + \frac{\partial^2 F}{\partial x \partial y} \, f'(x) + \Big[\frac{\partial^2 F}{\partial x \partial y} + \frac{\partial^2 F}{\partial y^2} \, f'(x) \Big]f'(x)\\
    & = 2\gamma A\exp(2\gamma A f(x)- 1) \, f'(x) \, \Big[2 - \frac{2\gamma A \, x \, \exp(2\gamma A f(x)- 1)}{1 + 2\gamma A \, x \, \exp(2\gamma A f(x)- 1)}\Big] < 0. 
\end{aligned}
\end{equation*}
Since $\partial F/\partial y > 0$, it implies $f''(x) > 0$. Hence, we have shown the convexity of $Q_T$ with respect to the order volume. In the case of $q_0 \leq 0$, following a similar discussion, we can still observe that the market maker demands a higher premium when $A$ increases (due to inventory risk) or $\gamma$ decreases (taking advantage of the illiquid order book).

For more general settings, we rely on numerical solutions. We set $\Lambda(x) = \exp(-\gamma x)$. To solve the FBSDE \eqref{non_lip_fbsde}, we define the decoupling field $Y_t = u(t, Q_t)$, where $u$ satisfies a PDE derived by matching the coefficients; see \cite{ma2012non}. The function $u$ is then computed using the finite difference method, and the solution is obtained by substituting $u$ back into \eqref{non_lip_fbsde}. In addition, we intend to examine the relationship between the impact $\delta_T^a$ and the order imbalance $\int_0^T a_t \, dt - \int_0^T b_t \, dt$. To achieve this, we begin by generating processes $a$ and $b$ with constant rates. Next, we scale $a$ such that $\int_0^T a_t \, dt - \int_0^T b_t \, dt = 5n$ and compute the corresponding impact, for $n \in \{0, \dots, 19\}$. Finally, we plot the resulting graph of order imbalance against $\delta_T^a$. Figure \ref{deter_impact} reveals a concave relationship between price impact and order imbalance. Intuitively, the price impact increases with the terminal penalty parameter.

Next, we generate order flow as sequences of independent and identically distributed random variables (treated as deterministic) and scale them such that $\int_0^T a_t \, dt - \int_0^T b_t \, dt = 5n$. We then compute the corresponding impact for $n \in \{0, \dots, 19\}$. This procedure is repeated $40$ times. The results, shown in Figure \ref{random_impact}, are significantly noisier but can be well-fitted by a power function with a degree less than $1$.

\begin{figure}
    \centering
    \includegraphics[width = 0.6\linewidth]{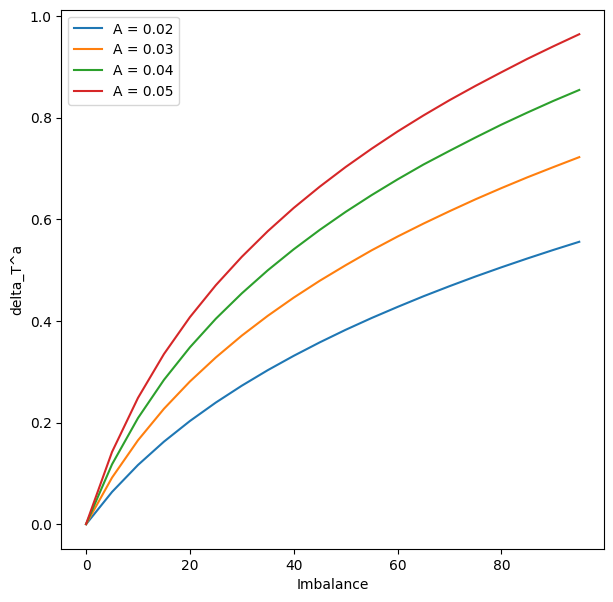}
    \caption{Price impacts corresponding to constant order flows at varying imbalance levels. Parameters: $\gamma = T = 1$ and $\phi = 0.02$, $q_0 = 0$.}
    \label{deter_impact}
\end{figure}

\begin{figure}
    \centering
    \includegraphics[width = 0.6\linewidth]{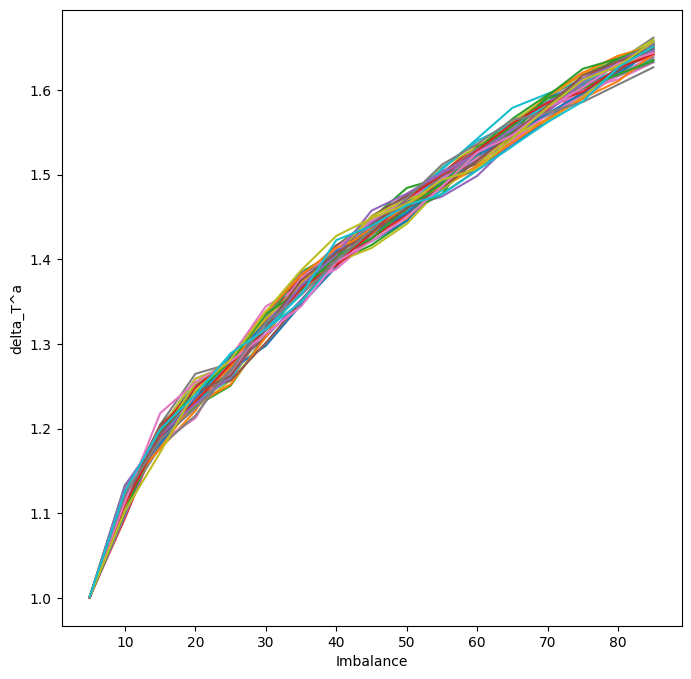}
    \caption{Forty realizations of price impacts corresponding to random order flows at varying imbalance levels. Parameters: $\gamma = T = 1$ and $\phi = A = 0.01$, $q_0 = 0$.}
    \label{random_impact}
\end{figure}

\subsection{Link with optimal execution}
Since order flows are modelled as rates in both frameworks, we can now establish a connection between macroscopic market making and optimal execution. Before doing so, we briefly review the setup of an optimal execution problem. A trader intends to trade $q_0^\mathfrak{e}$ amount of assets through market orders. Here, the superscript $\mathfrak{e}$ is used to denote that the agent is operating under an execution program. Given any trading strategy $(v_t)_t$, her inventory follows
\begin{equation*}
    Q_t^\mathfrak{e} = q_0^\mathfrak{e} + \int_0^t v_u \,du.
\end{equation*}
The trading activities bring an impact on the price, the market price $(D_t)_t$ has the dynamics
\begin{equation*}
    D_t = P_t + \int_0^t f(v_u) \, du.
\end{equation*}
Here, the process $(P_t)_t$ represents the reference price, while the function $f$ denotes a permanent impact function (e.g., a linear function as used in \cite{almgren2001optimal}). Considering liquidity and transaction costs, the transaction price $(\hat{S}_t)_t$ ---distinct from the market price---is given by
\begin{equation*}
    \hat{S}_t = D_t + g(v_t),
\end{equation*}
where function $g$ indicates a temporary impact function (e.g., a linear function as in \cite{almgren2001optimal}). The agent intends to maximize the objective functional
\begin{equation*}
    \mathbb{E}\Big[ -\int_0^T \hat{S}_t \, v_t \,dt - \int_0^T \phi^\mathfrak{e} \, (Q_t^\mathfrak{e})^2 \,dt - A^\mathfrak{e} \, (Q_T^\mathfrak{e})^2\Big]
\end{equation*}
that consists of the cash account and the inventory penalty. 

While many well-formulated models of the permanent impact function 
$f$ have been introduced, the function itself is chosen independently of the trading strategy. In other words, the impact function is exogenous. Our objective is to introduce an endogenous permanent impact from the perspective of liquidity provision. To achieve this, we additionally consider a market maker within the framework of macroscopic model. Furthermore, we assume that this market maker consistently provides the best price over the horizon $[0, T]$. Consequently, from the trader's perspective, the market price is $S_t^a$ when buying and $S_t^b$ when selling at time $t$. If we reserve the temporary impact function $g$, the terminal cash account is then given by
\begin{equation*}
    -\int_0^T \big(S_t + \delta_t^a + g(v_t) \big) \, v_t \, \mathbb{I}(v_t \geq 0) \, dt
    -\int_0^T \big(S_t - \delta_t^b + g(v_t) \big) \, v_t \, \mathbb{I}(v_t \leq 0) \, dt.
\end{equation*}
Since $(S_t)_t$ is typically a martingale, the objective functional reads
\begin{equation*}
\begin{aligned}
    \mathbb{E}\Big[ -\int_0^T \big(\delta_t^a + g(v_t) \big) \, v_t \, \mathbb{I}(v_t \geq 0) \, dt
    -\int_0^T \big(- \delta_t^b &+ g(v_t) \big) \, v_t \, \mathbb{I}(v_t < 0) \, dt\\
    &- \int_0^T \phi^\mathfrak{e} \, (Q_t^\mathfrak{e})^2 \,dt - A^\mathfrak{e} \, (Q_T^\mathfrak{e})^2\Big].
\end{aligned}
\end{equation*}
In this context, the overall procedure is:
\begin{itemize}
    \item[(1)] The trader propose a trading scheme $(v_t)_t$;\\
    \vspace{-0.2 cm}
    \item[(2)] The order flows are now given by
    \begin{equation*}
        a_t = \tilde{a}_t + v_t \, \mathbb{I}(v_t \geq 0) \text{ \; and \; } b_t = \tilde{b}_t - v_t \, \mathbb{I}(v_t < 0),
    \end{equation*}
    where $(\tilde{a}_t)_t$ and $(\tilde{b}_t)_t$ represent the trading activities from the others. Order flows $a$, $b$, and penalty parameters serve as the input for the macroscopic market making problem. The output consists of the optimal quoting strategies $(\delta_t^a)_t$ and $(\delta_t^b)_t$ from the market maker.\\
    \vspace{-0.2cm}

    \item[(3)] The strategy $(v_t)_t$ is thus evaluated by the objective functional.
\end{itemize}
The trader seeks to determine the optimal trading scheme. The key novelty of this optimisation problem is that price impact arises from the strategic interactions between the trader and the market maker, rather than being dictated by a pre-specified function. In fact, the above procedure describes a stochastic differential game between the trader and the market maker, the mathematical analysis of which will be presented in our subsequent work. Moreover, the assumption that the market maker under consideration always provides the best price may be overly strong. We refer the reader to \cite{guo2024macroscopic} for an analysis of the game between multiple market makers.

We present three simple trading schemes and numerically evaluate their performance. Consider an optimal liquidation problem with $q_0^\mathfrak{e} = 40$. For simplicity, we set $\phi^\mathfrak{e} = 0$ and $g \equiv 0$. Additionally, the trader submits only sell market orders with the constraint that $Q_T^\mathfrak{e} = 0$. Setting $\Lambda(x) = \exp(-\gamma x)$, the objective functional simplifies to $-\mathbb{E}[\int_0^T v_t \, Y_t \, dt]$ after removing a constant term. In each experiment, we randomly generate order flows $a, b$ such that $\int_0^T a_t \, dt - \int_0^T b_t \, dt = 30$. They are again treated as deterministic order flow. The trading schemes under consideration are:
\begin{itemize}
    \item Time-Weighted Average Price (TWAP): the total volume $q_0^\mathfrak{e}$ is evenly distributed across time;\\
    \vspace{-0.1cm}

    \item Volume-Weighted Average Price (VWAP): The trading volume at time $t$ is distributed in proportion to the market volume or activity, represented here by the selling rate $b_t$;\\
    \vspace{-0.1cm}
    
    \item Exploitative strategy: Since the buying volume is known to be higher than the selling volume, we anticipate that the imbalance during the first half will push the price up. Therefore, trades can be executed evenly during the second half of the time horizon.
\end{itemize}
Figure \ref{trade} presents the histogram of outcomes from $100$ experiments. In the given scenario, the TWAP and VWAP strategies show similar performance, while the exploitative strategy demonstrates a clear advantage.

\begin{figure}
    \centering
    \includegraphics[width = 0.7\linewidth]{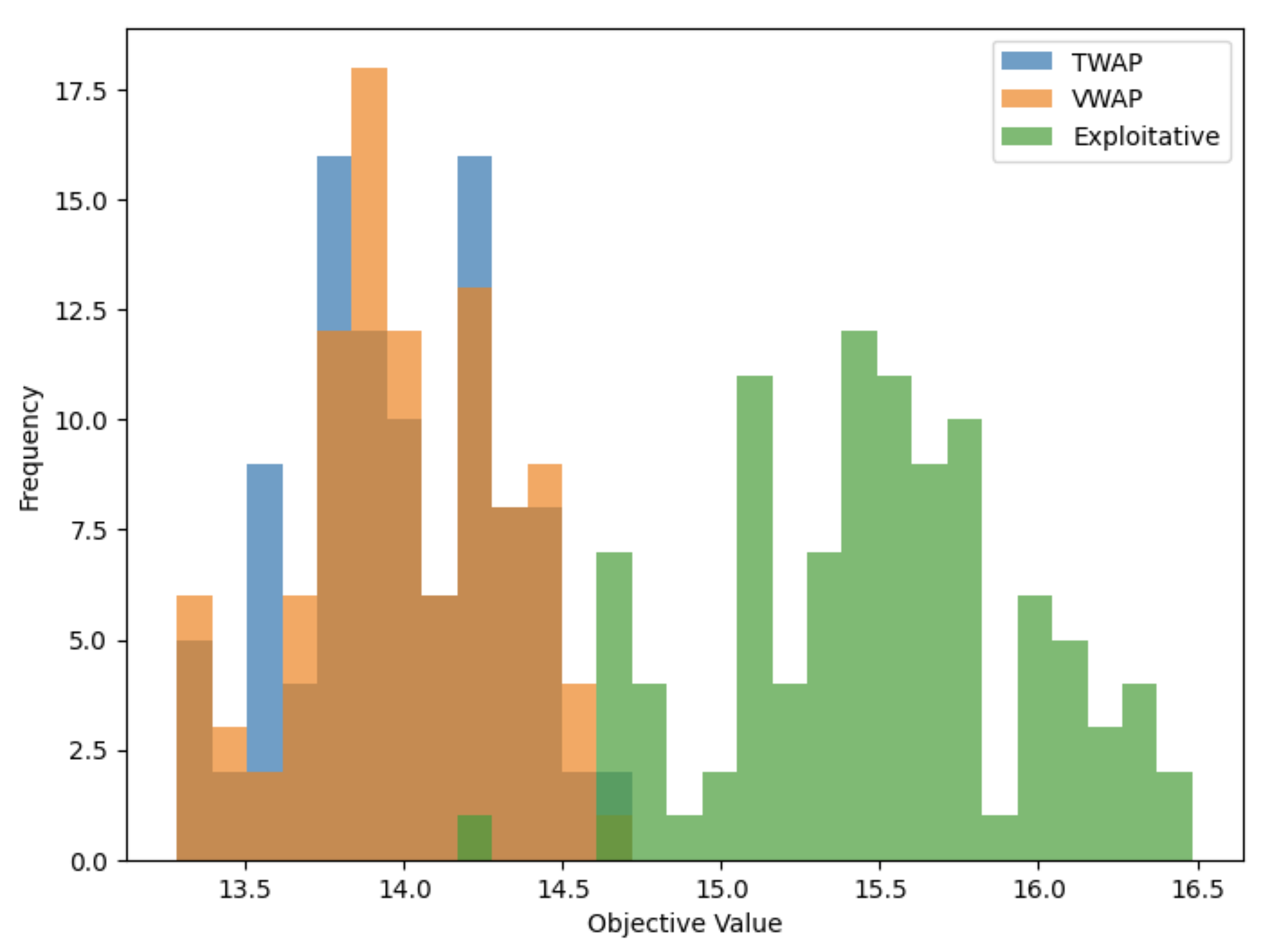}
    \caption{Histogram of objective function values for the three trading strategies. Parameters: $\gamma = T = 1$, $\phi = A = 0.04$, $q_0 = 0$.}
    \label{trade}
\end{figure}

\section{Appendix: Stochastic Maximum Principle} \label{section_5}

\noindent This section is devoted to the stochastic maximum principle for the general intensity function. We follow the argument in \cite{carmona2016lectures}, and start with the \textit{necessary} condition. Fix an admissible control $\boldsymbol{\delta}\in \mathbb{A}\times\mathbb{A}$ and denote by $Q = Q^{\boldsymbol{\delta}}\in\mathbb{S}^2$ the corresponding controlled inventory, which is also uniformly bounded. Next, consider $\boldsymbol{\beta}\in \mathbb{A}\times\mathbb{A}$---that is uniformly bounded and satisfies that $\boldsymbol{\delta}+\boldsymbol{\beta}\in\mathbb{A}\times\mathbb{A}$---as the direction in which we are computing the
G\^ateaux derivative of $J$. For each $\epsilon > 0$ small enough,
we consider the admissible control $\boldsymbol{\delta}^\epsilon\in \mathbb{A}\times\mathbb{A}$ defined by $\boldsymbol{\delta}^\epsilon_t = \boldsymbol{\delta}_t + \epsilon\,\boldsymbol{\beta}_t$, and the corresponding controlled state $Q^\epsilon:= Q^{\boldsymbol{\delta}^\epsilon}\in\mathbb{S}^2$. Define $V$ as the solution of the equation
\begin{equation*}
    dV_t=\partial_{\boldsymbol{\delta}}\mu(t, \boldsymbol{\delta}_t)\cdot\boldsymbol{\beta}_t\,dt,
\end{equation*}
with the initial condition $V_0 = 0$, where
\begin{equation*}
    \mu(t,\boldsymbol{\delta}):=-a_t\,\Lambda^a(\delta_t^a)+b_t\,\Lambda^b(\delta_t^b).
\end{equation*}
It is clear that $V$ is uniformly bounded in the finite time horizon. We then look at the G\^ateaux differentiability of the functional $J$.

\begin{lemma}
The functional $\boldsymbol{\delta}\hookrightarrow J(\boldsymbol{\delta})$ is G\^ateaux differentiable and the derivative reads
\begin{equation}
\begin{aligned}
    \frac{d}{d\epsilon}J(\boldsymbol{\delta}+\epsilon\boldsymbol{\beta})\big|_{\epsilon=0}:&=\lim_{\epsilon\searrow0}\frac{1}{\epsilon}\,\big[J(\boldsymbol{\delta}+\epsilon\boldsymbol{\beta})-J(\boldsymbol{\delta})\big]\\
    &=\mathbb{E}\Big[\int_0^T\big(\partial_qf(t,Q_t,\boldsymbol{\delta}_t)\,V_t+\partial_{\boldsymbol{\delta}}f(t,Q_t,\boldsymbol{\delta}_t)\cdot\beta_t\big)\,dt-2\,A\,V_T\,Q_T\Big],
    \label{gen_gat_deriv}
\end{aligned}
\end{equation}
where $f$ is the running payoff defined by
\begin{equation*}
    f(t, Q_t, \boldsymbol{\delta}_t)=a_t\,\delta_t^a\,\Lambda^a(\delta_t^a)+b_t\,\delta_t^b\,\Lambda^b(\delta_t^b)-\phi_t\,Q_t^2.
\end{equation*}
\end{lemma}

\begin{proof}
Regarding the ask side, we compute that
\begin{equation*}
\begin{aligned}
\frac{1}{\epsilon}\,&\mathbb{E}\Big[\int_0^Ta_t\,(\delta_t^a+\epsilon\,\beta_t^a)\,\Lambda^a(\delta_t^a+\epsilon\,\beta_t^a)\,dt-\int_0^T a_t\,\delta_t^a\,\,\Lambda^a(\delta_t^a)\,dt\Big]\\
&=\mathbb{E}\Big[\int_0^Ta_t\,\delta_t^a\,\frac{1}{\epsilon}\,\big(\Lambda^a(\delta_t^a+\epsilon\,\beta_t^a)-\Lambda^a(\delta_t^a)\big)\,dt\Big]+\mathbb{E}\Big[\int_0^Ta_t\,\beta_t^a\,\Lambda^a(\delta_t^a+\epsilon\,\beta_t^a)\,dt\Big].
\end{aligned}
\end{equation*}
To perform the limiting procedure, note that
\begin{equation*}
\begin{aligned}
    a_t\,\delta_t^a\,\frac{1}{\epsilon}\,\big(\Lambda^a(\delta_t^a+\epsilon\,\beta_t^a)-\Lambda^a(\delta_t^a)\big)&\leq C\,a_t\,|\delta_t^a|,\\
    a_t\,\beta_t^a\,\Lambda^a(\delta_t^a+\epsilon\,\beta_t^a) &\leq C\, a_t,
\end{aligned}
\end{equation*}
the right hand sides of which are $d\mathbb{P}\times dt$-integrable on the time horizon. Therefore, it follows by the dominated convergence theorem that
\begin{equation*}
\begin{aligned}
    \lim_{\epsilon\to 0}\frac{1}{\epsilon}\,\mathbb{E}\Big[\int_0^Ta_t\,(\delta_t^a+\epsilon\,\beta_t^a)\,&\Lambda^a(\delta_t^a+\epsilon\,\beta_t^a)\,dt-\int_0^T a_t\,\delta_t^a\,\,\Lambda^a(\delta_t^a)\,dt\Big]\\
    &=\mathbb{E}\Big[\int_0^Ta_t\,\delta_t^a\,\big(\Lambda^{a}(\delta_t^a)\big)'\,\beta_t^a\,dt\Big]+\mathbb{E}\Big[\int_0^Ta_t\,\beta_t^a\,\Lambda^a(\delta_t^a)\,dt\Big]\\
    &=\mathbb{E}\Big[\int_0^T\partial_{\delta^a}f(t,Q_t,\boldsymbol{\delta}_t)\,\beta_t^a\,dt\Big].
\end{aligned}
\end{equation*}
The bid side can be computed in the same way:
\begin{equation*}
\begin{aligned}
    \lim_{\epsilon\to 0}\frac{1}{\epsilon}\,\mathbb{E}\Big[\int_0^Tb_t\,(\delta_t^b+\epsilon\,\beta_t^b)\,&\Lambda^b(\delta_t^b+\epsilon\,\beta_t^b)\,dt-\int_0^T b_t\,\delta_t^b\,\,\Lambda^b(\delta_t^b)\,dt\Big]\\
    &=\mathbb{E}\Big[\int_0^T\partial_{\delta^b}f(t,Q_t,\boldsymbol{\delta}_t)\,\beta_t^b\,dt\Big].
\end{aligned}
\end{equation*}
With respect to the running penalty, we further calculate that
\begin{equation*}
\begin{aligned}
    \frac{1}{\epsilon}\,\mathbb{E}\Big[\int_0^T &\phi_t\,(Q_t^\epsilon)^2\,dt-\int_0^T \phi_t\,(Q_t)^2\,dt\Big]\\
    &=\mathbb{E}\Big[\int_0^T \phi_t\,(Q_t^\epsilon+Q_t)\,\Big(-\int_0^ta_u\,\frac{1}{\epsilon}\,\big[\Lambda^a(\delta_u^a+\epsilon\,\beta_u^a)-\Lambda^a(\delta_u^a)\big]\,du\\
    &\hspace{4cm}+\int_0^tb_u\,\frac{1}{\epsilon}\,\big[\Lambda^b(\delta_u^b+\epsilon\,\beta_u^b)-\Lambda^b(\delta_u^b)\big]\,du\Big)\,dt\Big].
\end{aligned}
\end{equation*}
Being aware of the boundedness of $\phi$, $Q$, $Q^\epsilon$, $a$,  $b$ as well as the one of
\begin{equation*}
 \frac{1}{\epsilon}\,\big[\Lambda^a(\delta_u^a+\epsilon\,\beta_u^a)-\Lambda^a(\delta_u^a)\big],
 \text{\quad and \quad} \frac{1}{\epsilon}\,\big[\Lambda^b(\delta_u^b+\epsilon\,\beta_u^b)-\Lambda^b(\delta_u^b)\big],
\end{equation*}
we can conclude
\begin{equation*}
    \lim_{\epsilon\to 0}\,\frac{1}{\epsilon}\,\mathbb{E}\Big[\int_0^T \phi_t\,(Q_t^\epsilon)^2\,dt-\int_0^T \phi_t\,(Q_t)^2\,dt\Big]=2\,\mathbb{E}\Big[\int_0^T \phi_t\,Q_t\,V_t\,dt\Big].
\end{equation*}
The terminal penalty part can be justified in the same way and the proof is complete. 
\end{proof}

\noindent Let $(Y,Z)$ be the adjoint processes associated with $\boldsymbol{\delta}$, i.e., processes $(Y,Z)$ solve the BSDE
\begin{equation*}
    dY_t=2\,\phi_t\,Q_t\,dt+Z_t\,dW_t, \text{\quad and \quad} Y_T=-2A\,Q_T.
\end{equation*}
The well-posedness of such equation is guaranteed due to the boundedness of $A$, $\phi$ and $Q$. The duality relation provides an expression of the G\^ateaux derivative of the cost functional in terms of the Hamiltonian of the system.

\begin{lemma}
The duality relation is given by
\begin{equation*}
    \mathbb{E}[Y_T\,V_T]=\mathbb{E}\Big[\int_0^T\big(\partial_{\boldsymbol{\delta}}\mu(t,\boldsymbol{\delta}_t)\cdot\boldsymbol{\beta}_t\,Y_t-\partial_qf(t,Q_t,\boldsymbol{\delta}_t)\,V_t\big)\,dt\Big].
\end{equation*}
The duality further implies
\begin{equation*}
    \frac{d}{d\epsilon}\,J(\boldsymbol{\delta}+\epsilon\boldsymbol{\beta})\big|_{\epsilon=0}=\mathbb{E}\Big[\int_0^T\partial_{\boldsymbol{\delta}}\mathscr{H}(t,Q_t,Y_t,\boldsymbol{\delta}_t)\cdot\boldsymbol{\beta}_t\,dt\Big],
\end{equation*}
where $\mathscr{H}$ denotes the Hamiltonian as follows
\begin{equation*}
    \mathscr{H}(t, q, y, \boldsymbol{\delta}):=\big(-a_t\,\Lambda^a(\delta^a)+b_t\,\Lambda^b(\delta^b)\big)\,y+a_t\,\delta^a\,\Lambda^a(\delta^a)+b_t\,\delta^b\,\Lambda^b(\delta^b)-\phi_t\,(q)^2.
\end{equation*}
\end{lemma}

\begin{proof}
The integration by parts gives
\begin{equation*}
\begin{aligned}
    Y_T\,V_T&=Y_0\,V_0+\int_0^T Y_t\,dV_t+\int_0^T V_t\,dY_t\\
    &=\int_0^TY_t\, \partial_{\boldsymbol{\delta}}b(t, \boldsymbol{\delta}_t)\cdot\boldsymbol{\beta}_t\,dt+\int_0^T 2\,\phi_t\,V_t\,Q_t\,dt+\int_0^T V_t\,Z_t\,dW_t.
\end{aligned}
\end{equation*}
The duality is verified after taking the expectation, since the last term is a square integrable martingale. Using this relation, we can further compute
\begin{equation*}
\begin{aligned}
    \frac{d}{d\epsilon}&J(\boldsymbol{\delta}+\epsilon\boldsymbol{\beta})\big|_{\epsilon=0}\\
    &=\mathbb{E}\Big[\int_0^T\big(\partial_qf(t,Q_t,\boldsymbol{\delta}_t)\,V_t+\partial_{\boldsymbol{\delta}}f(t,Q_t,\boldsymbol{\delta}_t)\cdot\boldsymbol{\beta}_t\big)\,dt+V_T\,Y_T\Big]\\
    &=\mathbb{E}\Big[\int_0^TY_t\, \partial_{\boldsymbol{\delta}}b(t, \boldsymbol{\delta}_t)\cdot\boldsymbol{\beta}_t\,dt+\int_0^T\partial_{\boldsymbol{\delta}}f(t,Q_t,\boldsymbol{\delta}_t)\cdot\boldsymbol{\beta}_t\,dt\Big]\\
    &=\mathbb{E}\Big[\int_0^T\partial_{\boldsymbol{\delta}}\mathscr{H}(t,Q_t,Y_t,\boldsymbol{\delta}_t)\cdot\boldsymbol{\beta}_t\,dt\Big].
\end{aligned}
\end{equation*}

\end{proof}

\noindent With these preliminaries, we finally reach the necessary condition for optimality.

\begin{theorem}
If the admissible control $\boldsymbol{\delta}$ is optimal, $Q$ is the associated controlled inventory, and $(Y, Z)$ are the associated adjoint
processes, then
\begin{equation*}
    \mathscr{H}(t,Q_t,Y_t,\boldsymbol{\delta}_t)\geq\mathscr{H}(t,Q_t,Y_t,\beta), \qquad \text{a.e. in } t\in[0,T],\; \mathbb{P} \text{-} a.s.,
\end{equation*}
for any $\beta\in[-\xi, \xi]\times[-\xi, \xi]$.
\end{theorem}

\begin{proof}
Due to the convexity of the admissible control, given any admissible and bounded $\boldsymbol{\beta}\in\mathbb{A}\times\mathbb{A}$, we can choose the perturbation $\boldsymbol{\delta}^\epsilon=\boldsymbol{\delta}+\epsilon\,(\boldsymbol{\beta}-\boldsymbol{\delta})$ which is still admissible. Since $\boldsymbol{\delta}$ is optimal, we have the inequality
\begin{equation*}
    \frac{d}{d\epsilon}\,J(\boldsymbol{\delta}+\epsilon\,(\boldsymbol{\beta}-\boldsymbol{\delta}))\big|_{\epsilon=0}=\mathbb{E}\Big[\int_0^T\partial_{\boldsymbol{\delta}}\mathscr{H}(t,Q_t,Y_t,\boldsymbol{\delta}_t)\cdot(\boldsymbol{\beta}_t-\boldsymbol{\delta}_t)\,dt\Big]\leq0.
\end{equation*}
From this we can see 
\begin{equation}
    \partial_{\boldsymbol{\delta}}\mathscr{H}(t,Q_t,Y_t,\boldsymbol{\delta}_t)\cdot(\boldsymbol{\beta}-\boldsymbol{\delta}_t)\leq 0, \qquad \text{a.e. in } t\in[0,T],\; \mathbb{P}-a.s.,
    \label{smp_gener}
\end{equation}
for all $\beta\in[-\xi, \xi]\times[-\xi, \xi]$. Regarding the condition \eqref{smp_gener}, it can only happen at time $t$ if one of the following three cases holds:
\begin{equation*}
\begin{aligned}
    \delta_t^j\in(-\xi, +\xi) &\text{\quad with \quad} 
    \partial_{\boldsymbol{\delta}^j}\mathscr{H}(t,Q_t,Y_t,\boldsymbol{\delta}_t)=0;\\
    \delta_t^j=-\xi\leq \beta^j &\text{\quad with \quad} \partial_{\boldsymbol{\delta}^j}\mathscr{H}(t,Q_t,Y_t,\boldsymbol{\delta}_t)\leq0;\\
    \delta_t^j=+\xi\geq \beta^j &\text{\quad with \quad} \partial_{\boldsymbol{\delta}^j}\mathscr{H}(t,Q_t,Y_t,\boldsymbol{\delta}_t)\geq0,
\end{aligned}
\end{equation*}
where $j\in\{a, b\}$ representing the bid-ask side. Take the ask side for example. If
\begin{equation*}
    \partial_{\boldsymbol{\delta}^a}\mathscr{H}(t,Q_t,Y_t,\boldsymbol{\delta}_t)=-a_t\,(\Lambda^a(\delta_t^a))'\,\Big[Y_t-\delta_t^a-\frac{\Lambda^a(\delta_t^a)}{(\Lambda^a(\delta_t^a))'}\Big]=0,
\end{equation*}
the monotonicity of $\delta+\frac{\Lambda^a(\delta)}{(\Lambda^a(\delta))'}$ (see Assumption \ref{inten_assu}) implies that $\delta_t^a$ maximizes the corresponding part of the Hamiltonian. On the other hand, we look at the case $\delta_t^a=-\xi$ and  $\partial_{\boldsymbol{\delta}^a}\mathscr{H}(t,Q_t,Y_t,\boldsymbol{\delta}_t)\leq0$. Since  $\mathscr{H}$ is first increasing and then decreasing, such condition indicates that $\delta_t^a=-\xi$ is the maximizer of the Hamiltonian on the interval $[-\xi, \xi]$. While the final case can be discussed in the same way, all three cases imply that the optimal control $\boldsymbol{\delta}$ maximizes the Hamiltonian along the optimal paths, as stated by the theorem. 
\end{proof}

The \textit{sufficient} condition is more straightforward. We summarize the result in the following theorem.

\begin{theorem}
Let $\boldsymbol{\delta}$ be an admissible control, 
$Q = Q^{\boldsymbol{\delta}}$ be the corresponding controlled inventory, and $(Y, Z)$ be the adjoint processes. If it holds $\mathbb{P}$-a.s. that
\begin{equation}
    \mathscr{H}(t,Q_t,Y_t,\boldsymbol{\delta}_t)=\sup_{\boldsymbol{\alpha}\in [-\xi,\xi]^2}\mathscr{H}(t,Q_t,Y_t,\boldsymbol{\alpha}),  \qquad \text{a.e. in } t\in[0,T],
    \nonumber
\end{equation}
then $\boldsymbol{\delta}$ is the optimal control, that is to say, $J(\boldsymbol{\delta})=\sup_{\boldsymbol{\alpha}\in \mathbb{A}\times\mathbb{A}}J(\boldsymbol{\alpha})$.
\end{theorem}

\begin{proof}
Since the inventory process---that is associated with any admissible strategy---is bounded, the adjoint BSDE
\begin{equation}
    dY_t=2\phi_t\, Q_t\,dt+Z_t\,dW_t, \quad \text{with} \quad Y_T=-2A\,Q_T,
    \nonumber
\end{equation}
admits a unique solution $(Y,Z)\in\mathbb{S}^2\times\mathbb{H}^2$. Let $\boldsymbol{\alpha}\in \mathbb{A}\times\mathbb{A}$ be a generic admissible strategy and $Q'$ be the associated state process. Due to the concaveness of the terminal penalty, we have
\begin{equation}
    \begin{aligned}
        \mathbb{E}[-A\,&(Q_T)^2 + A\,(Q'_T)^2]\\
        &\geq \mathbb{E}\big[-2 A\, Q_T\,(Q_T-Q'_T)\big]\\
        &=\mathbb{E}\big[Y_T\,(Q_T-Q'_T)\big]\\
        &=\mathbb{E}\Big[\int_0^T2\phi_t\,Q_t\cdot(Q_t-Q'_t)\,dt+\int_0^T(Q_t-Q'_t)\cdot Z_t\,dW_t+\int_0^TY_t\,d(Q_t-Q'_t)\Big]\\
        &=\mathbb{E}\Big[\int_0^T2\phi_t\,Q_t\cdot(Q_t-Q'_t)\,dt+\int_0^TY_t\,d(Q_t-Q'_t)\Big],
    \end{aligned}
    \nonumber
\end{equation}
where we have used the integration by parts and the fact that the stochastic integral is a bona fide martingale. By the definition of the Hamiltonian, one can also obtain
\begin{equation}
    \begin{aligned}
    \mathbb{E}\int_0^T\Big[a_t\,&\delta_t^a\,\Lambda^a(\delta_t^a)+b_t\,\delta_t^b\,\Lambda^b(\delta_t^b)\\
    &\hspace{2cm}-\phi_t\,(Q_t)^2-\big(a_t\,\alpha_t^a\,\Lambda^a(\alpha_t^a)+b_t\,\alpha_t^b\,\Lambda^a(\alpha_t^a)-\phi_t\,(Q'_t)^2\big)\Big]\,dt\\
    &=\mathbb{E}\int_0^T\big[\mathscr{H}(t,Q_t,Y_t,\boldsymbol{\delta}_t)-\mathscr{H}(t,Q'_t,Y_t,\boldsymbol{\alpha}_t)\big]\,dt-\mathbb{E}\int_0^TY_t\,d(Q_t-Q'_t).
    \end{aligned}
    \nonumber
\end{equation}
Combining two results above, finally we can get
\begin{equation}
\begin{aligned}
    J&(\boldsymbol{\delta})-J(\boldsymbol{\alpha})\\
    &=\mathbb{E}\int_0^T\Big[a_t\,\delta_t^a\,\Lambda^a(\delta_t^a)+b_t\,\delta_t^b\,\Lambda^b(\delta_t^b)-\phi_t\,(Q_t)^2\\
    &\hspace{1cm}-\big(a_t\,\alpha_t^a\,\Lambda^a(\alpha_t^a)+b_t\,\alpha_t^b\,\Lambda^b(\alpha_t^b)-\phi_t\,(Q'_t)^2\big)\Big]\,dt+\mathbb{E}\big[-2A\,(Q_T)^2 +2A\, (Q'_T)^2\big]\\
    &\geq\mathbb{E}\int_0^T\big[\mathscr{H}(t,Q_t,Y_t,\boldsymbol{\delta}_t)-\mathscr{H}(t,Q'_t,Y_t,\boldsymbol{\alpha}_t)+2\phi_t\,Q_t\cdot(Q_t-Q'_t)\big]\,dt\\
    &\geq\mathbb{E}\int_0^T\big[\mathscr{H}(t,Q_t,Y_t,\boldsymbol{\alpha}_t)-\mathscr{H}(t,Q'_t,Y_t,\boldsymbol{\alpha}_t)+2\phi_t\,Q_t\cdot(Q_t-Q'_t)\big]\,dt\\
    &=\mathbb{E}\int_0^T\phi_t\,(Q_t-Q'_t)^2\,dt\\
    &\geq 0.
\end{aligned}
    \nonumber
\end{equation} 
\end{proof}

\sloppy
\printbibliography

\end{document}